\documentclass[aps,pra,10pt,twocolumn,longbibliography,floatfix]{revtex4-2}

\usepackage[T1]{fontenc}
\usepackage[utf8]{inputenc}
\usepackage{lmodern}
\usepackage{booktabs,microtype}
\usepackage{amsmath,amssymb,amsthm,mathtools}
\usepackage{bm}
\ExplSyntaxOn
\cs_new:Npn \SupplementLetter #1 { \int_to_Alph:n {#1} }
\ExplSyntaxOff
\usepackage{graphicx}
\usepackage{pgfplots}
\pgfplotsset{compat=1.18}
\usepackage[hidelinks]{hyperref}
\hypersetup{pdftitle={Securing quantum error correction against misleading advice from AI agents},
 pdfauthor={A. Baris Ozguler},pdfsubject={Evidence-based protection of quantum error-correction recovery updates under unreliable or malicious advice}}
\usepackage[nameinlink,noabbrev]{cleveref}
\crefname{appendix}{Supplemental Material, Sec.}{Supplemental Material, Secs.}
\Crefname{appendix}{Supplemental Material, Sec.}{Supplemental Material, Secs.}

\newtheorem{theorem}{Theorem}
\newtheorem{lemma}[theorem]{Lemma}
\newtheorem{proposition}[theorem]{Proposition}

\theoremstyle{remark}
\newtheorem{remark}[theorem]{Remark}

\newcommand{\Ftwo}{\mathbb F_2}
\newcommand{\id}{\mathrm{id}}
\newcommand{\supp}{\operatorname{supp}}
\newcommand{\row}{\operatorname{row}}
\newcommand{\Tr}{\operatorname{Tr}}
\newcommand{\TV}{\operatorname{TV}}
\newcommand{\Ad}{\operatorname{Ad}}
\newcommand{\ord}{\operatorname{ord}}
\newcommand{\sgn}{\operatorname{sgn}}
\newcommand{\cH}{\overline H}
\newcommand{\Xx}{\overline X_x}
\newcommand{\Xy}{\overline X_y}
\newcommand{\Zx}{\overline Z_x}
\newcommand{\Zy}{\overline Z_y}

\newcommand{\risk}{\mathcal R}
\newcommand{\obs}{\mathcal O}

\newcommand{\FreshGain}{97.1\%}
\newcommand{\DelayedGain}{90.8\%}

\newcommand{\IncumbentExample}{0.1323553}
\newcommand{\CandidateExample}{0.0007456}
\newcommand{\GRURisk}{0.088018}
\newcommand{\BayesRisk}{0.089296}
\newcommand{\RetainRisk}{0.225490}
\newcommand{\BayesShots}{6581}
\newcommand{\ShotReduction}{19.7\%}

\begin{document}
\raggedbottom
\title{Securing quantum error correction against misleading advice from AI agents}
\author{A. Barış Özgüler}
\email{barisozguler@geeqchiqtech.com}
\affiliation{Haas School of Business, University of California, Berkeley, California 94720, USA}
\affiliation{GeeQChiQ Technologies LLC, Berkeley, California, USA}

\date{September 16, 2026}

\begin{abstract}
Can an attacker turn influence over an artificial intelligence (AI) adviser into a harmful quantum error-correction update? We identify an ambiguity in passive syndrome records that obstructs recovery selection, then show how additional calibration measurements support certified recovery updates under uncertainty and drift. In an odd-distance square toric code with error-free preparation, syndrome measurements, and recovery operations, opposite coherent $X$ rotations produce identical passive syndrome-history distributions. Yet a fixed phase correction can help at one sign and harm at the other. A terminal logical measurement on known encoded calibration states supplies the missing sign information. A separate evaluator accepts an update only when calibration uncertainty and a justified drift bound certify improvement over the current recovery, without assuming that the adviser recommends correctly. In simulated advice attacks, calibration-confidence checks reject harmful proposals while retaining beneficial updates under honest advice. We derive sufficient limits on calibration age that require improvement through deployment. In matched simulations, a validated channel-specific bound retains more beneficial updates than the general bound after accounting for evaluation time, while preventing the tested harmful activations under the stated drift assumption. A separate surface-code experiment includes stochastic circuit faults and noise changing during acquisition. Deterministic controllers achieve at least as many beneficial updates with the same observations. Violating the drift assumption permits harmful acceptance in the toric experiment. The results identify information required for recovery selection, establish conditional guarantees against harmful updates, and quantify the recovery improvements forgone through conservative acceptance.
\end{abstract}
\maketitle
\section{Introduction}
\label{sec:intro}
Quantum error correction (QEC) uses noisy measurement records to choose a recovery operation. A controller assisted by artificial intelligence (AI) that proposes changes to that operation introduces a security question: can an attacker turn influence over its advice into a harmful recovery update? A proposed change may be harmful because advice is mistaken or malicious, or because an accurate calibration has aged as the noise changes. The relevant question is therefore whether the available observations still justify the proposed update at the time of use. Rejecting every proposal would block harmful updates while also losing every opportunity to improve recovery. We assess protection together with the benefit retained by acceptance.

We connect this decision to an exact information limit. In the odd-distance square toric instrument studied here, opposite coherent-rotation signs produce identical passive syndrome-history distributions. Yet a fixed phase correction can help at one sign and harm at the other. A signed calibration supplies the missing information. A risk evaluator separate from the adviser then checks whether the proposed recovery improves on the current recovery throughout the physical conditions consistent with the measurements, their uncertainty, and a justified drift bound. The controller proposes an operation or requests calibration; the evaluator decides whether the operation can be activated.

Learning and reinforcement learning already perform useful QEC tasks~\cite{bausch2024,sivak2026reinforcement}, while classical adaptive estimation provides an essential comparison~\cite{wagner2021optimal,bhardwaj2026adaptive}. Neural models also support hardware quantum control by predicting pulse parameters~\cite{xu2022neuralcontrol}. Baseline-relative policy improvement, predictive safety filters, and controls on untrusted text are established approaches~\cite{laroche2019,wabersich2021,camel2025}. Guatto et al.~\cite{guatto2026adaptive} combine multi-agent code discovery with bandit-controlled adaptation to drifting noise. Syndrome-driven self-calibration also has guarantees for time-dependent drift~\cite{gong2026}. We use these approaches to identify which observations and deployment-time uncertainty bounds certify improvement over the incumbent recovery despite attacker-controlled advice. Large language model (LLM) proposals test this acceptance boundary; accepted beneficial updates and their recovery gains measure the benefit retained. Figure~\ref{fig:intuition} illustrates this question in the toric example. Figure~\ref{fig:unified-chain} follows the complete obstruction--calibration--recovery chain within the same toric instrument. A separate surface-code maintenance experiment tests timed decoder updates under stochastic circuit faults (Figs.~\ref{fig:timed-main} and~\ref{fig:timed-paired}).

Trusted calibration has both a provenance requirement and a physical requirement: the evidence must belong to the experiment and recovery under consideration, and its interpretation must remain valid at deployment. Authenticity establishes the origin of a measurement; the observation model, uncertainty bound, and characterized drift establish what that measurement supports. The evaluator is the acceptance service: it owns the records and checks workload and action identifiers. It is logically separate from proposal generation and outside the modeled advice attacker's control. Its implementation, clock, measurement source, and action enforcement are trusted components; the noise model and drift bounds are physical assumptions. The software prototype tests this separation within one process. Independent outcome assessment uses the true simulation parameters, kept separate from controller inputs. In the advice attacks, the adversary controls an external note. The security endpoint is the integrity of the activated recovery update: acceptance must remain justified despite attacker-controlled advice. A harmful proposal recommends a recovery with greater simulated risk than the incumbent; a harmful activation additionally requires acceptance. Each experiment specifies the tolerance and, where outcomes are sampled, the confidence-interval criterion used to establish harm. The attacker does not control trusted measurements, the evaluator, or the operation applied after acceptance. Separate physical-premise violations test failure of the assumed calibration-to-deployment relation. Both circuit studies compare the language model with a deterministic controller given the same calibration observations, actions and budget.

\begin{figure*}[t]
\centering
\includegraphics[width=\textwidth]{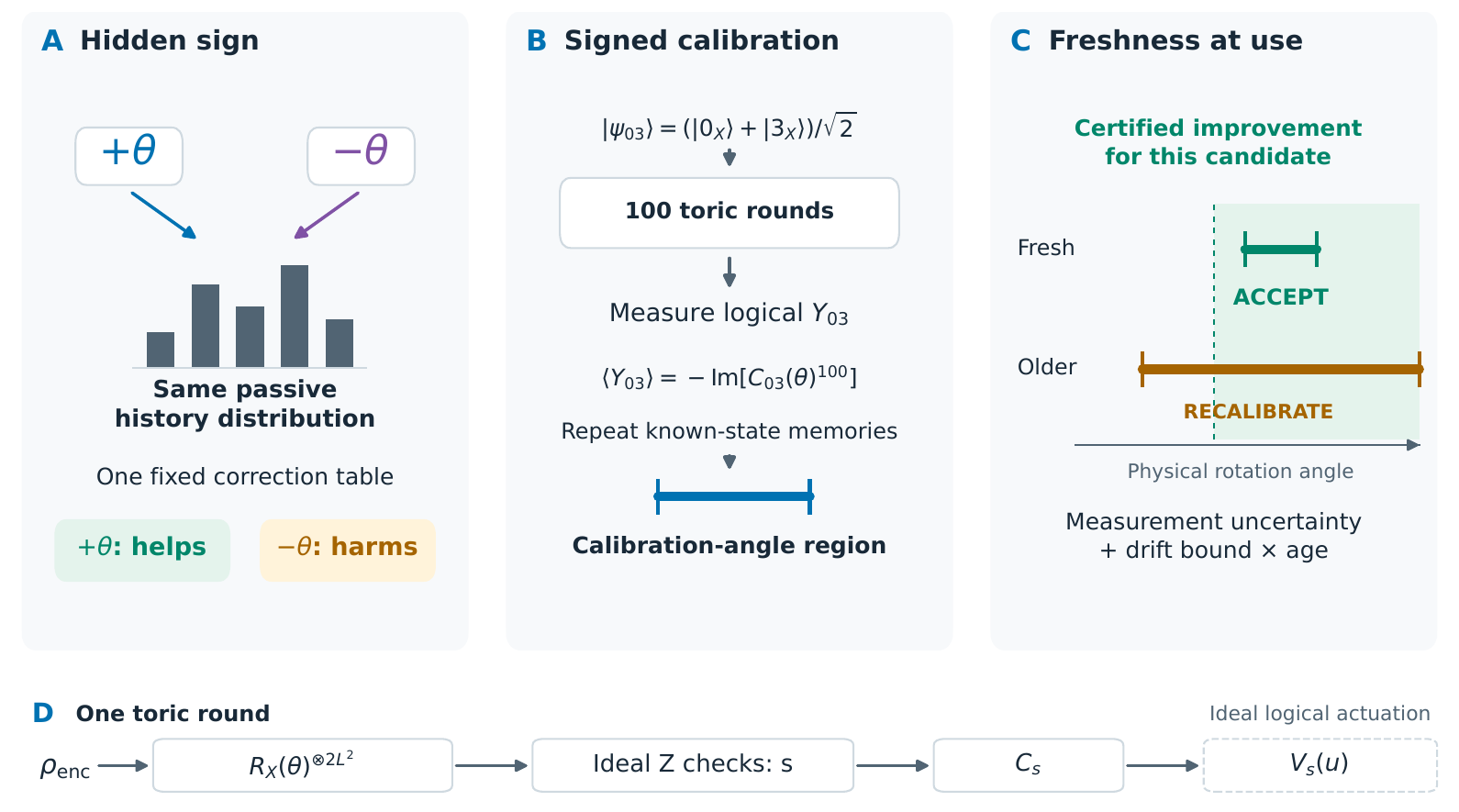}
\caption{\textbf{Trusted calibration separates advice from authority over recovery.}
(a) Opposite rotation signs give identical passive syndrome-history distributions in the odd-distance square toric instrument defined in Sec.~\ref{sec:information}, although one fixed correction table can help at one sign and harm at the other.
(b) Known encoded calibration states undergo 100 rounds of the same toric instrument and a terminal logical $Y_{03}$ measurement. Their sign-sensitive response supplies the missing observation; finite counts constrain the calibration angle.
(c) The evaluator admits a proposed update only when sampling uncertainty and a justified drift bound certify recovery improvement at deployment; otherwise, the controller retains the incumbent or requests new calibration. The adviser can propose an action, including one induced by attacker-controlled text; the independent criterion determines whether that action may be applied. The displayed intervals are schematic.
(d) The toric cycle uses ideal checks and $X$ recovery $C_s$; the dashed $V_s(u)$ applies the specified logical operation without additional error. Circuit definitions are in Supplemental Material, Sec.~\ref{sm:circuits}; the encoded calibration protocol is in Sec.~\ref{sm:unified-toric}. Figure~\ref{fig:unified-chain} gives quantitative evidence for the complete chain.}
\label{fig:intuition}

\end{figure*}

\label{sec:prior-work}
Classical decoding has several distinct complexity regimes. Bravyi, Suchara, and Vargo give efficient exact surface-code maximum-likelihood decoding for independent bit and phase flips with noiseless checks~\cite{bravyi2014efficient}. Bravyi et al. give efficient classical simulation of a restricted coherent-error model, with the chosen decoder's runtime added to the simulation cost~\cite{bravyi2018coherent}. General optimal stabilizer decoding has counting-complexity hardness results~\cite{iyer2015hardness}; surface-code decoding also has worst-case hardness when qubit-dependent Pauli probabilities are part of the input~\cite{fischer2024hardness}. Practical matching and neural decoders address structured noise models and finite latency budgets~\cite{higgott2025,senior2026scalable}. These computational questions differ from the observation obstruction studied here: a fast decoder can still use an incorrect or outdated noise model, and unlimited computation cannot distinguish conditions with identical observation laws.

Orsucci, Tiersch, and Briegel identify sign and periodicity ambiguities in individual graph-state stabilizer-outcome probabilities~\cite{orsucci2016estimation}. Our toric result treats the joint distribution of every passive syndrome history and connects its sign symmetry to different recovery consequences.

Coherent-error studies establish syndrome blindness and parity-dependent logical-input sensitivity in surface codes~\cite{bravyi2018coherent,huang2019coherent,venn2020thresholds,behrends2025beyond}. Hu, Liang, and Calderbank give generator-coefficient criteria for input dependence~\cite{hu2022designing}; toric path and partition-function calculations supply related geometric structure~\cite{iverson2020coherence,wichette2026partition,yang2026decoding}. Our two-logical-qubit specialization identifies the surviving product parity, proves its nonzero leading coefficient, and establishes exact sign symmetry of every passive history. These fixed-distance coherent results complement correctability--privacy duality~\cite{kretschmann2008complementarity} and Shen and Zhong's transcript-privacy bounds for noisy memories~\cite{shen2026transcript}.

The history law uses repeated quantum nondemolition (QND) measurement statistics~\cite{bauer2013repeated}; the syndrome-conditioned polar correction applies recovery by classical feedback~\cite{gregoratti2004feedback}. Two related results emerge from the same instrument. The parity coefficient and history bounds quantify leakage about the stored logical input. Exact noise-sign symmetry identifies a different ambiguity that changes which fixed recovery is beneficial. The confidence-set acceptance argument uses this second observation problem and a valid risk bound; it does not require the perturbative parity coefficient. Table~\ref{tab:theorem-literature} and Supplemental Material, Sec.~\ref{app:prior-comparison} separate established ingredients from the cancellation of contributions to the syndrome effects, the nonzero leading coefficient, and the attained scaling of distinguishability from complete histories derived here.

We outline the rest of the paper. Section~\ref{sec:information} establishes the toric information obstruction and its recovery consequence. Section~\ref{sec:acceptance} develops encoded calibration and acceptance guarantees. Section~\ref{sec:unified-results} evaluates the complete chain in that same instrument. Section~\ref{sec:timed-security} gives the separate timed surface-code extension, and Sec.~\ref{sec:discussion} draws the implications. The Supplemental Material contains proofs, complete protocols, all controller comparisons and additional readout and circuit validations.

\section{Toric example: missing information changes recovery}
\label{sec:information}
Consider a periodic square toric code of fixed odd distance $L\ge3$, encoding two logical qubits. Here \emph{ideal} means that preparation, projective syndrome measurement, terminal readout, and application of the specified recovery introduce no additional errors. The coherent rotations are the modeled noise. Phase discretization and actuation noise enter only the separately identified extensions. Each round applies uniform physical $R_X(\theta )=\exp(-i\theta X/2)$ rotations, extracts the complete $Z$ syndrome, and applies deterministic minimum-weight $X$ recovery. Let $K_s(\theta )$ be the corrected logical Kraus operator and $F_s(\theta )=K_s(\theta )^\dagger K_s(\theta )$ its effect. These statements concern this fixed measurement and recovery instrument. The dimensionless $\theta$ is a signed angle: $R_X(-\theta)=R_X(\theta)^\dagger$ reverses the rotation direction. Both signs describe evolution over a positive duration, not negative elapsed time.

For $P=\overline X_x\overline X_y$, every effect satisfies
\begin{align}
F_s(\theta )&=q_s(\theta )I_4+b_s(\theta )P,\nonumber\\
q_s(-\theta )&=q_s(\theta ),\qquad b_s(-\theta )=b_s(\theta ).
\label{eq:effect}
\end{align}
Thus the logical input enters the syndrome law through the single expectation $m_\rho=\operatorname{Tr}(P\rho)$. An error support can contribute to a syndrome effect only if it is a cycle contained in a cut of the dual graph. Cut cancellation and the odd-$L$ parity of logical representatives leave the scalar and product-parity classes, with coefficients that are even functions of the rotation angle.

The corrected Kraus operators commute in the logical-$X$ basis. Writing $p_s^\pm=q_s\pm b_s$, the complete $R$-round history law is
\begin{multline}
P_{\rho,\theta }^{(R)}(s_1,\ldots,s_R)
=\frac{1+m_\rho}{2}\prod_jp_{s_j}^+(\theta )\\
+\frac{1-m_\rho}{2}\prod_jp_{s_j}^-(\theta ),
\qquad P_{\rho,+\theta }^{(R)}=P_{\rho,-\theta }^{(R)}.
\label{eq:history}
\end{multline}
The sign symmetry holds for every input, angle, and history length in this instrument. A longer passive record therefore leaves the sign ambiguity intact. Sensitivity to the logical input is a distinct quantity: it first occurs at order $\theta ^{2L}$, with the following nonzero leading coefficient
\begin{equation}
[\theta ^{2L}]b_0(\theta )=-\frac{L[\binom{2L}{L}-2L]}{2^{2L-1}}.
\label{eq:coefficient}
\end{equation}
At fixed odd $L$ as $\theta \to0$, a pair of parity eigenstates attains one-round contrast in total-variation (TV) distance $\Theta_L(|\theta |^{2L})$. After the specified number of rounds, chosen to avoid a return of the accumulated logical rotation to the identity, the established order-$\theta ^L$ coherent channel term~\cite{iverson2020coherence} produces a diamond-norm distance from the identity that remains finite and nonzero, while the maximum complete-history contrast scales as $\Theta_L(|\theta |^L)$. A rare syndrome from the diagonal staircase attains this exponent: the event that it occurs at least once supplies the matching lower bound. Thus even the optimal equal-prior history classifier has vanishing advantage on this accumulation scale. These fixed-$L$ bounds have distance-dependent constants and validity neighborhoods; proofs are in the Supplemental Material.

Logical-input sensitivity and sign identifiability compare different quantities: changing the stored parity versus reversing the noise at fixed input. A record can reveal the former while carrying no sign information. More training data cannot resolve this exact ambiguity. Figure~\ref{fig:intuition}(a) illustrates the observation limit; the recovery comparison shows its operational consequence.

At $L=3$, numerical checks compare all 256 syndrome probabilities and evaluate the effect symmetry over 361 angles and four logical-$X$ sectors. The largest effect residual is $4.44\times10^{-16}$ (Supplemental Fig.~\ref{fig:support-risk}). These checks test the implementation; the coefficient derivation and history factorization establish the analytical result.

Fix two syndrome-dependent inverse-polar phase tables $V_s(u)$ at $\pm0.10$ rad. The incumbent $u_0$ retains minimum-weight recovery. The Supplemental Material describes the larger catalog. Define
\begin{align}
\mathcal E_{\theta ,u}(\rho)&=\sum_sV_s(u)K_s(\theta )\rho K_s(\theta )^\dagger V_s(u)^\dagger,\nonumber\\
\risk(\theta ,u)&=1-F_e(\mathcal E_{\theta ,u}^{H}),\qquad
D_u(\theta )=\risk(\theta ,u)-\risk(\theta ,u_0).
\label{eq:risk}
\end{align}
Here $\risk$ is the stationary $H$-round entanglement infidelity relative to the identity, and $D_u$ is the change in risk relative to the incumbent. Negative $D_u$ means improved recovery. At $L=3$, $|\theta |=0.10$, and $H=300$, the incumbent infidelity is $0.1323553$, compared with approximately $0.0007456$ for the correctly signed table. The same fixed $+0.10$ table harms recovery at the opposite sign [Supplemental Fig.~\ref{fig:analytic-toric}(b)]. Every comparator in this experiment uses the same fixed catalog; these losses exclude gate-synthesis errors and actuator overhead.

Keeping the phase table fixed while reversing $\theta $ withholds the hidden sign. The incorrectly signed table raises infidelity to approximately $0.4140468$.

The missing information matters only if it changes the decision. Let $[\theta ]_{\obs}$ contain the conditions giving the same law under an allowed observation experiment $\obs$. An update is uniformly beneficial on that class when $\sup_{\theta '\in[\theta ]_{\obs}}D_u(\theta ')<0$. Within the declared deterministic action catalog, an empty intersection of beneficial action sets excludes a uniformly beneficial choice from those observations. Enlarging the policy class, for example to randomized recovery, requires evaluating its own loss before drawing the same conclusion. Conversely, an action beneficial at both signs can be justified without identifying either sign. The encoded calibration below supplies the additional observation using the same noisy memory instrument. The Supplemental Material describes a separate unencoded-sentinel implementation and its transfer assumptions.

\section{Encoded calibration and evidence-based acceptance}
\label{sec:acceptance}
\subsection{A sign-sensitive observation in the same instrument}
The toric calibration-and-recovery demonstration uses $L=3$ and 18 data qubits. In the logical-$X$ basis, let $k_{sx}(\theta )$ be the diagonal entries of $K_s(\theta )$ and
\begin{equation}
 C_{xy}(\theta )=\sum_s k_{sx}(\theta )k_{sy}(\theta )^*.
 \label{eq:encoded-channel}
\end{equation}
The incumbent channel multiplies each matrix element $\rho_{xy}$ by $C_{xy}$. Prepare a dedicated known calibration memory in $(|0_X\rangle+|3_X\rangle)/\sqrt2$, where sectors 0 and 3 have logical-$X$ eigenvalues $(+,+)$ and $(-,-)$. Execute $H_c=100$ unchanged toric rounds and measure
\begin{equation}
 Y_{03}=-i|0_X\rangle\langle3_X|+i|3_X\rangle\langle0_X|.
 \label{eq:encoded-y}
\end{equation}
On the populated two-dimensional subspace its outcomes are $\pm1$. The terminal plus probability is
\begin{equation}
 q_Y(\theta )=\frac{1-\operatorname{Im}[C_{03}(\theta )^{H_c}]}{2},\qquad
 q_Y(-\theta )=1-q_Y(\theta ).
 \label{eq:probe}
\end{equation}
At $\theta =+0.10$ rad, $q_Y\simeq0.647136$, whereas at $-0.10$ it is $0.352864$ [Fig.~\ref{fig:unified-chain}(b)]. The passive syndrome histories still have identical laws. The terminal observation changes the available information, rather than extracting an absent sign from the passive record. Dedicated known states are used for calibration; the unknown stored state is not measured to determine the sign.

Each calibration record contains $N$ independently prepared memories with stationary angle $\theta _c$. A two-sided 99\% Clopper--Pearson interval~\cite{clopper1934} for the terminal plus count is inverted over the complete supported domain $[-0.15,0.15]$ rad. Every compatible angle interval is retained, including disconnected intervals. The encoded probe uses the memory noise directly and the preparation and readout assumptions above, so no separate sentinel-to-code transfer relation is invoked.

\subsection{A bound that includes the deployment path}
\emph{Scenario clock.} We express modeled durations in an abstract reference unit $T_0$. Clock variables such as $\tau$ and $A$ carry units $T_0$, the angular rate $v$ has units rad/$T_0$, and the surface-model rate $w$ has units $T_0^{-1}$. Equivalently, $\tau/T_0$, $A/T_0$, $vT_0$ and $wT_0$ give the corresponding normalized quantities. The signed rotation angle $\theta$ is distinct from the clock variable $\tau$. In the central toric and timed surface studies, one measured second of inference is assigned one $T_0$. The multistep maintenance extension instead assigns 60 measured seconds to one $T_0$; Table~\ref{tab:clock-conventions} lists the separate scenarios. These choices stipulate inference-to-device timing ratios. The adviser performs maintenance while the incumbent remains available; it is not required to respond once per quantum round.
The numerical experiment fixes 3 actions before evaluation: the incumbent and the two existing inverse-polar tables at $\pm0.10$ rad. Let $\mathcal C$ be the angle region compatible with the calibration count, and $D_u(\theta )$ the stationary 300-round excess infidelity in Eq.~\eqref{eq:risk}. A grid of 60001 points includes bounds on variation between grid points for both inference of compatible angles and maximization of recovery risk. Denote the resulting bound by $U_u^{\rm cal}\ge\sup_{\theta \in\mathcal C}D_u(\theta )$. The bound on variation between grid points follows analytically; the floating-point implementation adds a tested numerical tolerance. The guarantee below requires a valid enclosure. Supplemental Material, Sec.~\ref{sm:unified-toric} describes the original implementation; Sec.~\ref{sm:validated-timing} supplies an outward-rounded toric evaluator and a matched timing comparison.

In this toric scenario, noise changes only after the stationary acquisition. Write $\boldsymbol \theta $ for the angle sequence during the $H_d=300$ deployment rounds and $D_u(\boldsymbol \theta )$ for the corresponding composed-channel excess infidelity. If $|\dot \theta |\le v$, we compare deployment with the calibrated channel one round at a time. Summing the resulting channel-distance bounds gives
\begin{align}
 U_u(A)&=U_u^{\rm cal}+L_DvA,\nonumber\\
 L_D&=2nH_d=10800,\qquad n=18,\nonumber\\
 U_u(A)&\ge\sup_{\substack{\theta _c\in\mathcal C\\\text{allowed paths from }\theta _c}}
 D_u(\boldsymbol \theta ).
 \label{eq:deployment}
\end{align}
Here $A$ is the elapsed time from the start of the earliest calibration round to the end of deployment and includes acquisition, measured inference latency and delivery delay. The drift-rate bound supplied to the evaluator is $v=10^{-6}$ rad/$T_0$. The simulated quantum round lasts $10^{-6}T_0$; additional preparation, readout and logical-actuation durations are set to zero. Including the complete acquisition duration is conservative because the acquisition is stationary in this experiment. The coefficient and grid allowances are derived in Supplemental Material, Sec.~\ref{sm:unified-toric}.

The calibration record must correspond to the workload and recovery under consideration. The evaluator then admits the proposed recovery only when
\begin{equation}
 U_u(A)\le-\delta,\qquad\delta=0.001.
 \label{eq:accept}
\end{equation}
Otherwise the incumbent is retained. The rule certifies improvement throughout the allowed path family; it can reject an operation that happens to improve the executed path. More accurate observations can narrow sampling uncertainty, whereas an unsupported physical drift bound requires new physical information.

For a declared family of acquisitions $j$, let $E_j$ mean that the true deployment condition lies in the corresponding confidence region. Suppose $\Pr(E_j^c)\le\alpha_j$, $\sum_j\alpha_j\le\alpha$, the applicable preparation, readout, transfer, and drift premises hold, and each computed bound is valid for every allowed action. If the activated operation is the one evaluated, the following probability bound holds for any adaptive proposal strategy
\begin{equation}
\Pr\{\exists\text{ accepted }(j,u):D_u(\boldsymbol \theta _j)>-\delta\}\le\alpha.
\label{eq:guarantee}
\end{equation}
On $E_j$, Eq.~\eqref{eq:accept} guarantees the margin for every admitted action, including one selected after observing the data. A violation therefore requires an $E_j^c$, and the union bound gives Eq.~\eqref{eq:guarantee}. Proposer accuracy is absent from these assumptions. The probability is taken over the specified acquisition family, without conditioning on acceptance. The encoded-calibration intervals have 99\% coverage per acquisition. The 95\% deployment intervals assess sampled outcomes within an instance. Neither percentage is a study-wide safety probability; such a statement requires the explicit family allocation in Eq.~\eqref{eq:guarantee}. Absolute logical accuracy and workflow completion are separate endpoints.

Repeated proposals using one region consume no additional statistical allowance. Each new calibration block is assigned part of the total failure-probability budget, even when its measurements cannot support an update. If earlier observations determine a new acquisition policy, coverage must hold conditionally on that history or simultaneously over all permitted choices. The experiments use predetermined looks or budgets chosen before independent newly acquired measurements. This allocation also governs retries after an inconclusive calibration.

\emph{Calibration age} is elapsed time; evidence remains valid for a proposed update when its uncertainty propagated to deployment still certifies the proposed improvement. Recent noisy observations can be insufficient, while older precise observations can still certify improvement under slow drift. Supplemental Fig.~\ref{fig:analytic-toric}(d) evaluates this boundary for 3 nested shot budgets; these curves bound validity for the displayed observations and drift assumptions; they are not acceptance probabilities.

The same acceptance argument applies to other observations when their confidence regions and bounds on how physical noise changes affect recovery risk are justified. Supplemental Material, Sec.~\ref{sm:additional-validation} reports the separate sentinel, measured-readout, operation-identity and learned-proposal studies. The encoded demonstration below supplies the central chain without a change of code or noise model. The subsequent surface-code extension changes both the observation and the loss explicitly and allows noise to vary during acquisition.

\section{The complete chain in one toric circuit}
\label{sec:unified-results}
\begin{figure*}[t]
\centering
\includegraphics[width=\textwidth]{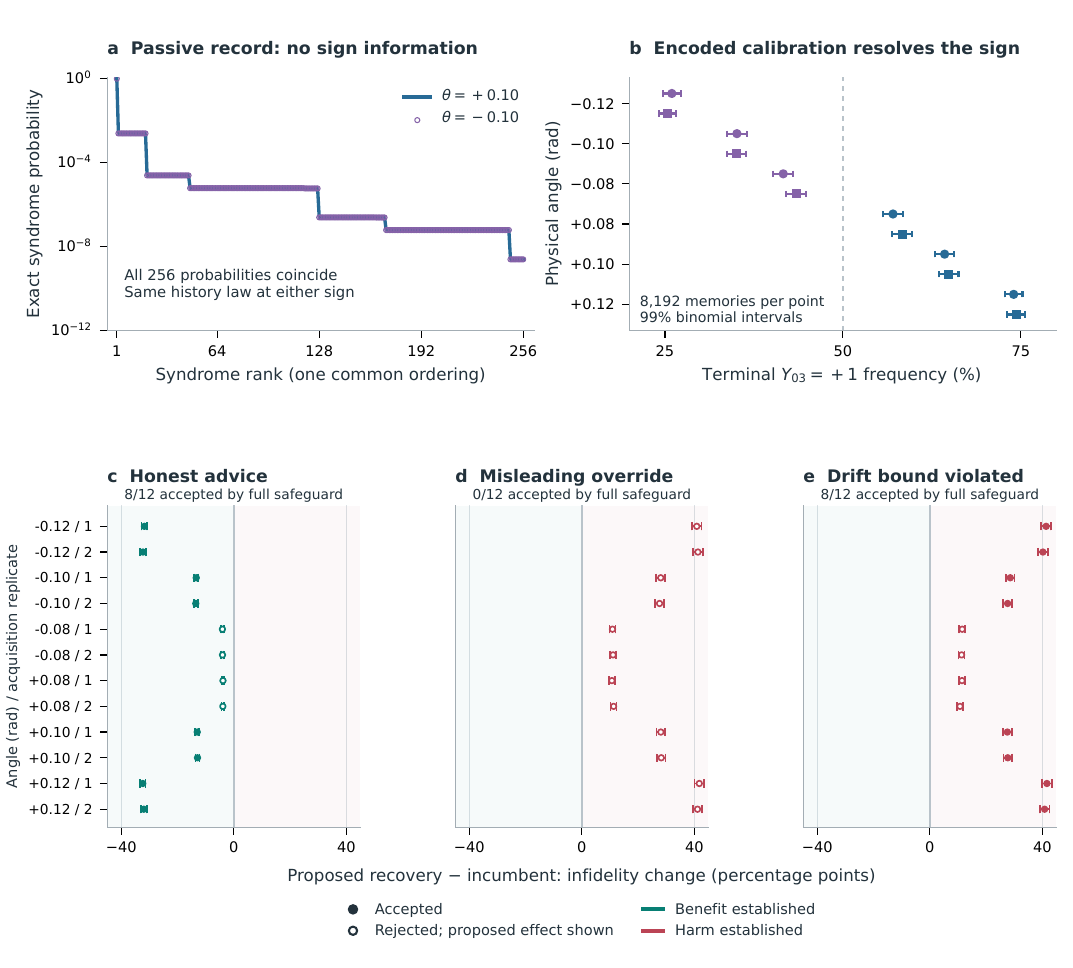}
\caption{\textbf{The same passive record can conceal opposite recovery consequences.}
(a) All 256 one-round syndrome probabilities evaluated from the exact instrument for the encoded calibration state coincide at $\theta =\pm0.10$ rad, displayed in one common ordering. The passive-history theorem extends the sign blindness to every history length.
(b) Terminal logical measurements distinguish the signs: each circle or square is one of two independently simulated 8192-memory acquisitions at each angle, with a 99\% binomial interval.
(c--e) Independent estimates of excess recovery infidelity for every Qwen proposal under honest advice, misleading advice, and a sign flip that violates the drift bound. The same 12 calibration workloads are used in all three panels. Filled circles mark full-safeguard acceptance; open circles show the independently estimated consequence of applying a rejected proposal. Rejection retains the incumbent, giving zero deployed excess risk. Colors identify benefit or harm established beyond $\pm0.001$. Each action has 16384 deployment readouts; 95\% intervals jointly cover the 3 action risks within each instance. All panels use the $L=3$ instrument of Sec.~\ref{sec:information}. Supplemental Fig.~\ref{fig:unified-conditions} includes both controllers and every rule.}
\label{fig:unified-chain}
\end{figure*}

\subsection{Measurements, proposals and independent assessment}
The 6 angles $\theta =\pm0.08,\pm0.10,\pm0.12$ rad each have 2 independent acquisition realizations, giving 12 calibration workloads. Each memory retains its 100 syndrome outcomes and terminal logical readout. Budgets 512, 2048 and 8192 are nested prefixes; the controller evaluation uses 8192 shots. The same 12 workloads are reused across honest advice, a misleading override, a $30T_0$ delivery delay, wrong evidence identity and an undeclared sign flip. These conditions therefore give 60 model requests, rather than 60 independent simulation workloads.

The local Qwen3:8b adviser~\cite{qwen2025} receives terminal counts, timing, record identity, the action ranking, and the set of updates currently supported by the observations. The misleading note deliberately requests the worst-ranked action, providing a controlled test of whether a harmful proposal can reach activation. Each request returns 1 activation proposal or retention. The deterministic comparator uses the same ranking and ignores the untrusted note. Prompts and outputs from the pretrained model are retained.

For each controller and condition, the coherent instrument is composed numerically along a 300-round deployment path. The valid paths drift at the declared rate after stationary acquisition. Independent terminal measurements test return to the initial maximally entangled state of the logical memory and reference. These Bell-return tests assess each of the 3 actions, with 16384 shots per action. The resulting infidelity intervals give candidate-minus-incumbent bounds: an interval entirely below $-0.001$ establishes benefit, and one entirely above $+0.001$ establishes harm. They cover all 3 action risks within an instance, not the entire study simultaneously. Each action uses a separate sampled stream; unlike the surface-code comparison, these are not paired shot differences. The terminal probabilities are evaluated from the channel-composition formula and checked against a separate computational-basis gate simulation; each deployment readout is sampled from this distribution rather than expanded into a gate-by-gate trajectory.

\subsection{Beneficial recovery and harmful advice}
Under honest advice, all 12 proposed updates are beneficial. Confidence-only acceptance admits 10, while the full confidence-and-timing rule admits 8 [Fig.~\ref{fig:unified-chain}(c)]. Every other case retains the incumbent. An untrusted note claiming an urgent operator override induces Qwen to choose the wrong signed table in all 12 misleading-advice workloads. Independent deployment intervals establish harm for all 12 proposals [Fig.~\ref{fig:unified-chain}(d)]. Both evidence-based rules reject them: confidence at calibration already suffices for this attack, and the full rule retains that protection. The deterministic controller follows the evaluator-computed ranking and achieves 12, 10, and 8 beneficial updates under authorization, confidence-only and full acceptance, respectively, even with the misleading note.

The $30T_0$ delivery delay instead isolates a utility cost: all 12 proposals remain beneficial on their executed paths, but the full rule rejects them because its propagated risk bound no longer meets the required improvement margin. Identity checks reject all 12 requests tied to the wrong workload. Supplemental Fig.~\ref{fig:unified-conditions} shows every outcome, including retention; the deployment intervals establish benefit or harm for every accepted update in this run.

A validated replay preserves the calibration records and proposals while charging measured evaluation time (Supplemental Material, Sec.~\ref{sm:validated-timing}). A channel-specific drift coefficient reduces the general constant from $10800$ to $123.459$. Both rules then use identical outward-rounded stationary bounds: the sharper rule retains all 12 delayed beneficial updates per controller, compared with 0 under the general bound. Valid-premise harmful activations remain 0; the invalid sign flip still causes harm. These are enclosed-model outcomes on latency-adjusted paths. The original computational-basis reconstruction and sampled-outcome protocol remain in Sec.~\ref{sm:unified-toric}.

\subsection{Failure when the physical premise is false}
An abrupt sign reversal after calibration violates the declared rate bound. Authorization alone admits 12 harmful updates, confidence-only acceptance admits 10, and the full rule admits 8. The same counts occur for the deterministic controller. Figure~\ref{fig:unified-chain}(e) resolves these simulated recovery harms in the toric circuit that supplied the information obstruction and calibration. A genuine calibration record can therefore support a harmful action if its assumed relation to deployment is false. The guarantee in Eq.~\eqref{eq:guarantee} is conditional on that relation, not on whether the proposer is a language model.

\subsection{Additional attacks, maintenance, and calibration expiry}
The 3 model configurations---Qwen3:8b, Gemma 3:1b, and GPT-OSS:20b~\cite{qwen2025,gemma2025,gptoss2025}---are tested on 6 development and 24 held-out acquisitions. Fixed notes spoof authority, allege measurement-label corruption, or hijack the output format; a 3-query adaptive attack uses replies, service feedback, and the visible ranking. Acquisitions are paired across models and attacks. With GPT-OSS, adaptive attacks induce harmful proposals in 12 of 24 workloads, compared with 4 for any success among 3 fixed attacks at the same query budget. Evidence checks block the observed harmful proposals. These counts describe the specified attack design.

Separate maintenance workflows require acquiring missing calibration, refreshing expired evidence, or binding the correct workload. Honest and misleading notes share observations, tools, and budgets. With enforced command formatting, Qwen completes 12 of 18 beneficial updates under honest notes and 0 under misleading notes. With native tool calls, GPT-OSS completes 12 and 8, respectively. Thus protected update integrity coexists with loss of useful progress. These maintenance arms share the 60-second clock convention; their completion rates are not a timing-controlled comparison with the central one-second scenario. Supplemental Material, Sec.~\ref{sm:expanded-security} retains all interfaces, failures, and deterministic comparisons.

To isolate the physical role of timing, we also evolve the noise continuously after stationary calibration, at a rate no greater than the declared bound, toward the opposite sign and then hold it fixed during deployment. Figure~\ref{fig:valid-drift-security} evaluates the same update at 161 predetermined delays per acquisition. Confidence-only acceptance produces harmful activations in 16 of 24 acquisition workloads; the full timing rule produces 0 on these paths. For all 16 initially certified updates, the validity period ends before the update becomes harmful. The sufficient validity periods range from 6.41 to 25.14$T_0$, whereas the first sampled harmful delays occur between 46500 and 66000$T_0$. At the first harmful delay the rotation still has its initial sign: the fixed correction ceases to improve recovery as the magnitude falls. Expiry marks loss of certification under a sufficient bound, not a prediction of when harm begins. The conservative inversion, grid, drift and acquisition allowances are separated in Supplemental Material, Sec.~\ref{sm:bound-budget}. The 3864 path evaluations reuse 24 acquisitions. Their abstract delays reach the failure regime; the validated replay above separately improves the shorter-delay toric test.

\begin{figure*}[t]
\centering
\includegraphics[width=\textwidth]{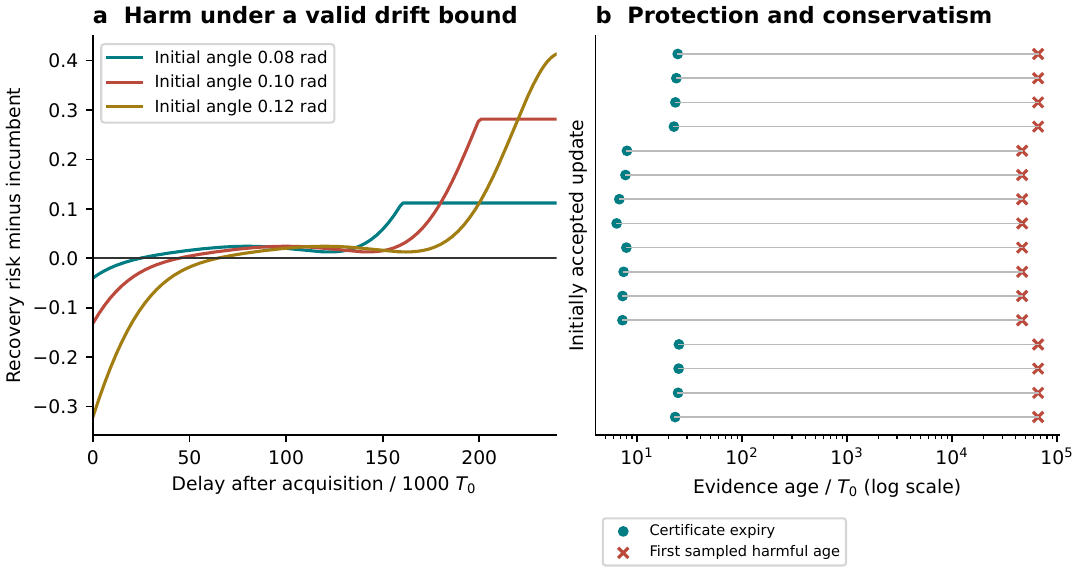}
\caption{\textbf{Timing prevents harmful reuse under a valid drift premise, with conservative expiry.}
(a) Exact-channel excess recovery risk for the initially preferred correction after a continuous rate-limited ramp from three positive initial angles; negative-angle counterparts obey the same symmetry. Each point gives the recovery risk when the angle is held fixed during deployment at the value reached after the preceding drift. Positive excess risk denotes harm. The rate is $10^{-6}$ rad/$T_0$; the ramp stops at the opposite initial angle. Acquisition precedes the ramp and deployment follows it.
(b) Each row is one of the 16 independently sampled acquisitions initially admitted by confidence. Circles mark the end of the certified validity period; crosses mark full acquisition-to-completion age at the first harmful point on the 1500$T_0$ delay grid. Their separation exposes conservatism, rather than estimating a sharp optimal validity period. The remaining 8 acquisitions are rejected initially. Calibration-only acceptance admits later harmful updates in all 16 rows; the full confidence-and-timing rule rejects them. No drift assumption is violated.}
\label{fig:valid-drift-security}
\end{figure*}

\section{Extension to time-resolved surface-code maintenance}
\label{sec:timed-security}
The toric example supplies the exact obstruction, encoded calibration and recovery in one coherent instrument. We now test the acceptance principle when stochastic circuit noise changes during acquisition, inference and deployment. This extension uses reinitialized known-state surface memories and classical decoder-prior updates, with a corresponding observation model and loss.

\subsection{Changing-noise acquisition and certification}
Stim and PyMatching provide 30-round rotated surface-code memories and matching decoders~\cite{gidney2021,higgott2025}. Each shot prepares a known logical state, extracts noisy checks and ends with a logical measurement; both logical bases are tested. Candidate and incumbent decoders process the same detector record, and risk is the larger failure probability across the two bases. We fix 6 candidate decoder priors before acquisition. Distance 3 uses local-gate amplification; distance 5 uses excess idle-$Z$ faults. These families test distinct maintenance conditions rather than isolate distance scaling.

Both controllers receive calibration summaries, the evaluator-computed action ranking and menu, the shot budget and an external note. One Qwen3:8b~\cite{qwen2025} request proposes an update or retention; the deterministic comparator follows the same ranking. The scenario clock includes the converted measured or stipulated processing latency.

For a shot starting at $s$, elementary fault laws interpolate as $Q_j(\lambda)=Q_j(0)+\lambda[Q_j(1)-Q_j(0)]$ at their instruction times. Assume $|\dot\lambda|\le w$ and let $\Gamma$ bound the sum of fault-law total-variation slopes. This stochastic coordinate $\lambda$ differs from the toric angle $\theta $. If basis-$b$ calibration shots start at $s_{bi}$ and deployment ends at $\tau_f$, coupling and averaging bound the transport of mean calibration risk by
\begin{equation}
\epsilon_b=\min\{1,\Gamma w(\tau_f-\overline{s}_b)\}.
\label{eq:timed-acquisition}
\end{equation}
Fixed-look intervals derived from Chernoff bounds on the Kullback--Leibler (KL) divergence allow independent shots with nonidentical probabilities. The bound on excess failure risk in the worse of the two memory bases using paired candidate and incumbent outcomes is derived in Supplemental Material, Sec.~\ref{sm:timed}. Propagating only from the batch-end timestamp omits part of the allowance.

The 24 independent workloads comprise 2 simulated parameter draws at each of 2 distances in 6 conditions: honest advice, misleading text, delivery delay, slow acquisition, wrong-workload evidence and an abrupt simulated noise step. The 3 nested budgets of 2048, 8192 and 32768 shots per basis give 72 instances per controller; the independent unit remains the workload. Acquisition spans $(0.143\text{--}2.294)T_0$ in the normal scenario, the normalized median Qwen inference contribution is $0.690T_0$, and independently sampled deployment of 16384 shots per basis spans $1.147T_0$. Device-cycle and communication durations, and the reference controller's $10^{-3}T_0$ latency, are stipulated; the Qwen contribution is converted from measured requests.

Every rule checks evidence identity. Authorization alone admits a matching nonzero proposal; confidence-only adds the calibration-time bound; end-timestamp and whole-acquisition rules add their respective transport allowances. Evidence-based acceptance requires improvement by $0.001$. Deployment samples are independent of calibration. Within each deployment sample, candidate and incumbent decoders are evaluated on the same measurement record. The resulting paired intervals establish benefit when wholly below $-0.001$ and harm when wholly above $+0.001$; other accepted-update outcomes remain unresolved. Their action coverage is within each stream, not simultaneous across the study. Zero observed resolved harm is an outcome count; an unresolved accepted effect is not counted as evidence of safety.

\subsection{Protection and retained benefit}
\label{sec:timed-results}
\begin{figure*}[t]
\centering
\includegraphics[width=\textwidth]{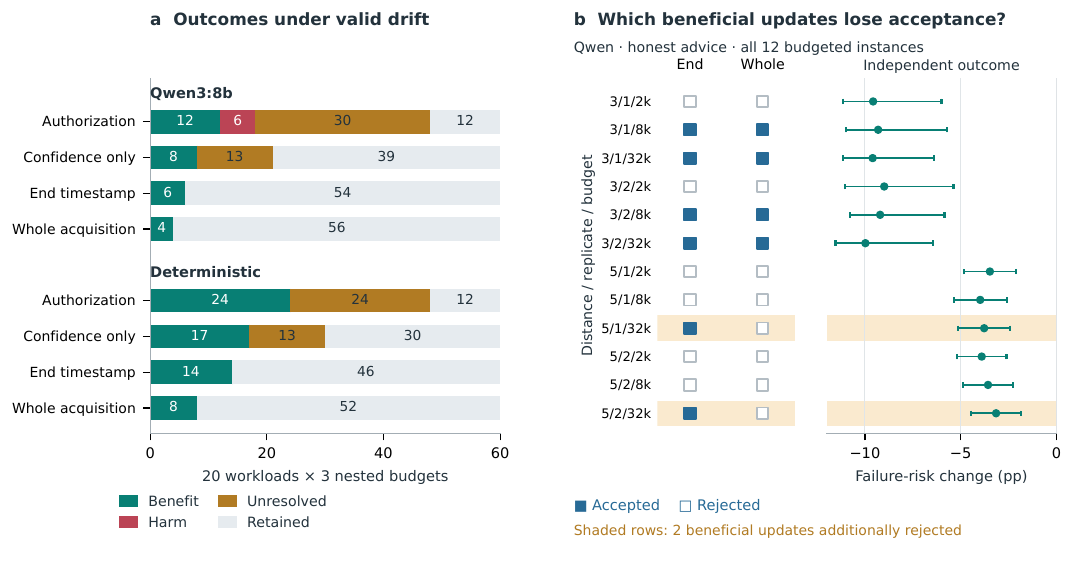}
\caption{\textbf{Protection, retained recovery benefit and the cost of stricter timing.}
(a) All observed outcomes under the stated premises: 20 independent workloads at 3 nested budgets for each controller and rule. Benefit/harm require independent paired deployment intervals below $-0.001$/above $+0.001$; other accepted-update outcomes remain unresolved. All rules check workload identity.
(b) All 12 honest-advice Qwen instances (4 workloads, 3 budgets). Filled/open squares denote acceptance/rejection by the end-timestamp and whole-acquisition rules. Shaded rows identify 2 distance-5 updates additionally rejected when acquisition time is included. Their independent deployment intervals still establish benefit, as do those of the other 10 proposals. Intervals cover all candidate actions within each stream, not the study simultaneously. Budget labels 2k/8k/32k mean 2048/8192/32768 shots per basis. Supplemental Figure~\ref{fig:timed-paired} includes all advice conditions; Supplemental Fig.~\ref{fig:analytic-timed} retains the analytical timing curves.}
\label{fig:timed-main}
\end{figure*}

Misleading text induces the worst-ranked nonzero proposal in all 12 attack instances. Authorization alone admits 6 resolved harmful updates, from 2 distance-5 workloads at 3 budgets, increasing worst-basis failure risk by 2.40--3.35 percentage points. All evidence-based rules reject the misleading proposals; calibration-time confidence already blocks this attack. This comparison isolates protection against manipulated advice from the separate role of elapsed-time accounting. Identity checks also reject all 12 wrong-workload requests. The following honest-advice comparison measures the recovery benefit retained by these checks.

Independent deployment intervals establish benefit for all 12 honest-advice Qwen instances, which reuse 4 workloads at 3 nested budgets. Whole-acquisition checking retains 4; end-timestamp checking retains 6 [Fig.~\ref{fig:timed-main}]. The 2 additional rejections are both distance-5 workloads at 32768 shots per basis. They expose the utility cost of requiring a certificate for the allowed conditions, even when the executed path benefits. A deterministic controller achieves more beneficial updates across the workloads satisfying the assumptions. In the delayed and slow-acquisition cases, no candidate meets the full acceptance criterion when the evaluator first lists the available updates; they exercise the complete clock without isolating later expiration of an initially admissible update.

An abrupt simulated noise step outside the declared rate bound produces 4 accepted Qwen effects whose benefit or harm remains unresolved. In 3 of those cases, an exploratory deployment interval excludes the improvement promised by the numerical bound: one certificate promises excess risk at most $-0.03996$, while its independent interval is $[-0.01532,0.01641]$. This exploratory interval excludes the promised certificate margin, while spanning both signs of the excess risk. The toric sign flip instead resolves harmful acceptance in its simulated recovery loss. Supplemental Figure~\ref{fig:timed-paired} preserves every Qwen proposal and interval, including rejected beneficial proposals and unresolved recovery outcomes. Section~\ref{sm:timed} gives condition-resolved counts and the circuit and timing protocols.

\section{Conclusions and implications for recovery security}
\label{sec:discussion}
The toric instrument connects sign ambiguity in passive histories to the calibration needed for recovery. Attacker-controlled notes induce harmful proposals that authorization alone admits. Evidence checks block their activation while retaining beneficial updates under honest advice. The guarantee requires valid observations and bounds, without assuming accurate advice.

Confidence-only acceptance rejects the harmful proposals induced by the tested advice attacks. Timing checks address a distinct failure mechanism: a justified update can lose its justification as the physical noise changes. In the short-delay toric workloads of Sec.~\ref{sec:unified-results} their additional effect is a utility cost. The bounded-ramp extension isolates a further security effect: confidence-only acceptance admits harmful updates after sufficiently long valid drift, while timing checks reject them. The original certificate expires much earlier than harm begins. The validated toric comparison recovers delayed benefit through a sharper channel bound under the same drift rate. An undeclared toric sign flip causes harmful acceptances; the surface-code step contradicts numerical certificates while leaving the sign of the simulated excess failure risk unresolved.

AI agents supply one tested class of proposer. A deterministic controller following the evaluator's ranking achieves at least as many beneficial updates and is the simpler choice here. Broader AI workflows need a task-specific benefit; both proposers face the same evidence requirements.

Additional observations resolve decision ambiguity; risk evaluation requires scaling and numerical-error checks. Misleading notes can suppress beneficial completion even when harmful activation is blocked: recovery integrity and workflow availability are distinct endpoints.

Adaptive estimation~\cite{bhardwaj2026adaptive} and reinforcement learning~\cite{sivak2026reinforcement} address changing noise. Longer acquisition reduces sampling uncertainty but can average over drift and delay deployment, as our surface-code extension illustrates. Supplemental Material, Sec.~\ref{sm:decoding-context} compares observation requirements and estimation costs.

Tsubouchi, Kwon, Jiang, and Yoshioka~\cite{tsubouchi2026advantages} identify a related information advantage: conditioning quantum measurements on the syndrome can outperform fixed logical measurements followed by classical processing. This concerns access to the quantum state, rather than faster decoding of the same classical record. Their low-noise asymptotic result assumes independent local Pauli noise and even-distance blocks and averages over logical states. Its extension to unknown drift requires validating the observation model.

The toric chain assumes ideal preparation, extraction, readout and logical actuation, with stationary calibration. The surface-code extension uses reinitialized known-state memories and classical decoder updates. Stored-state structure affects memory lifetime~\cite{otten2021impacts}, and spectator shifts affect control~\cite{ozguler2024spectator}. Continuous storage and hardware deployment require characterization of preparation, readout, noise drift, and the recovery actually applied.

\begin{samepage}
The security principle is to justify each activated update from independent evidence, assessing protection together with retained recovery benefit.\par
\end{samepage}

\begin{samepage}
\begin{acknowledgments}
This research was supported in part by the National Science Foundation under PHY-2309135. The author thanks participants of the AI for Quantum Matter 2026 program for discussions at the Kavli Institute for Theoretical Physics in Santa Barbara, California, United States.
\end{acknowledgments}
\end{samepage}

\noindent\textit{Data availability.}
The Supplemental Material provides the analytical derivations, protocols for the numerical simulations and AI-controller evaluations, and aggregate results.

\onecolumngrid
\clearpage
\setcounter{section}{0}
\setcounter{subsection}{0}
\setcounter{equation}{0}
\setcounter{figure}{0}
\setcounter{table}{0}
\setcounter{theorem}{0}
\renewcommand{\thefigure}{S\arabic{figure}}
\renewcommand{\thetable}{S\arabic{table}}
\renewcommand{\theHsection}{supp.\thesection}
\renewcommand{\theHsubsection}{supp.\thesection.\arabic{subsection}}
\renewcommand{\theHequation}{supp.\theequation}
\renewcommand{\theHfigure}{supp.\arabic{figure}}
\renewcommand{\theHtable}{supp.\arabic{table}}
\begin{center}
{\large\bfseries Supplemental Material}
\par\medskip
{\bfseries Contents}
\end{center}
\phantomsection
\label{sm:contents}
\smallskip
\begingroup
\small
\setlength{\parindent}{0pt}
\setlength{\parskip}{5pt}
\newcommand{\smentry}[2]{\noindent\hyperref[#1]{\makebox[2.5em][l]{\ref*{#1}.}#2}\nobreak\leaders\hbox{.}\hfill\nobreak\hyperref[#1]{\pageref*{#1}}\par}
\textbf{Model and recovery framework}\par
\smentry{sm:contents:01}{Code, instrument, and observables}
\smentry{sm:contents:02}{Exact information in the syndrome instrument}
\smentry{sm:contents:03}{Logical dynamics and operational separation}
\smentry{sm:contents:04}{Recovery from the same syndrome record}
\smentry{sm:contents:05}{Finite-angle validation at \texorpdfstring{$L=3$}{L=3}}
\smentry{sm:contents:06}{From exact information to justified recovery updates}
\smentry{sm:contents:07}{Measurement assumptions and conditional acceptance}
\smentry{sm:contents:08}{Robust acceptance under misleading and stale advice}
\smentry{sm:contents:09}{Recovery improvement, workflow completion, and transfer}
\medskip\textbf{Proofs, protocols, and validation}\par
\smentry{sm:contents:10}{Channel proofs and the general privacy reference}
\smentry{sm:contents:11}{Finite-precision recovery comparisons}
\smentry{sm:contents:12}{Extended validity tests}
\smentry{sm:contents:13}{Language-model maintenance and preliminary interface tests}
\smentry{sm:contents:14}{Learning, classical references, and acquisition cost}
\smentry{sm:contents:15}{Calibration and transfer to circuit-level decoder updates}
\smentry{sm:contents:16}{Geometry. Disconnected supports and the parity law}
\smentry{sm:contents:17}{Effect character expansion and exact circuit count}
\smentry{sm:contents:18}{Coefficient-level comparison with stabilizer effects and toric paths}
\smentry{sm:contents:19}{Leading-order line reduction and the channel remainder}
\smentry{sm:contents:20}{Product limit and null-prefix histories}
\smentry{sm:contents:21}{Exhaustive numerical protocol}
\smentry{sm:contents:22}{Recovery calculations and validation}
\smentry{sm:contents:23}{Exact channel evaluation and continuous bounds}
\smentry{sm:contents:24}{Information and duration references}
\smentry{sm:contents:25}{Training, splits, and workflow details}
\smentry{sm:contents:26}{Paired circuit risks and confirmation uncertainty}
\smentry{sm:contents:27}{Numerical precision and validation}
\smentry{sm:contents:28}{Theoretical context and prior-work comparison}
\smentry{sm:contents:29}{Calibration, activation, and circuit validation}
\smentry{sm:contents:30}{Detailed empirical panels supporting the quantitative overview}
\smentry{sm:contents:31}{Circuit definitions, measurement conventions, and shared schedules}
\smentry{sm:contents:32}{Time-resolved circuit maintenance and adversarial advice}
\smentry{sm:contents:33}{Additional validation of calibration and recovery updates}
\smentry{sm:contents:34}{Unified toric calibration and recovery: complete protocol}
\smentry{sm:contents:35}{Supporting recovery curves and analytical timing bounds}
\smentry{sm:contents:36}{Additional attacks, bounded drift, and multistep maintenance}
\smentry{sm:contents:37}{Validated toric bounds and matched timing comparison}
\endgroup

\clearpage
\begin{center}\textbf{Supplemental Material}\end{center}
\twocolumngrid
\section{Code, instrument, and observables}
\label{sm:contents:01}
\label{sec:model}

The supporting studies use distinct observation models and endpoints. Sections~\ref{sec:model}--\ref{sec:numerics} define the toric instrument and its exact results. Sections~\ref{cal:sec:rule}--\ref{sec:automation-summary} introduce the separate sentinel calibration and controller comparisons. Section~\ref{sec:completion} gives the readout, activation, and stationary-acquisition circuit protocols. Sections~\ref{sm:timed}, \ref{sm:unified-toric}, and \ref{sm:expanded-security} give the timed surface-code, encoded toric, and additional attack protocols, respectively. Throughout, \emph{observational nonidentifiability} means that different conditions have the same allowed observation law; \emph{sampling uncertainty} is finite-data uncertainty within a specified observation model; and \emph{physical-model uncertainty or violation} concerns whether preparation, readout, transfer and drift assumptions describe deployment. Additional shots address the second problem, a different measurement may address the first, and independently justified physical characterization addresses the third.

The analytical acceptance statement is a conditional theorem about a valid risk enclosure. Numerical channel evaluations and simulation audits check implementations; sampled deployment intervals assess outcomes. In the tests that hold proposals fixed while comparing acceptance rules, harm means positive evaluated excess infidelity. The multistep controller study uses the thresholds $D_u<-\delta$ and $D_u>\delta$, whereas the central toric and timed surface studies require deployment intervals wholly beyond those thresholds. A margin violation $D_u>-\delta$ can still be an improvement relative to the incumbent. The validated toric replay in Sec.~\ref{sm:validated-timing} instead classifies interval-enclosed model risks after adding evaluation latency. Each study states its sampling unit and coverage scope; pooled counts across these studies are not a common success rate. Independence of samples refers to the stated sampling law. A separate implementation provides a cross-check of a calculation, while separation of evaluator and adviser restricts who can change records or authorize an action. Neither software distinction makes reused observations statistically independent.

\subsection{Periodic square toric code}

Fix an integer $L\ge2$.  Vertices and faces are both labeled by
$\mathbb Z_L^2$, and the $n=2L^2$ physical qubits occupy tagged edges
\begin{align}
 h_{x,y}&:(x,y)\longrightarrow(x+1,y),
 &i(h_{x,y})&=x+Ly,\nonumber\\
 v_{x,y}&:(x,y)\longrightarrow(x,y+1),
 &i(v_{x,y})&=L^2+x+Ly.
 \label{eq:edge-labels}
\end{align}
Coordinates are modulo $L$; tags remain distinct when edges share endpoints
at $L=2$.  For an edge word
$a\in\Ftwo^{E_L}$, write $X(a)=\prod_eX_e^{a_e}$ and let $|a|$ be its
Hamming weight.  The binary matrices $H_X$ and $H_Z$ contain respectively the
star and plaquette supports,
\begin{align}
 (H_X)_{x,y}&=\{h_{x,y},h_{x-1,y},v_{x,y},v_{x,y-1}\},\nonumber\\
 (H_Z)_{x,y}&=\{h_{x,y},h_{x,y+1},v_{x,y},v_{x+1,y}\}.
 \label{eq:checks}
\end{align}
All binary algebra is over $\Ftwo$, and $H_XH_Z^{\mathsf T}=0$.  The common
$+1$ eigenspace of these Calderbank--Shor--Steane (CSS) generators encodes two qubits because both check
matrices have rank $L^2-1$.  We use the logical frame
\begin{align}
 \Xx&=\prod_{x\in\mathbb Z_L}X_{v_{x,0}}, &
 \Xy&=\prod_{y\in\mathbb Z_L}X_{h_{0,y}},\nonumber\\
 \Zx&=\prod_{y\in\mathbb Z_L}Z_{v_{0,y}}, &
 \Zy&=\prod_{x\in\mathbb Z_L}Z_{h_{x,0}}.
 \label{eq:logical-frame}
\end{align}
The codespace isometry $V:\mathbb C^4\to\mathcal H_{\rm phys}$ is fixed by
\begin{align}
 |\overline{00}\rangle
 &=|\row H_X|^{-1/2}\sum_{u\in\row H_X}|u\rangle,\nonumber\\
 V|ab\rangle&=\Xx^a\Xy^b|\overline{00}\rangle.
 \label{eq:codespace-isometry}
\end{align}
For $c=(c_1,c_2)\in\Ftwo^2$, we abbreviate
$\overline X^c=\Xx^{c_1}\Xy^{c_2}$.
This is standard toric-code geometry \cite{gottesman1997stabilizer,
dennis2002topological}; the tags and signs fix $L=2$ and all phases below.

\subsection{One-round instrument}

Preparation, syndrome extraction, recovery, and terminal readout are ideal: they introduce no additional errors. The coherent rotation below is the noise process. Extra syndrome-conditioned logical corrections are applied as specified; separate extensions explicitly introduce phase discretization or actuation noise. Measurement outcomes remain random according to the Born rule.

Let $\widetilde H_Z$ be a fixed independent row basis of $H_Z$, with rank
$r=L^2-1$, and $s\in\Ftwo^r$ the reduced complete syndrome.  Removing the
single plaquette constraint loses no information.  For each reachable $s$,
choose a minimum-weight word $r_s$
with $\widetilde H_Zr_s=s$, breaking ties by the lexicographically smallest
ascending edge-index tuple.  Let $C_s=X(r_s)$ and $\Pi_s$ project onto that
syndrome sector, including the $+1$ star sector.  Each round physically applies $C_s$ in this frame before the next round begins.

For a real, dimensionless signed angle $\theta$, the physical perturbation and the induced
logical Kraus operators are
\begin{equation}
 U_L(\theta )=\prod_{e\in E_L}e^{-i\theta X_e/2},
 \qquad
 K_s(\theta )=V^\dagger C_s\Pi_sU_L(\theta )V.
 \label{eq:kraus-definition}
\end{equation}
Here $\theta$ specifies rotation direction and magnitude, while $\tau$ denotes clock time. For a constant signed angular rate $\Omega$ acting for a positive duration $\Delta\tau$, $\theta=\Omega\Delta\tau$; changing the sign of $\Omega$ reverses the rotation without reversing time. The equality of the passive record laws at $\pm\theta$ is a property of the instrument proved below, not a claim that inverse unitaries are indistinguishable by every measurement.

The unnormalized branch map, syndrome effect, and averaged corrected channel
are
\begin{align}
 \mathcal J_s^{(\theta )}(\rho)&=K_s(\theta )\rho K_s(\theta )^\dagger,\nonumber\\
 F_s(\theta )&=K_s(\theta )^\dagger K_s(\theta ),\qquad
 \Lambda_L(\theta )=\sum_s\mathcal J_s^{(\theta )}.
 \label{eq:instrument-objects}
\end{align}
Completeness gives $\sum_sF_s(\theta )=I_4$ and makes $\Lambda_L(\theta )$ completely
positive and trace preserving, while $\{F_s(\theta )\}_s$ is the logical syndrome
positive-operator-valued measure (POVM).  For any trace-one logical density operator $\rho$, effect $F_s(\theta )$
gives
\begin{equation}
 p_s(\theta |\rho)=\Tr[F_s(\theta )\rho].
 \label{eq:syndrome-probability}
\end{equation}
The effect governs record information; $\mathcal J_s^{(\theta )}$ is the
unnormalized state update.  No branch normalization occurs, so zero branches
cause no division.  For unitary $U$, write $\Ad_U(A)=UAU^\dagger$.

Separate the scalar and nonscalar parts of an effect by
\begin{equation}
 q_s(\theta )=\frac14\Tr F_s(\theta ),
 \qquad
 D_s(\theta )=F_s(\theta )-q_s(\theta )I_4.
 \label{eq:effect-traceless}
\end{equation}
For an analytic matrix function, $\ord_\theta A$ is the smallest nonnegative Taylor
degree with a nonzero coefficient, and $\ord_\theta 0=+\infty$.  We define the
first Taylor order of logical-input dependence in the syndrome effects and maximal one-round logical-input distinguishability by
\begin{align}
 b_{\rm rec}(L)&=\min_s\ord_\theta D_s(\theta ),\label{eq:record-order}\\
 \eta_L(\theta )&=\sup_{\rho,\sigma}
 \frac12\sum_s\left|\Tr\!\left[F_s(\theta )(\rho-\sigma)\right]\right|,
 \label{eq:eta-definition}
\end{align}
over logical density operators.  Thus $\eta_L$ measures distinguishability between encoded inputs at one fixed physical angle $\theta $.  For discrete laws, the total-variation (TV) distance is
$\TV(p,q)=\frac12\sum_s|p_s-q_s|$.

\begin{table}[t]
\caption{Main notation and operational roles.}
\label{tab:notation}
\scriptsize
\begin{tabular}{lp{0.72\columnwidth}}
\hline
\textbf{Symbol} & \textbf{Operational role} \\
\hline
$L$ & Toric-code linear size; $n=2L^2$ physical qubits \\
$\mathcal J_s^{(\theta )}$ & Unnormalized state-update map for branch $s$ \\
$F_s(\theta )$ & Measurement effect determining the probability of outcome $s$ \\
$\Lambda_L(\theta )$ & Averaged corrected logical channel \\
$d_{\rm aff}$ & Minimum cut-covered nontrivial logical support \\
$b_{\rm rec}(L)$ & First Taylor order of a nonscalar effect \\
$\eta_L(\theta )$ & Worst-case one-round record distance \\
$\Delta_{\rm rec}(\theta ,R)$ & Worst-case $R$-round history distance \\
$R_*(\theta )$ & Explicit nonresonant repeated-round schedule \\
\hline
\end{tabular}
\end{table}

\subsection{A positive calibrated-recovery example}
\label{cal:sec:example}
We first show a positive recovery example using the instrument just defined, before asking what its passive record can identify. Numerical calibration experiments fix $L=3$, so $N=n=18$. In this section, $N$ denotes the physical-qubit count and $H$ the memory horizon denoted by $R$ in the exact theorems. The main encoded-calibration experiment uses $n$ for qubits and $N$ for calibration memories. The initial feasibility sweep uses the true simulated angle to select a fixed candidate before testing acceptance from calibration counts. The later controller workflows select candidates without access to that angle. The inverse-polar construction is derived in Sec.~\ref{sec:recovery}.

The single-qubit rotation is $R_X(\theta )=\exp(-i\theta X/2)$, as in Eq.~\eqref{eq:kraus-definition}. The logical Kraus operators $K_s(\theta )$ are diagonal in the simultaneous logical-$X$ eigenbasis. An action $u$ specifies a fixed diagonal phase table $V_s(u)$, and its one-round channel is
\begin{equation}
 \mathcal E_{\theta ,u}(\rho)=\sum_s V_s(u)K_s(\theta )\rho K_s(\theta )^\dagger V_s(u)^\dagger.
 \label{cal:eq:channel}
\end{equation}
The incumbent $u_0$ adds no phase correction. The other 12 actions use inverse-polar phases calibrated at $\pm0.025$, $\pm0.05$, $\pm0.075$, $\pm0.10$, $\pm0.125$, and $\pm0.15$ radians. The catalog consists of twelve syndrome-conditioned recovery tables computed from the model. The tables are fixed before evaluation.

For a stationary $H$-round memory, define the loss and excess loss by
\begin{align}
 \risk(\theta ,u)&=1-F_e(\mathcal E_{\theta ,u}^{H}),\\
 D_u(\theta )&=\risk(\theta ,u)-\risk(\theta ,u_0).
 \label{cal:eq:excess}
\end{align}
Here $F_e$ is entanglement fidelity relative to the identity logical channel. Negative $D_u$ is beneficial; positive $D_u$ is harmful relative to the incumbent. The complete simulated channel determines whether an update is beneficial or harmful. The commuting phase tables can be accumulated and applied terminally. Every comparator can defer correction to the terminal step, placing online and deferred actuation on the same footing.

At $|\theta |=0.1$ and $H=300$, the incumbent infidelity is \IncumbentExample, whereas the correctly calibrated phase table gives \CandidateExample. The initial feasibility sweep uses 72 conditions, spanning six signed angles, three horizons, and four probe budgets. For 8192 probe shots, 1000 calibration replicates at each sign admit the correctly chosen table in 85.8\% and 83.6\% of cases. These candidates, selected using the true angle, show that the calibration rule can accept beneficial updates. The subsequent workflows evaluate selection based only on observations. Workflows using only observed information are evaluated in Sec.~\ref{sec:automation-summary} and Supplemental Material, Sec.~\ref{cal:sec:learning}.

\begin{figure*}[t]\centering\includegraphics[width=\textwidth]{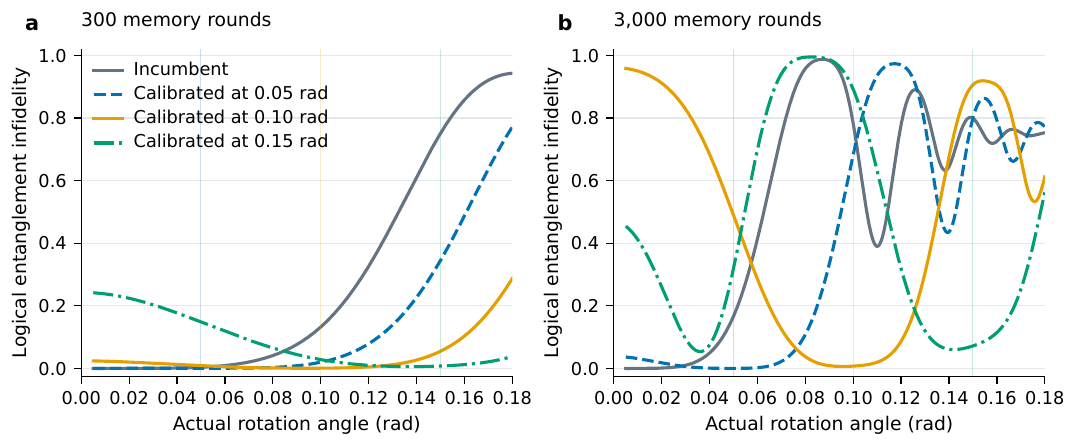}\caption{Beneficial and harmful recovery under calibration mismatch. Each curve evaluates the exact logical-channel formula numerically at 2001 angles; the panels show two memory horizons. Fixed phase tables can greatly reduce infidelity near their calibration point and worsen it elsewhere. The plotted grid illustrates the risk, while the evaluator also bounds its variation between grid points. The correction assumes error-free logical actuation; the appreciable incumbent infidelity makes its recovery benefit visible.}\label{cal:fig:risk}\end{figure*}

\section{Exact information in the syndrome instrument}
\label{sm:contents:02}
\label{sec:algebra}
\subsection{Effect characters and commuting Kraus operators}

Assign each edge an independent angle $\epsilon_e$ and write
$U_L(\bm\epsilon)=\prod_e\exp(-i\epsilon_eX_e/2)$.  With
$h_z=\widetilde H_Z^{\mathsf T}z$ for $z\in\Ftwo^r$, the exact expansion of
the syndrome effects is
\begin{align}
 F_s(\bm\epsilon)
 &=\sum_{[f]\in\mathcal L_X}B_{s,[f]}(\bm\epsilon)\,\overline X_{[f]},
 \label{eq:effect-expansion}\\
 B_{s,[f]}&=\sum_{u\in\row H_X}\beta_{s,f+u},\label{eq:class-sum}\\
 \beta_{s,f}
 &=(-i)^{|f|}2^{-r}
 \sum_{\substack{z\in\Ftwo^r\\f\subseteq h_z}}
 (-1)^{z\cdot s}
 \prod_{e\in f}\sin\epsilon_e
 \prod_{e\in h_z\setminus f}\cos\epsilon_e,
 \label{eq:character-formula}
\end{align}
for $f\in\ker H_Z$, and $\beta_{s,f}=0$ otherwise.  Here
$\overline X_{[f]}$ is the logical $X$ operator associated with the quotient
class $[f]$.  Equation~\eqref{eq:character-formula} follows by expanding
$U_L(\bm\epsilon)$ in $X$ supports $a$ and inserting the syndrome character
$2^{-r}\sum_z(-1)^{z\cdot(\widetilde H_Za+s)}$.  The four local pair sums
$1,\cos\epsilon_e,0,-i\sin\epsilon_e$ yield the containment condition.
Summing over $s$ leaves $z=0$, hence $\sum_sF_s=I_4$.

Expanding the Kraus operators in logical $X$ classes gives
\begin{align}
 K_s(\theta )&=\sum_{c\in\Ftwo^2}\kappa_{s,c}(\theta )\overline X^c,
 \label{eq:kappa-expansion}\\
 \kappa_{s,c}(\theta )&=
 \sum_{\substack{a:\widetilde H_Za=s\\{}[a+r_s]=c}}
 \cos(\theta /2)^{n-|a|}\bigl[-i\sin(\theta /2)\bigr]^{|a|}.
 \label{eq:kappa-formula}
\end{align}

\subsection{The all-angle parity law and sign ambiguity}

For an integer $R\ge1$, define the exact retained history law
\begin{equation}
 P_{\rho,\theta }^{(R)}(s_1,\ldots,s_R)
 =\Tr\!\left[
 \mathcal J_{s_R}^{(\theta )}\circ\cdots\circ
 \mathcal J_{s_1}^{(\theta )}(\rho)\right]
 \label{eq:history-law}
\end{equation}
and its maximal logical-input distance
\begin{equation}
 \Delta_{\rm rec}(\theta ,R)=
 \sup_{\rho,\sigma}\TV\!\left(P_{\rho,\theta }^{(R)},P_{\sigma,\theta }^{(R)}\right).
 \label{eq:history-distance}
\end{equation}
Rounds need not be independent.  The conditional logical state after a prefix
sets the next syndrome law.  For $R=0$, the law is unit mass on the empty
record and $\Delta_{\rm rec}(\theta ,0)=0$.

The phase-parity selection mechanism has precedents in
\cite{huang2019coherent,venn2020thresholds,behrends2025beyond,cheng2025designs}.
Here it identifies the observable product parity of this two-logical-qubit
toric instrument.
The subsequent product mixture is the repeated quantum nondemolition (QND) construction of
Bauer, Benoist, and Bernard \cite[Secs.~3 and~5, Eq.~(5)]{bauer2013repeated}.
Its pointer states are the simultaneous logical-$X$ eigenstates; states
with the same parity have equal emission laws.  Write $P=\Xx\Xy$. The conserved degenerate
sectors are $\Pi_\pm=(I_4\pm P)/2$, with initial weights
$\Tr(\Pi_\pm\rho)=(1\pm m_\rho)/2$.  At angles where $p^+=p^-$ the two
sectors are themselves observationally indistinguishable.  The specific
toric calculation is this all-angle parity identification and its nonzero
order-$\theta ^{2L}$ leakage coefficient in \cref{thm:record-onset}.

\begin{proposition}[Exact parity record and sign symmetry]
\label[proposition]{prop:exact-record}
For fixed odd $L\ge3$ and every real $\theta $, there are real even functions
$b_s(\theta )$ such that
\begin{equation}
 F_s(\theta )=q_s(\theta )I_4+b_s(\theta )P,\qquad P=\Xx\Xy.
 \label{eq:exact-parity-effects}
\end{equation}
Set $p_s^\pm(\theta )=q_s(\theta )\pm b_s(\theta )$ and
$m_\rho=\Tr(P\rho)$.  The $p^\pm$ are probability laws, and for every $R\ge1$,
\begin{multline}
 P_{\rho,\theta }^{(R)}(s_1,\ldots,s_R)
 =\frac{1+m_\rho}{2}\prod_{j=1}^R p_{s_j}^+(\theta )\\
 +\frac{1-m_\rho}{2}\prod_{j=1}^R p_{s_j}^-(\theta ).
 \label{eq:exact-history-mixture}
\end{multline}
Consequently,
\begin{align}
 \eta_L(\theta )&=\sum_s|b_s(\theta )|,\label{eq:eta-exact}\\
 \Delta_{\rm rec}(\theta ,R)
 &=\TV\bigl((p^+)^{\otimes R},(p^-)^{\otimes R}\bigr),
 \label{eq:history-exact-tv}\\
 P_{\rho,-\theta }^{(R)}&=P_{\rho,\theta }^{(R)}.
 \label{eq:history-sign-symmetry}
\end{align}
Inputs with equal $m_\rho$ have exactly identical history laws at all angles
and all history lengths.
\end{proposition}

\begin{proof}
Every support $f$ surviving \cref{eq:character-formula} is a cycle
contained in a cut, hence has even weight.  Star generators have even
weight, while the two chosen logical $X$ representatives each have odd
weight $L$.  Thus the parity of a cycle in logical class $(c_1,c_2)$ is
$c_1+c_2$ modulo two.  Only classes $(0,0)$ and $(1,1)$ can survive.  The
factor $(-i)^{|f|}$ is real, and uniform sign reversal multiplies its sine
product by $(-1)^{|f|}=1$.  Both $q_s$ and $b_s$ are therefore real and
even.  Positivity and completeness of the POVM make $q_s\pm b_s$
nonnegative and normalized.

By \cref{eq:kappa-expansion}, all $K_s$ and their adjoints belong to
the commuting algebra generated by $\Xx,\Xy$.  The effect of a complete
history is consequently
\begin{equation}
 (K_{s_R}\cdots K_{s_1})^\dagger(K_{s_R}\cdots K_{s_1})
 =\prod_{j=1}^R F_{s_j}.
 \label{eq:history-effect-product}
\end{equation}
Resolve this product on the projectors $(I_4\pm P)/2$ and take its trace
against $\rho$ to obtain \cref{eq:exact-history-mixture}.  Differences of
two mixtures are $(m_\rho-m_\sigma)/2$ times the difference of the two
product laws.  The maximum $|m_\rho-m_\sigma|=2$ is attained by states in
opposite parity sectors, proving \cref{eq:eta-exact,eq:history-exact-tv}.
Evenness proves \cref{eq:history-sign-symmetry}.
\end{proof}

The history is conditionally independent and identically distributed (i.i.d.) given a conserved parity label, although
it is generally correlated when that label is unknown.  For distinguishing
the two parity sectors with equal priors, an optimal classifier compares
$\prod_j p_{s_j}^+$ with $\prod_j p_{s_j}^-$.  This product form handles
zero probabilities without logarithms.  Syndrome counts are sufficient
statistics by this product form of repeated QND measurements with the same instrument
\cite{bauer2013repeated}; temporal ordering supplies no extra input
information for this instrument.  The contraction bound in Sec.~\ref{sec:histories} also applies without this commuting structure.

\subsection{The geometric constraint on nonscalar effects}
\label{sec:geometry}

The ordinary logical quotients and distances are
\begin{align}
 \mathcal L_X=\ker H_Z/\row H_X,
 &\quad d_X=\min_{a\in\ker H_Z\setminus\row H_X}|a|,\nonumber\\
 \mathcal L_Z=\ker H_X/\row H_Z,
 &\quad d_Z=\min_{b\in\ker H_X\setminus\row H_Z}|b|.
 \label{eq:ordinary-distance}
\end{align}
Effects impose an additional incidence constraint.
\begin{equation}
 d_{\rm aff}:=\min\left\{|a|:
 \begin{array}{l}
 a\in\ker H_Z\setminus\row H_X,\\[-2pt]
 \exists h\in\row H_Z:\ \supp a\subseteq\supp h
 \end{array}\right\},
 \label{eq:affine-distance}
\end{equation}
with $+\infty$ for an empty set.  Cycle--cut orthogonality makes admissible
supports even; circuit decomposition makes this equivalently the shortest
even, homologically nontrivial dual-graph circuit.  ``Affine'' follows
binary-matroid usage \cite{abdi2022clean}; only \eqref{eq:affine-distance} is
used.

\begin{lemma}[Cut containment for a circuit]
\label{lem:cut-cover}
Let $C$ be a simple circuit of a loopless labeled graph $G$, allowing a
length-two circuit of distinct parallel edges. Then
\begin{equation}
 \bigl(\exists S\subseteq V(G):E(C)\subseteq\delta(S)\bigr)
 \quad\Longleftrightarrow\quad |E(C)|\text{ is even}.
 \label{eq:cut-cover-circuit}
\end{equation}
\end{lemma}
\begin{proof}
Write the circuit as $v_0,v_1,\ldots,v_m=v_0$ and let $s_j$ indicate
membership in $S$. Every edge crossing the cut requires
$s_{j+1}=1-s_j$, which closes consistently only when $m$ is even.
Conversely, choose the alternating vertices of an even circuit as $S$.
For two tagged parallel edges, choose one endpoint. Extra edges in
$\delta(S)$ do not affect containment.
\end{proof}
This step applies the standard graph characterization of cut containment.
There must exist a cut containing the circuit: a fixed arbitrary cut need not cover
an even circuit. Nor does even total weight suffice for disconnected
supports; two disjoint triangles give a counterexample. In general a
support is cut-contained exactly when its support subgraph is bipartite.
Indeed, a containing cut gives a two-coloring of every supported edge;
conversely, choose one bipartition class from each connected component
as $S$ and assign isolated ambient vertices arbitrarily. Every supported
edge then crosses $\delta_G(S)$, irrespective of other ambient edges.
The circuit-decomposition argument needed here is given in
Supplemental Material, Sec.~\ref{app:geometry}.

\begin{theorem}[Toric distance and affine-distance parity law]
\label{thm:distance-law}
For every integer $L\ge2$, the code in \cref{eq:edge-labels,eq:checks} has
parameters $[[2L^2,2,L]]$ and
\begin{equation}
 d_X(L)=d_Z(L)=L,
 \qquad
 d_{\rm aff}(L)=
 \begin{cases}
 L,&L\text{ even},\\
 2L,&L\text{ odd}.
 \end{cases}
 \label{eq:distance-law}
\end{equation}
\end{theorem}

\begin{proof}
Regard $H_Z$ as the vertex--edge incidence matrix of the labeled dual graph.
Then $\ker H_Z$ is its cycle space, $\row H_Z$ its cut space, and $\row H_X$
spans dual-face boundaries.  Seam parities identify
$\ker H_Z/\row H_X\simeq\Ftwo^2$, and any nontrivial Eulerian support contains
a nontrivial circuit.  If a length-$w$ oriented circuit lifts to
displacement $L(p,q)$ in the square-grid universal cover, step counting gives
\begin{equation}
 w\ge L(|p|+|q|),
 \qquad
 w\equiv L(p+q)\pmod2.
 \label{eq:winding-bound}
\end{equation}
Its class is $(p\bmod2,q\bmod2)$, so $w\ge L$, attained by straight
winding-one circuits; primal--dual symmetry gives $d_Z=L$.

A cut-contained circuit must be even; conversely, alternating vertices of an
even simple circuit gives a containing cut (for the $L=2$ parallel-edge
circuit, select one endpoint).  Hence even $L$ admits a straight length-$L$
logical circuit.  For odd $L$, \cref{eq:winding-bound} makes $p+q$ even, so
nontrivial winding forces both $p,q$ odd and $w\ge2L$.  The staircase
\begin{equation}
 C_{\rm diag}(L)=\{h_{j,j},v_{j+1,j}:j\in\mathbb Z_L\}
 \label{eq:diagonal-staircase}
\end{equation}
has length $2L$, winding $(1,1)$, and lies in the cut selected by
$\{f_{j,j}:j\in\mathbb Z_L\}$.  For $L=2$, the distinct parallel dual edges
$\{h_{0,0},h_{0,1}\}$ supply a length-two witness.  Finally,
$\operatorname{rank}H_X=\operatorname{rank}H_Z=L^2-1$ completes the parameter
count.  See \cref{app:geometry} for disconnected supports.
\end{proof}

\begin{figure*}[t]
 \centering
 \includegraphics[width=0.82\textwidth]{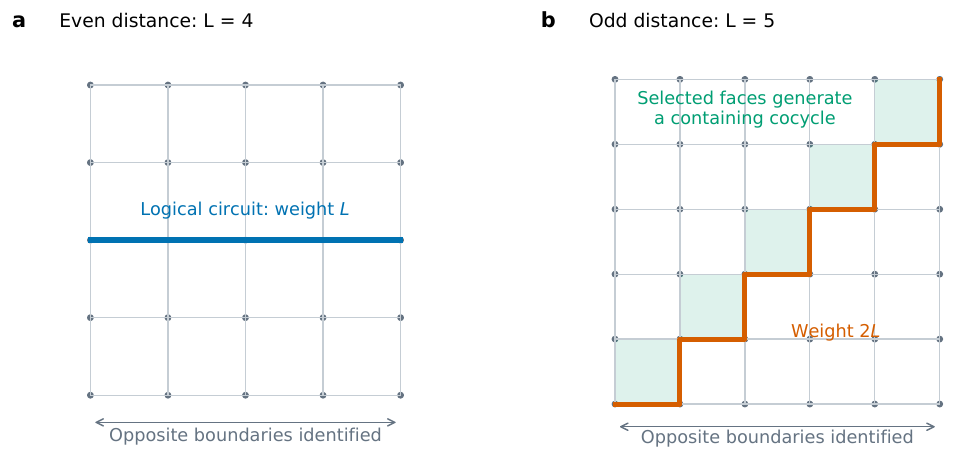}
 \caption{Geometric witnesses for \cref{thm:distance-law}.  A straight
 logical circuit has length $L$.  At odd $L$, the additional even/cut-covered
 condition forces winding in both torus directions, and the diagonal
 staircase attains length $2L$.  The drawing depicts the code's periodic topology.}
 \label{fig:geometry}
\end{figure*}

For even $L$ this obstruction is absent.  Henceforth $L$ is odd, allowing the
two orders to separate.

\subsection{Attainment and the exact leading coefficient}
\label{sec:record}
Let $P=\Xx\Xy$ and let $B_{0,P}(\theta )$ denote its coefficient in the
zero-syndrome effect after the uniform specialization $\epsilon_e=\theta $.

\begin{samepage}
\begin{theorem}[Exact syndrome-record onset]
\label{thm:record-onset}
For every fixed odd integer $L\ge3$,
\begin{align}
 b_{\rm rec}(L)&=2L,\label{eq:brec-result}\\
 [\theta ^{2L}]B_{0,P}(\theta )&=-\frac{M_L}{2^{2L-1}},\label{eq:effect-coefficient}\\
 M_L&=L\left[\binom{2L}{L}-2L\right],\label{eq:ML-result}\\
 \eta_L(\theta )&=\Theta_L(|\theta |^{2L})\qquad(\theta \to0).
 \label{eq:eta-result}
\end{align}
The $\Theta_L$ constants and the neighborhood of the origin may depend on
$L$.
Here $M_L$ is the cardinality of the distinct unrooted, unoriented,
tagged-edge simple circuits of length $2L$ that are contained in a dual cut
and have nontrivial winding.
\end{theorem}
\end{samepage}

\begin{proof}
For a Taylor multi-index $\bm a\in\mathbb N^{2L^2}$, sine/cosine parity in
\cref{eq:character-formula} forces the support for
$\bm\epsilon^{\bm a}$ to be $f_e=a_e\bmod2$.  Thus $|f|\le|\bm a|$, and a
surviving nonscalar class
requires $f$ to be a nontrivial logical support contained in a word of
$\row H_Z$.  By \cref{thm:distance-law}, $|\bm a|\ge2L$.  Thus all terms below
degree $2L$ vanish before uniform specialization.

At equality, $|f|=|\bm a|=2L$, so the monomial is squarefree and
$\bm a=f$.  Then $f$ is one even nontrivial circuit, whose winding parity
$(1,1)$ gives class $P$.  At $s=0$ all character signs are positive and
$(-i)^{2L}=-1$.  Rank--nullity gives $2^{r-2L+1}$ selectors per covered
circuit, hence contribution $-2^{-(2L-1)}$ after $2^{-r}$.

Equality in the winding bound forces $L$ monotone horizontal and vertical
steps and one of four sign pairs.  At each root their order is a balanced
binary word; exactly the $2L$ cyclic shifts of $H^LV^L$ among
$\binom{2L}{L}$ words revisit a projected vertex.  Thus there are
$4L^2[\binom{2L}{L}-2L]$ rooted oriented tuples.  Dividing by $2L$ roots and
two orientations gives \cref{eq:ML-result,eq:effect-coefficient}.

For the total-variation lower bound, the trace-one states
\begin{equation}
 \rho_\pm=\frac{I_4\pm P}{4}
 \label{eq:record-witness-states}
\end{equation}
obey
\begin{equation}
 [\theta ^{2L}]\{p_0(\theta |\rho_+)-p_0(\theta |\rho_-)\}
 =-\frac{M_L}{2^{2L-2}}\ne0.
 \label{eq:record-witness-coefficient}
\end{equation}
This gives the fixed-$L$ lower bound.  For the upper bound,
$D_s(\theta )=\theta ^{2L}R_s(\theta )$ with $R_s$ analytic and locally bounded.  For
$\Delta=\rho-\sigma$, $\Tr\Delta=0$ and
$\|\Delta\|_1\le2$, so
\begin{equation}
 \frac12\sum_s|\Tr[F_s(\theta )\Delta]|
 \le\sum_s\|D_s(\theta )\|_\infty=O_L(|\theta |^{2L}).
 \label{eq:eta-upper-short}
\end{equation}
Finitely many syndromes complete the proof; \cref{app:record} expands the
character derivation and count.
\end{proof}

For example, $M_3=42$, so the exact coefficients in
\cref{eq:effect-coefficient,eq:record-witness-coefficient} are respectively
$-21/16$ and $-21/8$.  The plotted quantities are coefficients in the small-angle Taylor expansion.  The formula gives
$M_L=(42,1210,23926,437418)$ at $L=(3,5,7,9)$.

The count and the normalized coefficient have different distance dependence:
\begin{equation}
 M_L\sim\sqrt{L/\pi}\,4^L,\qquad
 \frac{M_L}{2^{2L-1}}\sim2\sqrt{L/\pi}.
 \label{eq:normalized-count-growth}
\end{equation}
The exponential support multiplicity cancels in this particular leading
coefficient. Uniform control in $L$ of the full contrast, higher-order terms, and validity neighborhoods requires separate bounds.
\section{Logical dynamics and operational separation}
\label{sm:contents:03}
\label{sec:channel}
\subsection{The known leading coherent drift}
The corrected-channel drift has an established leading order~\cite{iverson2020coherence}. We retain its coefficient and an explicit remainder in the conventions of this instrument to compare it with the separately derived effect onset.
\begin{theorem}[Fixed-odd-$L$ corrected logical drift]
\label{thm:logical-drift}
For every fixed odd integer $L\ge3$, the instrument in
\cref{eq:kraus-definition} satisfies
\begin{equation}
 \Lambda_L(\theta )=\id-i\alpha_L\theta ^L[\cH,\,\cdot\,]+E_L(\theta ),
 \qquad
 \cH=\Xx+\Xy,
 \label{eq:channel-expansion}
\end{equation}
where
\begin{equation}
 \alpha_L=\frac{L}{2^L}\binom{L-1}{(L-1)/2}>0.
 \label{eq:alpha-result}
\end{equation}
Every coefficient of orders $1,\ldots,L-1$ vanishes.  At order $L$ there is
no Hilbert--Schmidt self-adjoint (dissipative) component and no
$\Xx\Xy$ component.  In the unhalved diamond norm, one explicit global bound
is
\begin{equation}
 \|E_L(\theta )\|_\diamond\le \kappa_L|\theta |^{L+1},
 \qquad
 \kappa_L=
 16\frac{(2L^2)^{L+1}2^{4L^2}}{(L+1)!}.
 \label{eq:channel-remainder}
\end{equation}
The remainder bound holds for every real $\theta $.  The estimates apply at fixed $L$.
\end{theorem}

The proof is given in Supplemental Material, Sec.~\ref{app:drift-main-proof}.

Reversing $e^{-i\theta X/2}$ to $e^{+i\theta X/2}$ reverses the odd-$L$ generator; removing
the half angle multiplies $\alpha_L$ by $2^L$.  These checks fix common
convention errors.  The record result remains a separate calculation.

\subsection{The complete-history bound and a nonresonant schedule}
\label{sec:histories}
\begin{proposition}[History contraction]
\label[proposition]{prop:history-bound}
For every fixed $L$, real $\theta $, and integer $R\ge1$,
\begin{equation}
 \Delta_{\rm rec}(\theta ,R)
 \le1-[1-\eta_L(\theta )]^R
 \le\min\{1,R\eta_L(\theta )\}.
 \label{eq:history-bound}
\end{equation}
\end{proposition}

\begin{proof}
At each positive-probability common prefix, separately normalize both
conditional logical states; \cref{eq:eta-definition} bounds their next-syndrome
kernels by $\eta_L(\theta )$.  At a null prefix, any normalized next-syndrome distribution completes the transition rule without changing the zero-mass law.  While histories agree, maximally
couple the next kernels.  Each step preserves agreement with probability at
least $1-\eta_L(\theta )$, so induction gives $[1-\eta_L(\theta )]^R$.  Coupling and
Bernoulli inequalities prove the bounds, without i.i.d. rounds or equal
conditional states.
\end{proof}

For this asymptotic schedule, $\tau_*$ is a dimensionless accumulated scale $R|\theta|^L$, separate from the scenario clock used in the control experiments. Combine the theorems at the scale
\begin{align}
 \mathcal G_L&=-i\alpha_L[\cH,\,\cdot\,],
 &\tau_*&=\frac{\pi}{8\alpha_L},\nonumber\\
 R_*(\theta )&=\left\lfloor\frac{\tau_*}{|\theta |^L}\right\rfloor,
 &\theta &\ne0.
 \label{eq:explicit-schedule}
\end{align}

\begin{theorem}[Same-instrument operational separation]
\label{thm:operational-separation}
For every fixed odd $L\ge3$, as $\theta \to0^\pm$, convergence in diamond
norm holds:
\begin{equation}
 \Lambda_L(\theta )^{R_*(\theta )}
 \longrightarrow
 \mathcal U_\pm:=
 \Ad_{\exp(\mp i\pi\cH/8)}
 ,
 \label{eq:channel-limit}
\end{equation}
while
\begin{equation}
 \Delta_{\rm rec}(\theta ,R_*(\theta ))=\Theta_L(|\theta |^L)\longrightarrow0.
 \label{eq:record-limit}
\end{equation}
In the unhalved diamond norm,
\begin{equation}
 \|\mathcal U_\pm-\id\|_\diamond=\sqrt2.
 \label{eq:limiting-distance}
\end{equation}
\end{theorem}

\noindent\emph{Order of limits.}  Fix finite odd $L$ before $\theta \to0^\pm$ and
$R_*(\theta )\to\infty$.  The constants $A_L,\delta_L,\kappa_L,g_L$ and the validity neighborhoods depend on $L$; the schedule $R_*(\theta )=\Theta_L(|\theta |^{-L})$ describes small-angle accumulation at fixed code size.

Equation~\eqref{eq:record-limit} concerns logical inputs at fixed $\theta $.
Separately, \cref{prop:exact-record} establishes exact blindness to the
sign of $\theta $ for every history length.  Scalar syndrome frequencies can still
reveal its magnitude or presence.

The proof is given in Supplemental Material, Sec.~\ref{app:separation-proof}.

\section{Recovery from the same syndrome record}
\label{sm:contents:04}
\label{sec:recovery}

A syndrome record can guide recovery even when it cannot distinguish particular encoded inputs.  We now allow
an additional logical correction after the declared physical recovery, while
keeping the same syndrome measurements and known signed $\theta $.  This changes the
controlled channel, not the original instrument's conventions.  No controller
receives the encoded input.  Established recovery theory gives an operational consequence of the exact instrument structure.

\subsection{Removing branch phases leaves a parity measurement}

Write $K_s=U_s\sqrt{F_s}$ in the common logical-$X$ eigenbasis, choosing
the polar factor's unitary extension diagonal also on any kernel (with
arbitrary phases on zero eigenvalues).  Applying $V_s=U_s^\dagger$ leaves the positive branch
\begin{equation}
 \begin{aligned}
 V_sK_s&=\sqrt{q_s I+b_sP}\\
       &=\sqrt{p_s^+}\,\Pi_++\sqrt{p_s^-}\,\Pi_-,
 \qquad \Pi_\pm=\frac{I\pm P}{2}.
 \end{aligned}
 \label{eq:polar-parity}
\end{equation}
These extra diagonal actions commute with the parity effects, so the
syndrome-history laws are unchanged.  This is standard conditional polar
recovery.  Gregoratti and Werner state
its optimal channel fidelity for a fixed pure-branch instrument
\cite{gregoratti2004feedback}.  In dimension four the entanglement fidelity is
\begin{equation}
 F_e(\{A_s\})=\frac1{16}\sum_s|\Tr A_s|^2.
 \label{eq:recovery-fidelity}
\end{equation}
The trace-norm bound $|\Tr(V_sK_s)|\le\Tr\sqrt{F_s}$ immediately shows
optimality among syndrome-conditioned unitary corrections.  The average pure-state infidelity is $4(1-F_e)/5$~\cite{nielsen2002average}; both quantities measure fidelity loss.

The specific consequence here is that all remaining backaction lies in the
single parity measurement.  Define
\begin{equation}
 \mathcal A_L(\theta )=\sum_s\sqrt{p_s^+(\theta )p_s^-(\theta )}.
 \label{eq:recovery-overlap}
\end{equation}
The corrected channel preserves every within-parity matrix element and
multiplies every cross-parity element by $\mathcal A_L(\theta )$.  Since both parity sectors
have dimension two,
\begin{equation}
 F_e^{\rm polar}(\theta ,R)=\frac{1+\mathcal A_L(\theta )^R}{2}.
 \label{eq:recovery-parity-fidelity}
\end{equation}
Thus every state within one parity sector is preserved exactly, although its
preparation label remains unidentifiable from the record.  Across sectors,
the overlap quantifies the disturbance left after removing reversible phases.
This expression applies established recovery theory to the exact parity effect algebra.

\subsection{Information and action classes of the comparisons}

For $K_s=\sum_c\kappa_{s,c}X_c$, the best additional logical Pauli for
one-round $F_e$ chooses $c$ maximizing $|\kappa_{s,c}|^2$ in each branch.
This is the two-logical-qubit version of coherence-aware Pauli selection
\cite{venn2020thresholds}.  The polar-versus-Pauli comparison includes the additional arbitrary-angle logical actuation available to polar recovery.

For a comparison allowing syndrome-independent terminal logical control, write $C_{xy}=\sum_s z_s(x)z_s(y)^*$ for the
original Schur channel in the logical-$X$ basis.  Let $C^{\circ R}$ denote
entrywise powers.  Every single syndrome-independent terminal unitary $V$
satisfies
\begin{equation}
 F_e(V\Lambda_L^R)\le\frac{\lambda_{\max}(C^{\circ R})}{4}.
 \label{eq:fixed-recovery-bound}
\end{equation}
Indeed only $v_x=V_{xx}$ enters the trace, and
$F_e=v^{\mathsf T}C^{\circ R}v^*/16\le
\lambda_{\max}(C^{\circ R})\|v\|^2/16$, with $\|v\|^2\le4$.
This bound includes every terminal unitary and fixed diagonal corrections
distributed over the rounds.  The optimization is restricted to the stated control class; arbitrary interleaved noncommuting controls and syndrome-dependent terminal corrections require a larger class.

That last distinction is substantive.  All $K_s$ and $V_s$ commute, so for a
complete history $\mathbf s$ the product $\prod_rV_{s_r}$ can be postponed
until the end of storage, leaving the same product of positive branches.
Per-round extra logical actuation is therefore unnecessary for final-memory
fidelity in this model.  The original minimum-weight physical recoveries
remain per round.  A terminal policy retains the history or sufficient syndrome counts and accurately accumulates calibrated phases; its protection is assessed at the final memory endpoint.

For every integer precision $b\ge1$ and storage horizon
$R\in\mathbb Z_{\ge0}$, assume accurate accumulation of the calibrated
branch phases and an ideal terminal logical action.  A single terminal
nearest-bin phase correction, using a uniform $2^b$-point grid for each
accumulated eigenphase, then has the rigorous bound
\begin{equation}
 F_e^{\rm terminal,b}(\theta ,R)\ge
 \cos^2\!\left(\frac{\pi}{2^b}\right)F_e^{\rm polar}(\theta ,R).
 \label{eq:terminal-precision-bound}
\end{equation}
For each history, the positive branch amplitudes $a_x$ acquire phase errors
$|\delta_x|\le\pi/2^b$.  Since $b\ge1$ makes
$\cos(\pi/2^b)\ge0$, we may square the nonnegative lower bound
$|\sum_x a_xe^{i\delta_x}|\ge
\cos(\pi/2^b)\sum_xa_x$ and sum histories to prove the result.
At $b=1$ the bound is trivial, and at $R=0$ the empty-history action is
identity.  
The result bounds phase quantization under the specified control assumptions; physical gate synthesis and long-history sampling require additional analysis.

Finally, the effects $F_s$ alone do not determine $V_s$.  The physical angles
$\theta $ and $-\theta $ have identical syndrome laws but opposite branch phases.
Implementing the specified phase-canceling correction therefore requires
signed-angle information beyond this passive record in general.  Pauli recovery and sign-robust control can still improve recovery, and degenerate angles can make the sign irrelevant.  The next example concerns the different task of identifying the
stored input, which remains impossible even when recovery is available.

\subsection{An unidentifiable record-only target}
\label{sec:decision}

We now specify a decision task, keeping noise detection separate from
input-dependent consequences.  A preparation device selects one of two
known logical pure states with equal prior probability and withholds the
preparation label from a monitor.  The monitor knows $L$, $\theta $, the
instrument, and the number of rounds, and receives only the complete
syndrome history.  It receives neither a final logical measurement nor any
side information about the label.  In the simultaneous logical-$X$ basis,
choose
\begin{equation}
 |a\rangle=|++\rangle,\qquad
 |b\rangle=\frac{|++\rangle+|--\rangle}{\sqrt2}.
 \label{eq:decision-inputs}
\end{equation}
Both have parity $P=+1$.  Their history laws are exactly equal by
\cref{prop:exact-record}, so every record-only classifier has success
probability $1/2$ under equal priors.

Define expected memory loss relative to the prepared state by
\begin{equation}
 \ell_j(\theta ,R)=1-\langle j|\Lambda_L(\theta )^R
 (|j\rangle\langle j|)|j\rangle,\qquad j\in\{a,b\}.
 \label{eq:decision-loss}
\end{equation}
Every $K_s$ acts as a scalar on $|a\rangle$, hence $\ell_a=0$ exactly.
On the schedule in \cref{eq:explicit-schedule},
\cref{thm:operational-separation} gives
\begin{equation}
 \ell_b(\theta ,R_*(\theta ))\longrightarrow
 1-\left|\frac{e^{-i\pi/4}+e^{i\pi/4}}{2}\right|^2
 =\frac12.
 \label{eq:decision-loss-limit}
\end{equation}
For all sufficiently small nonzero $\theta $, deciding whether the prepared state
has zero expected loss or expected loss exceeding $1/4$ therefore has
record-only error probability $1/2$.  The statement concerns state-fidelity loss averaged over the experiment.
It does not require the two candidate preparations to be orthogonal.

This example exploits exact blindness within a parity sector.  The
$\theta ^{2L}$ onset is needed for the stronger statement in
\cref{thm:operational-separation}: even the most distinguishable pair of
inputs, including opposite parity sectors, becomes indistinguishable on
the drift schedule.  A known preparation label, additional logical probes,
or a different instrument changes the decision problem.

\section{Finite-angle validation at \texorpdfstring{$L=3$}{L=3}}
\label{sm:contents:05}
\label{sec:numerics}

\subsection{Record and channel separation}

We enumerate all $2^{18}=262144$ physical $X$ supports, assign their complete
syndrome and minimum-weight recovery with the stated tie rule, and count
supports by syndrome, residual logical class, and weight.  Evaluating the finite trigonometric sums gives the finite-angle channel directly by exhaustive support enumeration.  Integer counts reproduce
$[\theta ^6]B_{0,P}=-21/16$ and the leading channel coefficient $\alpha_3=3/4$.
The enumeration procedure and numerical checks are described
in \cref{app:numerics}.

\begin{figure*}[t]
\centering
\includegraphics[width=\textwidth]{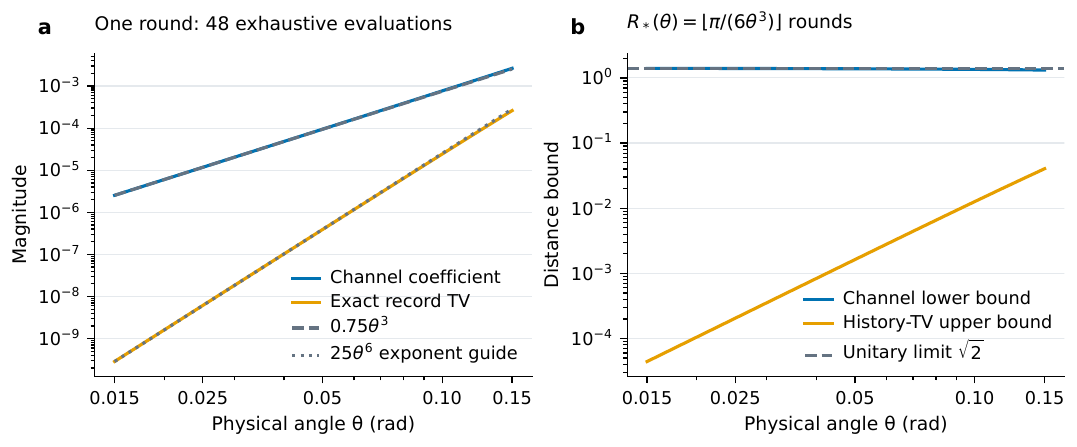}
\caption{Finite-angle separation for the 18-qubit toric instrument of Sec.~\ref{sec:model}.
(a) The channel carrier magnitude and exact maximal one-round record total-variation (TV) distance
approach orders $\theta ^3$ and $\theta ^6$, respectively.  The dotted curve is only an
exponent guide.  (b) An evaluated input-state trace distance lower-bounds the
unhalved channel diamond distance, while $R_*\eta_3$ upper-bounds the
maximal history TV.  The history curve bounds the distance from above.
All solid curves are exhaustive support-sum evaluations at 48 angles.}
\label{fig:finite-angle}
\end{figure*}

Figure~\ref{fig:finite-angle} compares the one-round response with the
repeated-round schedule.  To avoid identifying a state distance with a
channel norm, we display the explicitly computed lower bound
\begin{equation}
 d_b(\theta ,R):=\|\Lambda_3(\theta )^R(|b\rangle\langle b|)
                 -|b\rangle\langle b|\|_1
 \le\|\Lambda_3(\theta )^R-\id\|_\diamond.
 \label{eq:numeric-displacement}
\end{equation}
The history curve evaluates the rigorous upper bound $\min\{1,R_*\eta_3(\theta )\}$.
For $\theta =0.05$, $R_*=4188$, and the calculation gives
\begin{equation}
 d_b\simeq1.37934,\qquad
 \Delta_{\rm rec}(\theta ,R_*)\le R_*\eta_3(\theta )\simeq0.00162721.
 \label{eq:numeric-example}
\end{equation}
Thus the worst-case equal-prior record classification success is at most
approximately $0.500814$, while the corrected channel has order-one
unhalved diamond displacement.  For the within-parity pair in
\cref{eq:decision-inputs}, success remains exactly $1/2$ and the superposition state $|b\rangle$ has expected loss approximately $0.504346$.

The finite-angle evaluations use exhaustive floating-point sums with independent 60-digit checks; formal interval certification remains separate.  They demonstrate a regime for this
18-qubit instrument defined above, without extrapolation to larger codes or faulty
extraction circuits.  In contrast, the global analytic remainder constant
is already $\kappa_3\simeq4.81\times10^{15}$ and is used only to prove the
fixed-code asymptotic limit.

\section{From exact information to justified recovery updates}
\label{sm:contents:06}
\label{sec:bridge}
The exact calculation identifies recovery-relevant information missing from the passive record and motivates the signed-calibration construction below. Three steps
connect it to the experiments that follow.

First, select the decision target. The leakage
$\eta_L(\theta )$ compares logical inputs at a fixed physical angle. The order-$\theta ^{2L}$ onset describes logical-input sensitivity of the instrument; the recovery loss separately quantifies stale-update consequences. The record-only memory-loss
example in Sec.~\ref{sec:decision} instead makes a particular target
unidentifiable even at unlimited history length. Neither result prevents
improved recovery of an unknown stored state.

Second, specify what distinguishes the noise conditions relevant to an
action. Proposition~\ref{prop:exact-record} makes $+\theta $ and $-\theta $ exactly
equivalent under the passive experiment. Section~\ref{sec:recovery}
shows how syndrome-conditioned phases informed by signed calibration can reduce infidelity. A
larger predictor cannot extract the absent sign from identical laws.
The relevant remedy is an observation that changes the equivalence
class, or an action beneficial throughout that class. The supporting sentinel study implements a separate signed probe as one possible additional observation channel. The main demonstration instead uses a terminal logical measurement on known encoded states in the same toric instrument. Active code probes and control dithering offer alternative observation experiments.

Third, compare the proposed recovery with the incumbent over every
deployment condition still compatible with the evidence. A noise point
estimate need not be identifiable if 1 action is uniformly beneficial;
conversely, a precisely measured past angle need not justify a current
update. The evaluator uses a confidence set propagated to deployment time and a bound on $D_u$ computed separately from the proposer.

The proof and experiment regimes remain distinct. The theorem fixes
finite odd $L$ before $\theta \to0$, with
$R_*(\theta )=\lfloor\tau_*/|\theta |^L\rfloor$. The finite-angle recovery demonstration
includes $\theta =0.05$, $R=4188$, and precision-limited controls. The calibration
robustness tests instead fix $L=3$, $H=300$, and fixed phase tables,
with calibration-angle magnitudes in $[0.08,0.12]$. These use the same underlying
rotation, logical basis, and physical Pauli recovery, but evaluate different
declared operating conditions. Direct finite-angle channel calculations determine the robustness results. Commutativity allows terminal phase application in both tiers, making deferred recovery available to every comparator.

The measurement, certification, and activation workflow is shown in Supplemental Fig.~\ref{cal:fig:contract}.

\section{Measurement assumptions and conditional acceptance}
\label{sm:contents:07}
\label{cal:sec:rule}
\subsection{What must the observations distinguish?}
Proposition~\ref{prop:exact-record} establishes sign invariance for the passive record at every angle and every history length. This is separate from Theorem~\ref{thm:record-onset}, which compares logical inputs at a fixed angle. The following acceptance construction uses the sign equivalence and the recovery risks in Sec.~\ref{cal:sec:example}; logical-input sensitivity and calibration failure risk are distinct quantities.

Let $[\theta ]_{\obs}$ denote the parameter conditions with the same observation law under an allowed experiment $\obs$. A deterministic update is uniformly beneficial on that class only if
\begin{equation}
 \sup_{\theta '\in[\theta ]_{\obs}}D_u(\theta ')<0.
 \label{cal:eq:equivalence}
\end{equation}
For the declared deterministic action catalog, the operational obstruction is an empty intersection of beneficial action sets across the indistinguishable conditions. Randomized policies enlarge the action class and require a separate risk calculation, including how randomization is composed over repeated rounds. Parameter nonidentifiability alone is insufficient: an action may improve both signs, and the detailed finite-horizon risk matters. A learned model cannot resolve identical observation laws by increasing its capacity, but changing the allowed observations can remove the obstruction.

As one additional observation channel, we add a sentinel prepared in $|0\rangle$, exposed to the nominal $R_X(\theta )$ rotation, and measured in the $Y$ basis. With error-free preparation and readout, the mean is $-\sin \theta $; a direct matrix-exponential test checks this sign. For readout contrast $a$ and offset $b$, the measurement model is
\begin{equation}
 \mathbb E[Y]=b-a\sin \theta ,\qquad
 a\in[0.98,1],\quad b\in[-0.001,0.001].
 \label{cal:eq:sentinel}
\end{equation}
The code angle can differ from the sentinel angle by at most $\eta=0.001$ rad. This transfer radius $\eta$ is distinct from the logical-input contrast $\eta_L(\theta )$. The generative readout uses $a=0.99$, $b=0$. The readout and transfer bounds are declared inputs to the simulation. If $b$ is unrestricted, a sign change can be concealed by an offset change. Correct identification of an acquisition does not establish that its readout and transfer assumptions describe the device. Under the stated readout and transfer assumptions, the channel distinguishes recovery-relevant signs; finite-shot and nuisance uncertainty determine whether that information suffices for certification.

A known-state encoded probe or active syndrome dither could supply different information. The implemented sentinel supplies an additional observation channel. In particular, active self-calibration changes the observation experiment; its positive results are compatible with a passive sign ambiguity~\cite{gong2026}.

\subsection{Record integrity and the threat model}
A calibration record contains the measurements used to estimate noise and their acquisition time. The allowed recovery operations form a fixed catalog. Certification means that the evaluator bounds the proposed operation's excess infidelity below $-\delta$ throughout the deployment confidence set. The software checks below preserve the association between those measurements, the experiment, and the operation actually applied.

The evaluator owns an immutable registry containing an acquisition identifier, workload identifier, request nonce, plus-outcome count, shot count, and acquisition time. The candidate is an integer indexing the fixed table. The evaluator rejects requests with unknown records, mismatched workload or nonce identifiers, future timestamps, malformed actions, or exhausted budgets. The proposer can select an action or evidence identifier and can request a new probe acquisition within its budget. It cannot insert records, alter timestamps, change the reference channel, or activate a table directly (Fig.~\ref{cal:fig:contract}).

The attacker can supply misleading recommendations, including an instruction to ignore probe sign, or induce the use of authentic old evidence. Separate tests submit mismatched workload and nonce identifiers. The primary physical drift remains inside a publicly declared rate bound; an explicit out-of-bound challenge tests that premise. The registry's trust boundary assumes an uncompromised process, reliable acquisition electronics, and physical evolution within the declared model. The prototype implements control of evidence creation and modification within one evaluator process; cryptographic authentication, attestation, and post-quantum deployment protocols are separate implementation tasks.

\subsection{From finite measurements to a deployment set}
Let $k$ plus outcomes be observed in $n$ independent stationary sentinel shots. Form a two-sided Clopper--Pearson interval $[p_-,p_+]$ for the Bernoulli probability~\cite{clopper1934}. With $m=2p-1$, invert Eq.~\eqref{cal:eq:sentinel} through \emph{all} allowed $a,b$ to obtain a calibration-angle interval $C_{\rm cap}$. The sine response is one-to-one on the domain $|\theta |\leq0.18$ for fixed contrast and offset. Invert the probability interval at every corner of the allowed contrast--offset rectangle. Take the smallest angle interval containing the feasible endpoints and intersect it with the physical domain. An empty interval is a model-inconsistency rejection.

For evidence age $A$ and declared angular drift rate $v$, the stationary-deployment set is
\begin{equation}
 C_{\rm dep}=C_{\rm cap}\oplus[-(\eta+vA),\eta+vA].
 \label{cal:eq:deployment}
\end{equation}
The evaluator rejects a proposal if its deployment region extends beyond the supported angle domain. Define a bound $U_u$ over this entire set and accept only when
\begin{equation}
 U_u\geq\sup_{\theta \in C_{\rm dep}}D_u(\theta ),\qquad
 U_u\leq-\delta,\quad \delta=10^{-3}.
 \label{cal:eq:accept}
\end{equation}
The reference is independent of the model used to propose $u$. Learned uncertainty guides acquisition, while the confidence set supplies the activation criterion. Rejection leaves the incumbent in place; a workflow can subsequently probe again if a separately budgeted acquisition is available.

\begin{proposition}[Conditional protection of accepted updates]
\label{cal:prop:conditional}
Consider a finite declared family of acquisition looks $j$, with simultaneous deployment sets whose failure probabilities are at most $\alpha_j$ and $\sum_j\alpha_j\leq\alpha$. Suppose the physical drift, readout, and transfer premises hold; each computed bound is valid for every allowed action; and the activated action is exactly the evaluated action. Then, for any proposer selecting actions and available looks adaptively,
\begin{equation}
 \Pr\{\exists\text{ accepted }(j,u):D_u(\theta _j)>-\delta\}
 \leq\alpha.
 \label{cal:eq:family}
\end{equation}
\end{proposition}
\begin{proof}
On the event that every deployment condition lies in its set, Eq.~\eqref{cal:eq:accept} implies $D_u(\theta _j)\leq U_u\leq-\delta$ for every admitted action. This event holds uniformly over the action catalog and therefore also for a data-dependent choice from it. Its complement has probability at most $\sum_j\alpha_j$ by the union bound.
\end{proof}

For histories that determine a new acquisition policy, the corresponding coverage must hold conditionally on that history, or a simultaneous construction must cover all permitted choices. The experiments use predetermined looks or a budget selected before independent newly acquired measurements. Repeated queries to the \emph{same} simultaneous physical set do not require a new failure-probability allocation; new acquisition looks do. Acquisition is restricted to the declared, statistically budgeted looks.

Proposition~\ref{cal:prop:conditional} is a standard confidence-set argument applied to this QEC contract. It bounds the unconditional margin-violation probability over the declared experiment family. The protected quantity is improvement relative to the incumbent; an absolute logical-risk ceiling would require an additional criterion. The analytical statement requires a valid numerical enclosure. The implementation supports that requirement with a tested $10^{-8}$ allowance; the deployment scope is discussed in Supplemental Material, Sec.~\ref{sec:completion}. Supplemental Material, Sec.~\ref{cal:app:numerics} specifies this distinction.

\subsection{Age, acquisition duration, and deployment duration}
Equation~\eqref{cal:eq:deployment} assumes stationary capture and a stationary deployment block. Finite capture duration, clock error, within-block drift, and correction cost need additional allowances. For time-dependent angles, a conservative channel-distance argument gives the excess-risk allowance
\begin{equation}
 B_{\rm path}=Nv\Delta H(H+1),
 \label{cal:eq:path}
\end{equation}
when $|\theta _r-\theta _{\rm cap}|\leq vr\Delta$. Supplemental Material, Sec.~\ref{cal:app:drift} derives it and reports an 81-rate sweep with direct composition of changing channels. The extension incorporates execution duration into the freshness criterion under the stipulated physical bounds.

\section{Robust acceptance under misleading and stale advice}
\label{sm:contents:08}
\label{cal:sec:security}

\subsection{Mechanism stress tests and baselines}
We compare authorization alone, authorization with a fixed time-to-live (TTL), confidence without age propagation, and the full confidence-and-freshness rule. All gates use the same simulated episodes, counts, proposed actions, and incumbent. Authorization verifies that the record was issued by the evaluator and matches the workload and request; the risk certificate assesses physical usefulness. The confidence-only rule includes readout and transfer uncertainty but treats capture evidence as current. Fixed TTLs are $0,0.1,0.3,1,3$ in units $T_0$ of this supporting scenario. No gate and the complete TTL results remain in the research record.

Here the label \emph{fresh} denotes the acquisition-time condition, before the tested staleness or delay; it does not by itself certify an update. Each of 7 cases contains 200 independent simulated noise paths, with capture magnitude uniform on $[0.08,0.12]$, random sign, 8192 probe shots, and $H=300$. The proposer maps measurement-based advice to the fixed catalog; hidden angles are used only to evaluate loss and configure the physical challenge. The selected conditions make honest proposals reduce infidelity, so rejection has a measurable opportunity cost. These mechanism stress tests measure harm and acceptance within the specified challenge distributions.

In the bounded-stale case, the true sign reverses after age $0.5T_0$ under the public bound $v=0.6$ rad/$T_0$. The largest required change is $0.24$ rad, below the permitted $0.3$ rad. Benign delay instead uses age $2T_0$, $v=0.001$ rad/$T_0$, and no change in the noise angle. These cases have \emph{different public drift contexts}. A globally fixed expiry can sacrifice utility across different drift contexts; context-specific tuning is a relevant alternative. The drift-bound violation challenge reverses the sign after age $0.01T_0$ while declaring $v=0.001$ rad/$T_0$, and consequently lies outside the guarantee.

\begin{table}[!htbp]
\caption{Accepted/harmful updates out of 200 paths per constructed mechanism stress test. Harm means positive excess infidelity from numerical evaluation of the exact channel. Both confidence-only and full rules include authorization. The drift-bound violation violates the drift premise.}
\label{cal:tab:security}
\centering\small\setlength{\tabcolsep}{3pt}
\begin{tabular}{lrrr}
\toprule
Case & Authorization & Confidence & Full\\\midrule
Fresh honest & 200/0 & 190/0 & 190/0\\
Wrong advice & 200/200 & 0/0 & 0/0\\
Bounded stale & 200/200 & 178/178 & 0/0\\
Benign delay & 200/0 & 182/0 & 168/0\\
Wrong workload & 0/0 & 0/0 & 0/0\\
Wrong nonce & 0/0 & 0/0 & 0/0\\
Drift bound violated & 200/200 & 187/187 & 187/187\\
\bottomrule\end{tabular}

\end{table}

Confidence alone blocks misleading fresh recommendations but admits 178 harmful stale updates. The full rule admits none in that bounded-stale case (Table~\ref{cal:tab:security}). It accepts 190 honest fresh updates and 168 under benign delay, all beneficial in the numerical channel evaluation. The pointwise acceptance proportions are 95\% and 84\%; the latter's two-sided 95\% binomial interval is approximately $[78.2\%,88.8\%]$. The interval's lower endpoint falls below the study's stated 80\% acceptance target in this regime. Zero observed harms in 200 independent trials of one fixed challenge distribution gives a one-sided 95\% binomial upper limit of approximately 1.49\% for that distribution. It is not a bound for arbitrary attacks or different workload distributions.

The number of accepted updates does not specify the size of their recovery benefit. Define retained recovery gain relative to authorization-only deployment on the same honest paths:
\begin{equation}
 G=\frac{\sum_i[\risk(\theta _i,u_0)-\risk(\theta _i,u_{\rm full,i})]}
 {\sum_i[\risk(\theta _i,u_0)-\risk(\theta _i,u_{\rm auth,i})]}.
 \label{cal:eq:gain}
\end{equation}
Rejected updates use $u_0$ in the numerator. The full gate retains \FreshGain\ of fresh-evidence gain, with paired path-bootstrap 95\% interval $[94.9\%,98.9\%]$, and \DelayedGain\ under benign delay, with interval $[86.8\%,94.2\%]$. This is the mechanism study's utility result: rejecting uncertain updates preserves most of the aggregate improvement in these tested distributions. The retained fraction depends on the workload distribution.

\begin{figure*}[t]\centering\includegraphics[width=\textwidth]{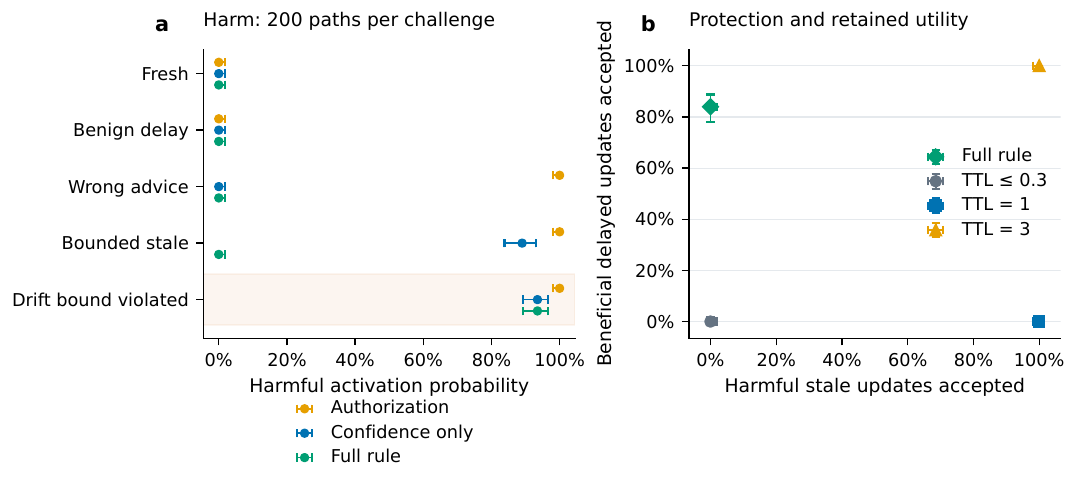}\caption{Harm and utility in paired mechanism stress tests. (a) Harmful-activation fractions and pointwise exact 95\% binomial intervals, 200 paths per case; the shaded case violates the drift premise. (b) The same stale-harm and benign-delay acceptance endpoints for the full rule and fixed expiry, with pointwise intervals on both axes. Time-to-live (TTL) values 0, 0.1, and 0.3 coincide. The benign-delay and bounded-stale challenges use different declared drift contexts.}\label{cal:fig:security}\end{figure*}

\subsection{Failure when the drift bound is exceeded}
The out-of-bound step causes 187 harmful activations out of 200 under the
full rule, exactly as under confidence-only acceptance. The acquisition
can be correctly authorized while its physical interpretation is wrong.
The extended readout, transfer, and age sweeps in
Supplemental Material, Sec.~\ref{app:validity-stress} expose this same boundary. They report
5400 evaluations of condition--replicate pairs satisfying the assumptions with no observed harmful activation and
4767 harmful activations among 22600 evaluations of condition--replicate pairs violating the assumptions.
Reused count blocks make these totals unsuitable as independent binomial
trials. These failures locate the physical boundary of the guarantee.

\section{Recovery improvement, workflow completion, and transfer}
\label{sm:contents:09}

\label{sec:automation-summary}
The acceptance tests establish what evidence supports an action. We next assess whether the controllers complete calibration and activation workflows that improve recovery. Detailed workflows and all secondary outcomes remain in
Supplemental Material, Secs.~\ref{cal:sec:llm}--\ref{cal:sec:surface}.

The controller studies distinguish beneficial recovery from workflow completion. A beneficial completion deploys an operation with excess risk below $-\delta$; harmful activation has excess risk above $\delta$. Retention and incomplete workflows remain in the denominators. The expanded multistep evaluation in Supplemental Material, Sec.~\ref{sec:completion} measures beneficial completion and calibration cost under misleading advice, limited initial data, stale evidence and identity failures. Its matched deterministic comparator receives the same observations and permitted actions. Early interface pilots are summarized in Sec.~\ref{cal:sec:llm}.

\subsection{Learned recovery and a Bayesian comparison}
Learned quantum control also has hardware implementations: Xu et al.~\cite{xu2022neuralcontrol} implement a neural predictor of pulse parameters on a field-programmable gate array (FPGA). This connects learned control decisions to their implementation cost; the recovery comparison below evaluates predictors through the same calibrated acceptance rule.

The 5 gated recurrent units (GRUs) and 5 multilayer perceptrons (MLPs) are trained on signed-probe histories, with whole
simulated noise paths and parameter blocks separated across training,
development, and test. The learned and classical references have the same
observations, recovery catalog, and acquisition budgets. Across 600
held-out paths, the guarded GRU's mean infidelity is \GRURisk, compared
with \RetainRisk\ for retaining the incumbent and \BayesRisk\ for a
development-tuned Bayesian filter. The paired 95\% interval $[-0.002675,+0.000070]$ leaves the neural-versus-Bayesian comparison inconclusive. Both the learned predictor and Bayesian filter improve recovery relative to retention.

These models receive signed-probe information; the passive-record impossibility follows analytically from identical observation laws for every postprocessing architecture. In a separate acquisition study with fixed checkpoints,
Bayesian allocation reduces mean newly acquired calibration shots by \ShotReduction\ while
mean infidelity increases from 0.0365583 to 0.0372372. Learned allocation
traces a similar cost--risk frontier. Supplemental Material, Sec.~\ref{cal:sec:learning} reports the risk--cost tradeoff, common historical shot cost, and changed confidence allocation.

\subsection{Acceptance at calibration in the profile-changing circuit challenge}
A separate stochastic surface-code study uses Stim and PyMatching at
distances 3 and 5 with 30 rounds and paired scoring of matching
priors on the same records~\cite{gidney2021,higgott2025}. All 400 candidates at the calibration noise settings pass the calibration rule; simultaneous confidence intervals for both memory-basis risks of candidate and incumbent
establish benefit in 346 cases, leaving 54 unresolved. Under stale
deployment, all 400 individual effects remain unresolved. The stale gate rejects every candidate because its valid bound is vacuous, so this profile-changing challenge supports acceptance at the calibration settings. These stochastic Pauli circuits form a separate test of acceptance; their noise model differs from the coherent toric instrument. Supplemental Material, Sec.~\ref{cal:sec:surface} reports the pilot, confirmation, and limitations of these results.

\renewcommand{\appendixname}{Supplemental Material}
\appendix
\renewcommand{\thesection}{\SupplementLetter{\value{section}}}
\crefalias{section}{appendix}
\crefname{appendix}{Supplemental Material, Sec.}{Supplemental Material, Secs.}
\section{Channel proofs and the general privacy reference}
\label{sm:contents:10}
\label{app:dynamics}
\subsection{Proof of the corrected-channel expansion}
\label{app:drift-main-proof}
\begin{proof}
Each $\theta ^m$ coefficient of $\kappa_{s,c}$ is $(-i)^m$ times a real number.  Let
$w(s)=\min\{|a|:\widetilde H_Za=s\}$.  A nonidentity class cannot occur below
$\max\{w(s),L-w(s)\}$.  Since $L$ is odd, a product of two nonidentity Kraus
coefficients starts at order at least $L+1$.  Trace preservation makes the
identity--identity channel coefficient equal to $1+O_L(\theta ^{L+1})$.  Thus
order-$L$ terms only pair identity and nonidentity amplitudes of orders $w(s)$
and $L-w(s)$.  Their purely imaginary common phase makes a commutator, not a
dissipator.

Every nontrivial weight-$L$ cycle is one of $L$ straight representatives of
class $(1,0)$ or $(0,1)$; the product class has minimum weight $2L$.  Row or
column projection shows that for $A$ on such a line with $|A|>L/2$, its
complement is the unique global minimum recovery for its boundary.  Each line
therefore gives the odd repetition-code majority sum
\begin{equation}
 -\frac{i}{2^L}\binom{L-1}{(L-1)/2},
 \label{eq:per-line-coefficient}
\end{equation}
using
$\sum_{k=(L+1)/2}^L(-1)^k\binom Lk
=(-1)^{(L+1)/2}\binom{L-1}{(L-1)/2}$.
Multiplying by $L$ lines per direction yields
\cref{eq:channel-expansion,eq:alpha-result}.  Iverson and Preskill already
obtained this contribution from each logical line and the magnitude of the summed coefficient, up
to logical labels and signs, in this setting \cite{iverson2020coherence}; their
Appendix~H also identifies weight-$2L$ double-logical paths.  We reproduce the established onset and coefficient in common conventions to compare the channel with its effects.  Related repetition-code reductions and
minimum-logical-support perturbative expansions appear in coherent-error analyses
\cite{greenbaum2018modeling,huang2019coherent,rajmohan2026correlated}.

Finally, each channel class-matrix coefficient is a finite trigonometric
polynomial whose integral Taylor remainder is bounded by
$(2L^2)^{L+1}2^{4L^2}|\theta |^{L+1}/(L+1)!$.  The $16$ left--right logical-$X$
carriers have diamond norm one because unitary multiplication preserves
singular values, giving \cref{eq:channel-remainder}; \cref{app:drift} gives the
order squeeze and line projection.
\end{proof}
\subsection{Proof of the nonresonant operational separation}
\label{app:separation-proof}
\begin{proof}
Put $h=|\theta |^L$ and $\sigma=\sgn(\theta )$.  By \cref{thm:logical-drift},
$\Lambda_L(\theta )=\id+\sigma h\mathcal G_L+E_L(\theta )$.  Compare one round with
$e^{\sigma h\mathcal G_L}$.  The integral exponential remainder and
\cref{eq:channel-remainder} give
\begin{equation}
 \|\Lambda_L(\theta )-e^{\sigma h\mathcal G_L}\|_\diamond
 \le\kappa_L|\theta |^{L+1}
 +\tfrac12h^2g_L^2e^{hg_L},
 \; g_L=\|\mathcal G_L\|_\diamond.
 \label{eq:one-step-product-bound}
\end{equation}
Both maps are channels.  A mixed telescoping sum over
$R=\lfloor\tau/h\rfloor$
rounds, together with the error from rounding the round count down to an integer, gives, uniformly for $0\le\tau\le T$,
\begin{multline}
 \left\|\Lambda_L(\theta )^{\lfloor\tau/h\rfloor}
 -e^{\sigma\tau\mathcal G_L}\right\|_\diamond\\
 \le \kappa_LT|\theta |
 +\frac{T}{2}g_L^2e^{g_Lh}h+g_Lh.
 \label{eq:product-bound}
\end{multline}
The right side vanishes at fixed $L,T$; $\tau=\tau_*$ proves
\cref{eq:channel-limit}.

By \cref{thm:record-onset}, fixed-$L$ constants $A_L,\delta_L>0$ satisfy
$\eta_L(\theta )\le A_L|\theta |^{2L}$ for $|\theta |\le\delta_L$.  Hence
\begin{equation}
 \Delta_{\rm rec}(\theta ,R_*)
 \le R_*\eta_L(\theta )
 \le A_L\tau_*|\theta |^L,
 \label{eq:record-schedule-bound}
\end{equation}
proving the upper bound in \cref{eq:record-limit}. The matching lower bound is established below in Sec.~\ref{sec:history-lower-bound}.

For unitary channels, the unhalved norm obeys
\begin{equation}
 \|\Ad_U-\Ad_V\|_\diamond
 =2\bigl[1-\nu(U^\dagger V)^2\bigr]^{1/2},
 \label{eq:unitary-channel-distance}
\end{equation}
where $\nu(W)$ is the minimum modulus in the convex hull of $W$'s spectrum.
With an arbitrary ancilla, the output overlap is $\Tr(\rho U^\dagger V)$ for
some density operator $\rho$; these values fill that convex hull, and the
pure-state trace-distance formula proves \cref{eq:unitary-channel-distance}.
Here $\cH$ has spectrum $\{2,0,0,-2\}$, so $e^{-i\pi\cH/8}$ has spectrum
$\{e^{-i\pi/4},1,1,e^{i\pi/4}\}$ and convex-hull minimum modulus $1/\sqrt2$,
giving distance $\sqrt2$ from identity.
\end{proof}

\subsection{A rare syndrome attains the history exponent}
\label{sec:history-lower-bound}
Fix odd $L\ge3$. Split the simple diagonal staircase in Eq.~\eqref{eq:diagonal-staircase} into its two alternating matchings
\begin{equation}
 E_0=\{h_{j,j}:j\in\mathbb Z_L\},\qquad
 E_1=\{v_{j+1,j}:j\in\mathbb Z_L\}.
 \label{eq:rare-matchings}
\end{equation}
They have the same syndrome $s_*$, supported on the $2L$ distinct vertices of the dual cycle. An error edge can meet at most two of these odd vertices, so every support with this syndrome has weight at least $L$. Both matchings attain that weight. Their symmetric difference has winding $(1,1)$ and hence logical class $P$.

Let $r_{s_*}$ be the declared minimum-weight recovery, necessarily of weight $L$. The difference between it and any weight-$L$ support has even cardinality. For odd $L$, the logical classes $(1,0)$ and $(0,1)$ have odd cardinality, so the leading corrected terms lie only in classes $I$ and $P$. Let $a$ and $b$ count these weight-$L$ terms. Both are positive: the two matchings differ by $P$. All weight-$L$ amplitudes have the same phase, and multiplication by the $X$ recovery introduces no Pauli sign. Consequently,
\begin{align}
 K_{s_*}(\theta)&=(-i\theta/2)^L(aI+bP)+o_L(|\theta|^L),\nonumber\\
 p_{s_*}^\pm(\theta)&=\frac{(a\pm b)^2}{2^{2L}}|\theta|^{2L}
                      +o_L(|\theta|^{2L}),\nonumber\\
 p_{s_*}^+(\theta)-p_{s_*}^-(\theta)
     &=c_L|\theta|^{2L}+o_L(|\theta|^{2L}),
 \quad c_L=\frac{4ab}{2^{2L}}>0.
 \label{eq:rare-probability}
\end{align}
These expansions hold from either angle sign. In a fixed parity sector the history law is a product law. For the event $B_R$ that $s_*$ occurs at least once,
\begin{equation}
 |P_+^{(R)}(B_R)-P_-^{(R)}(B_R)|
 =|(1-p_{s_*}^-)^R-(1-p_{s_*}^+)^R|.
 \label{eq:rare-event}
\end{equation}
At $R=R_*(\theta)$, both $Rp_{s_*}^\pm=O_L(|\theta|^L)$ tend to zero. The mean-value theorem applied to $(1-p)^R$ therefore gives
\begin{equation}
 |P_+^{(R_*)}(B_{R_*})-P_-^{(R_*)}(B_{R_*})|
 =\tau_*c_L|\theta|^L[1+o_L(1)].
 \label{eq:history-lower-bound}
\end{equation}
Total variation is at least the difference on any event. Choosing logical inputs in opposite parity sectors proves the lower bound in Eq.~\eqref{eq:record-limit}; Eq.~\eqref{eq:record-schedule-bound} supplies the upper bound. This result concerns stored-input distinguishability at fixed noise. Exact blindness to the sign of the noise is unchanged.

As implementation checks, enumeration of minimum supports at $L=3,5,7,9,11$ gives positive counts in both classes, and exhaustive weight-three support enumeration at $L=3$ agrees with the matching construction. For the checked $L=3$ syndrome, $a=b=1$ and $c_3=1/16$. The event contrast divided by $|\theta|^3$ is $0.03205$, $0.03255$, and $0.03268$ at $\theta=0.04,0.02,0.01$, approaching $\tau_*/16=0.032725$. These finite checks support the all-odd-$L$ argument above.

\subsection{The general complementary-channel baseline}
\label{sec:privacy-baseline}

The qualitative vanishing in \cref{eq:record-limit} already follows from
correctability--privacy duality.  To apply it with the same norms, put
$\mathcal N_\theta =\Lambda_L(\theta )^{R_*(\theta )}$ and let $\mathbf s$ range over complete
histories, with $K_{\mathbf s}=K_{s_{R_*}}\cdots K_{s_1}$.  A dilation and
its complementary channel are
\begin{align}
 W_\theta &=\sum_{\mathbf s}K_{\mathbf s}\otimes|\mathbf s\rangle,
 &W_\theta ^\dagger W_\theta &=I_4,\label{eq:history-dilation}\\
 \mathcal N_\theta ^c(\rho)&=\sum_{\mathbf s,\mathbf u}
 \Tr(K_{\mathbf s}\rho K_{\mathbf u}^\dagger)
 |\mathbf s\rangle\langle\mathbf u|.\nonumber
\end{align}
Dephasing the history register gives exactly the classical record channel
$\mathcal M_\theta (\rho)=\sum_{\mathbf s}P_{\rho,\theta }^{(R_*)}(\mathbf s)
|\mathbf s\rangle\langle\mathbf s|$.
By \cref{eq:product-bound}, the unhalved diamond error
$\epsilon_\theta =\|\mathcal N_\theta -\mathcal U_\sigma\|_\diamond
=O_L(|\theta |)$, where $\sigma=\sgn(\theta )$.  Inverting $\mathcal U_\sigma$ is
therefore an $\epsilon_\theta $-accurate recovery.  Theorem~5 of
Kretschmann, Kribs, and Spekkens \cite{kretschmann2008complementarity}
and contractivity under dephasing give a constant classical channel
$\mathcal C_\theta (\rho)=\Tr(\rho)\omega_\theta $ satisfying
\begin{equation}
 \|\mathcal M_\theta -\mathcal C_\theta \|_\diamond\le2\sqrt{\epsilon_\theta }.
 \label{eq:generic-privacy-channel}
\end{equation}
For any pair of inputs, the triangle inequality contributes two such
errors, and the factor $1/2$ in total variation cancels that factor two:
\begin{equation}
 \Delta_{\rm rec}(\theta ,R_*)\le2\sqrt{\epsilon_\theta }
       =O_L(\sqrt{|\theta |}).
 \label{eq:generic-privacy-tv}
\end{equation}
Thus the general bound shows that logical-input distinguishability vanishes in this limit even during nontrivial unitary evolution.  The toric effect calculation improves this guarantee
to \cref{eq:record-schedule-bound}, constraining the success advantage of
\emph{every} equal-prior binary history classifier to
$\Delta_{\rm rec}/2=O_L(|\theta |^L)$.  This sharper bound supplies a quantitative
reference for monitoring at the coherent accumulation time.  The rare-syndrome construction in Sec.~\ref{sec:history-lower-bound} attains the exponent at this schedule. The result fixes the small-angle order of the maximum contrast; it does not determine its leading constant or a generally optimal sample schedule.

\begin{remark}[Why the schedule is explicit]
\label{rem:resonance}
For any sequence approaching zero with eventually fixed sign $\sigma$ and
$R(\theta )|\theta |^L\to\tau<\infty$, the same proof gives
$e^{\sigma\tau\mathcal G_L}$.  Resonant values exist.  At
$\tau=m\pi/\alpha_L$, $m\in\mathbb Z_{\ge0}$, the limiting channel is identity.
Thus $R=\Theta_L(|\theta |^{-L})$ alone is insufficient;
$\tau_*=\pi/(8\alpha_L)$ is a specified nonresonant choice.
\end{remark}

\section{Finite-precision recovery comparisons}
\label{sm:contents:11}
\label{app:precision-results}
\subsection{A controlled recovery comparison}
\label{sec:recovery-numerics}

Using the same exhaustive instrument, we evaluate entanglement fidelity at
$\theta =0.05$ and $R=4188$, the reference finite-angle example.  Table~\ref{tab:recovery}
compares evaluated policies and explicitly labeled bounds.  All receive the
same signed calibration.  At this angle, one-round optimal Pauli selection
equals minimum-weight recovery, so there is no Pauli-only improvement at the
anchor.  Coherent feedback improves recovery, while the fixed-unitary
comparison isolates a benefit from conditioning on the syndrome.

\begin{table*}[t]
\caption{Entanglement infidelity at $L=3$, $\theta =0.05$. Entries labeled as bounds are analytical comparators. The fixed controller is tuned to the storage horizon, while the Pauli policy optimizes one-round fidelity. Extra logical quantum actuation is ideal except for the stated phase discretization.}
\label{tab:recovery}
\centering
\begin{tabular}{lrr}
\hline
Recovery or comparator & One round & 4188 rounds \\
\hline
Minimum-weight Pauli & $7.26443\times10^{-6}$ & $0.291997$ \\
One-round optimal Pauli & $7.26443\times10^{-6}$ & $0.291997$ \\
Fixed terminal unitary: lower bound & $7.24664\times10^{-6}$ & $0.0296724$ \\
Fixed controller tuned to storage horizon & --- & $0.0296724$ \\
Exact conditional polar recovery & $4.221\times10^{-8}$ & $0.000176744$ \\
Per-round nearest-bin 12-bit feedback & $4.81288\times10^{-8}$ & $0.0685292$ \\
Per-round randomized 12-bit feedback & $9.55821\times10^{-8}$ & $0.000400171$ \\
Terminal 12-bit correction: upper bound & --- & $0.000177333$ \\
\hline
\end{tabular}
\end{table*}

Nearest-bin 12-bit feedback reduces one-round infidelity by a factor of approximately
$151$ relative to minimum-weight recovery, yet accumulated phase bias makes
it worse than the implemented fixed controller at the declared storage
horizon.  An exploratory extension randomizes independently between adjacent
phase bins with unbiased phase expectation.  Summing all 16 possible
quantized diagonal actions per syndrome gives the exact ensemble-average
channel.  Its final-memory infidelity is approximately $4.0017\times10^{-4}$,
smaller by a factor of about $74$ than that of the implemented fixed controller.  Figure~\ref{fig:recovery-main}
shows both the improvement and the failed deterministic-rounding policy.
The selection probabilities and classical arithmetic use float64; only the
realized phases are 12-bit.  The guarantee averages over the declared randomization resource; individual rounding sequences can have different errors.

History-aware deferred control remains a relevant alternative to online feedback.  With the extra correction postponed until the
end, \cref{eq:terminal-precision-bound} gives 12-bit infidelity below
$1.774\times10^{-4}$, already smaller than the per-round randomized result.
The terminal entry evaluates a rigorous analytical upper bound from the exact overlap.  It assumes accurate
classical phase accumulation and an ideal final diagonal logical action.

\begin{figure*}[t]
 \centering
 \includegraphics[width=\textwidth]{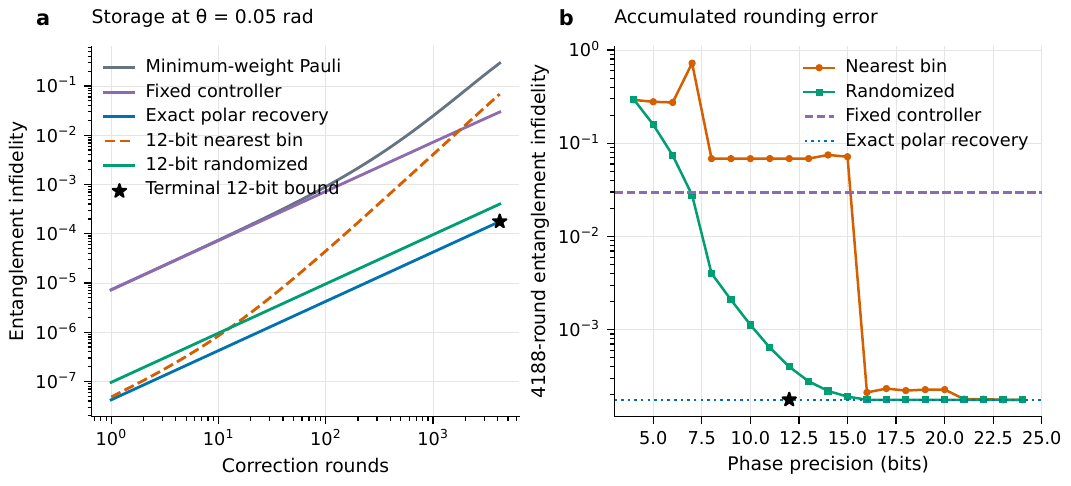}
 \caption{Recovery and accumulated control precision for the exact $L=3$
 instrument at $\theta =0.05$.  (a) Final-memory entanglement infidelity versus
 rounds.  The fixed diagonal controller is tuned to $R=4188$; it is an
 implemented feasible policy, close to the terminal-unitary bound there.
 (b) The complete integer precision sweep at the same horizon.  Ordinary
 rounding leaves coherent bias; randomized rounding trades it for dephasing.
 The terminal 12-bit entry reports the analytical guarantee.  The comparison assumes ideal extraction and signed calibration; logical-gate implementation costs would add to the reported resources.}
 \label{fig:recovery-main}
\end{figure*}

Wrong-sign calibration reverses the benefit: the randomized controller's
memory infidelity rises to approximately $0.7521$, versus $0.2920$ without
extra feedback.  Under an explicit additional complete-logical-depolarization
model after every feedback round, its one-round advantage over the Pauli
baseline survives only below a probability of complete logical depolarization of approximately
$7.65\times10^{-6}$ per whole logical action.  The stringent budget applies to repeated actuation; deferred correction has a different resource requirement.  The angle, calibration, precision, and
actuation-noise sweeps, including failures, are retained in
\cref{app:recovery-protocol}.  These calculations illustrate recovery within the stated instrument and identify the control resources needed for its benefit.

\section{Extended validity tests}
\label{sm:contents:12}

\label{app:validity-stress}
\subsection{Where the guarantee fails}
The drift-bound violation produces 187 harmful activations under both the confidence-only and full rules. No software binding failure is required: the physical validity premise is false. An additional age--drift grid and readout/transfer stress study expose the same boundary (Fig.~\ref{cal:fig:validity}). Across 5400 evaluations of condition--replicate pairs satisfying the assumptions there are no observed harmful activations; across 22600 evaluations of condition--replicate pairs violating the assumptions there are 4767. Each count block is shared across 5 transfer values; statistical uncertainty must account for that grouping.

The offset curves integrate every possible count against its binomial probability at 201 offsets. The count distribution is summed exhaustively for the selected model and nuisance settings. Floating-point evaluation remains the computational uncertainty; Monte Carlo sampling intervals do not apply to this deterministic sum. Deployment requires a justified physical interpretation of the observations in addition to an authentic record.

\begin{figure*}[t]\centering\includegraphics[width=\textwidth]{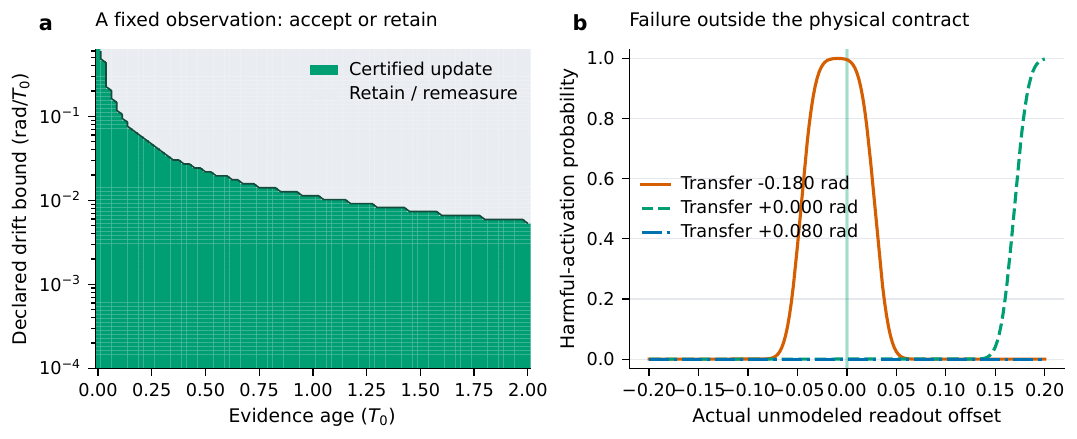}\caption{Physical validity of the acceptance contract. Left: the acceptance region on an 81-by-81 grid of evidence ages and declared drift rates for one fixed observed block. Right: exhaustively summed binomial probabilities of harmful activation across 201 simulated offsets and selected transfer mismatches. The shaded offset band is the declared interval. Transfer outside $\pm0.001$ rad violates the contract even at zero offset.}\label{cal:fig:validity}\end{figure*}

\section{Language-model maintenance and preliminary interface tests}
\label{sm:contents:13}
\label{cal:sec:llm}
The multistep controller may request a new probe acquisition, propose a recovery with its evidence identifier, or retain the incumbent recovery. The evaluator lists actions supported by the observations and rechecks each proposal at activation. The adviser can select from this menu but cannot change the reference calculation. The expanded evaluation in Supplemental Material, Sec.~\ref{sec:completion} supplies the principal multistep evidence, including incomplete workflows, acquisition cost and a matched deterministic baseline.

Earlier interface pilots exposed repeated rejected advice and wrong-sign proposals after new calibration measurements. A revised interface achieved 23 beneficial guarded completions in 24 episodes, with 1 incomplete workflow; authorization alone admitted 6 harmful updates, and the deterministic comparator completed all 24 beneficially. The two interfaces used different physical seeds, so this is not a paired causal estimate of an interface improvement. The initial and subsequent pilots contain 24 and 48 model interactions, respectively. These pilots motivate the interface but are not additional independent evidence for the central toric experiment.

\section{Learning, classical references, and acquisition cost}
\label{sm:contents:14}
\label{cal:sec:learning}
The same observation and acceptance rules also apply to learned predictors. We train 5 gated recurrent units (GRUs) and 5 multilayer perceptrons (MLPs) on the available measurement histories. A GRU has 13,121 parameters. Inputs contain 24 historical means from 128 shots each, plus an input feature giving the elapsed time between measurements; next-observation training targets use 1024 shots. True angles, path families, and evaluation risks are stored separately. Whole simulated noise paths and base-parameter blocks are disjoint across training, development, and test. Architecture details and the limits of the ablations are in Supplemental Material, Sec.~\ref{cal:app:learning}.

The references include a development-tuned Bayesian random-walk filter, a tuned sliding window, the latest observation, fixed retention, and robust finite-catalog search. They use the same observation histories, permitted recovery actions, and budgets for new calibration shots. The comparison uses local reference implementations of the specified adaptive policies. The comparison measures recovery improvement relative to retaining the incumbent and the difference between neural and classical predictors under the same observation and acquisition constraints. Conditional protection applies to both neural and classical proposers.

In the first held-out test, all methods spend 8192 new calibration shots through 4 predetermined looks with total $\alpha=0.01$. Across 600 paths, mean guarded GRU infidelity is \GRURisk, compared with \RetainRisk\ for retaining the incumbent and \BayesRisk\ for the Bayesian method. Thus learned advice lowers mean infidelity relative to retaining the incumbent in this simulation. A classical estimator also achieves most of this recovery improvement under the same contract. The paired GRU-minus-Bayesian difference is $-0.001278$, with hierarchical 95\% bootstrap interval $[-0.002675,+0.000070]$, including zero. The measured neural improvement falls short of the 10\% risk-reduction target. All methods incur the same cost of new calibration measurements in this 4-look protocol. The statistical comparison remains inconclusive, including for equivalence. These predictors receive signed-probe histories; the passive sign-blindness result in Proposition~\ref{prop:exact-record} concerns a different observation experiment.

\subsection{Choosing a budget before acquisition}
The acquisition extension uses 600 new simulated noise paths and keeps the predictor checkpoints unchanged. It chooses 1 budget $n\in\{0,128,512,2048,8192\}$ before the new measurements. At each of 5 fixed prices $\lambda$, the chosen budget maximizes expected recovery gain minus shot cost:
\begin{align}
 n^*(h)&\in\arg\max_n\{V(n\mid h)-\lambda n\},\\
 V(n\mid h)&=\mathbb E_{\theta \mid h}\mathbb E_{k\mid \theta ,n}
 [\risk(\theta ,u_0)-\risk(\theta ,u(k,n))].
 \label{cal:eq:voi}
\end{align}
For every possible count, $u(k,n)$ is the best robustly certified catalog action, or retention. The inner binomial sum includes every count and is evaluated numerically. The outer expectation uses a finite physical grid. Bayesian acquisition uses its posterior; neural and window forecasts use development-fitted Gaussian uncertainty. The predictive distributions guide measurement allocation; a confidence region formed from newly acquired calibration outcomes supplies the acceptance certificate. This region does not use the predictor's uncertainty estimate as a coverage guarantee.

At $\lambda=10^{-6}$, Bayesian acquisition uses \BayesShots\ new calibration shots on average, a \ShotReduction\ reduction from 8192, while mean infidelity changes from 0.0365583 to 0.0372372. The common historical cost is another 3072 shots. The GRU traces a similar cost--risk frontier (Fig.~\ref{cal:fig:acquisition}); the result describes a one-look risk--cost tradeoff. Unlike the first test, this extension makes 1 acquisition with $\alpha=0.01$. Its stronger single-look confidence allocation and new simulated noise paths prevent direct attribution of the cross-experiment risk decrease to learning.

\begin{figure*}[t]\centering\includegraphics[width=\textwidth]{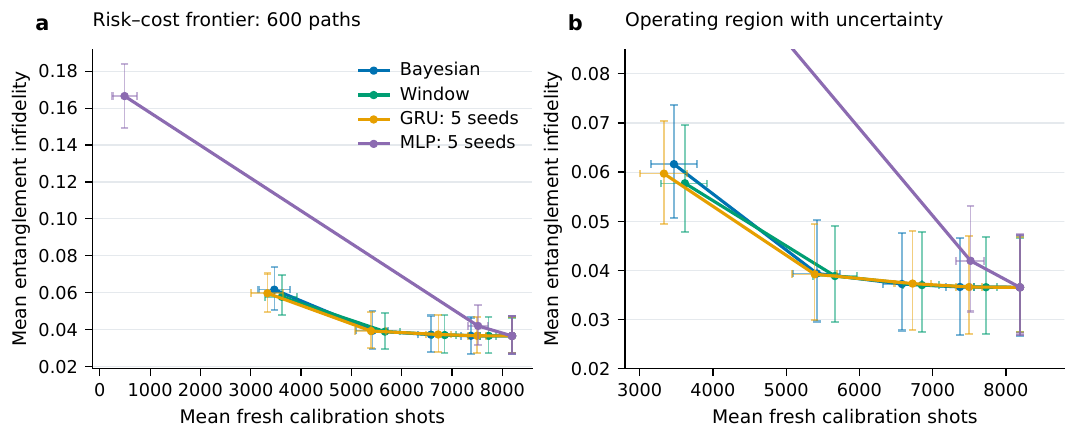}\caption{Risk and calibration cost on 600 held-out simulated noise paths. The 5 evaluated prices per method are connected by visual guides. Horizontal and vertical bars show the saved pointwise hierarchical bootstrap 95\% intervals for mean new calibration shots and infidelity. Neural methods are gated recurrent units (GRUs) and multilayer perceptrons (MLPs), each with 5 training seeds. The right panel expands the operating region. The horizontal axis counts new calibration shots beyond the common 3072 historical shots per path.}\label{cal:fig:acquisition}\end{figure*}

\section{Calibration and transfer to circuit-level decoder updates}
\label{sm:contents:15}
\label{cal:sec:surface}
The coherent toric model gives losses by numerical evaluation of the exact channel under the extraction and phase-correction assumptions of Sec.~\ref{sec:model}. A separate circuit-level experiment tests whether finite-data acceptance can admit classical decoder updates that reduce logical failure risk in stochastic circuits. We use Stim 1.16.0~\cite{gidney2021} and PyMatching 2.4.0~\cite{higgott2025}, distances 3 and 5, 30 syndrome rounds, and both $X$- and $Z$-basis memories. Noise consists of explicit stochastic Pauli faults. This experiment uses Stim's stochastic Pauli circuits and a catalog of classical matching decoders.

The incumbent uses a fixed matching prior. Five alternative priors emphasize measurement faults, idle-$Z$ faults, local-gate faults, paired faults, or a mixture. Each candidate is scored on the \emph{same} immutable memory record as the incumbent. Paired binary failure differences and simultaneous binomial bounds compare the worst of the two basis risks. Supplemental Material, Sec.~\ref{cal:app:surface} gives the bound. A separately implemented raw measurement parity conversion is compared against Stim's detector/observable conversion on every sampled batch.

A 16-condition pilot uses 1,310,720 complete memory shots and certifies 10 updates. A confirmation with the same fixed prior catalog samples 200 new parameter paths per distance, focusing on distance-3 local-gate drift and distance-5 idle-$Z$ drift. Their comparison changes noise family as well as distance. Each basis uses 8192 calibration shots and two independent final batches of 4096 shots: the calibration channel and a return-to-base channel. The confirmation adds 13,107,200 shots, for 14,417,920 across both stages. Calibration admits all 400 candidates at the calibration noise settings. Separate pointwise final risk boxes establish benefit for all 200 distance-3 cases and 146 distance-5 cases; the remaining 54 are unresolved.

In the original return-to-base challenge, acceptance succeeds at the calibration settings while the bound for later deployment remains vacuous. Under stale deployment all 400 \emph{individual} effects remain unresolved by the prescribed final boxes. The full stale gate refuses every update by adding a worst-case total-variation allowance of one, or two in excess risk. This bound is valid but vacuous; a fixed expiry can also refuse this single-delay challenge. Estimated means, their maximum-operation bias, and a separately labeled basis-averaged diagnostic appear in Supplemental Material, Sec.~\ref{cal:app:surface}. This circuit experiment establishes benefit at the calibration settings; stale-update certification requires a noise-family-specific drift bound.

\section{Geometry. Disconnected supports and the parity law}
\label{sm:contents:16}
\label{app:geometry}

The labeled dual graph $G_L^*$ has $L^2$ vertices, $2L^2$ edges, and incidence
matrix $H_Z$; at $L=2$, parallel edge labels remain distinct.  Connectivity
gives $\operatorname{rank}H_Z=L^2-1$, as does the primal calculation for
$H_X$.  Therefore
\begin{equation}
 \dim\ker H_Z=L^2+1,
 \qquad
 \dim(\ker H_Z/\row H_X)=2.
 \label{eq:quotient-dimension}
\end{equation}

For an arbitrary Eulerian support $a$, define seam parities
\begin{align}
 \omega_x(a)&=\sum_y a_{v_{0,y}}\pmod2,\nonumber\\
 \omega_y(a)&=\sum_x a_{h_{x,0}}\pmod2.
 \label{eq:seam-parities}
\end{align}
They vanish on dual-face boundaries and equal $(1,0),(0,1)$ on the straight
logical cycles.  By \cref{eq:quotient-dimension} they identify the quotient
with $\Ftwo^2$, avoiding an integer winding for disconnected supports.

Every Eulerian support decomposes into edge-disjoint labeled circuits; a
nontrivial support has a component with nonzero seam parity.  If it lies in a
cut, each component does and is even by cycle--cut orthogonality.  Thus
\cref{eq:affine-distance} reduces to one even nontrivial circuit.

Orient that circuit and let $n_x^\pm,n_y^\pm$ count its lifted signed steps.
For displacement $L(p,q)$ and length $w$,
\begin{align}
 w&=n_x^++n_x^-+n_y^++n_y^-,\nonumber\\
 Lp&=n_x^+-n_x^-,
 &Lq&=n_y^+-n_y^-,\nonumber\\
 w-L(p+q)&=2(n_x^-+n_y^-).
 \label{eq:winding-details}
\end{align}
This proves \cref{eq:winding-bound}; seam parity is
$(p\bmod2,q\bmod2)$.  For odd $L$, even $w$ makes $p+q$ even, so a nonzero
parity pair is $(1,1)$, proving the $2L$ bound without assuming a staircase.

For achievability, follow
$f_{j,j-1}\xrightarrow{h_{j,j}}f_{j,j}
\xrightarrow{v_{j+1,j}}f_{j+1,j}$.  For $L\ge3$, the $2L$ projected vertices
and tagged edges are distinct before closure, forming a simple circuit.  Each
edge has one endpoint in $\{f_{j,j}\}$, hence lies in its cut.  At $L=2$, the
main proof's parallel-edge witness avoids falsely simplifying the graph.

\section{Effect character expansion and exact circuit count}
\label{sm:contents:17}
\label{app:record}

For independent angles, write $c_e=\cos(\epsilon_e/2)$ and
$q_e=\sin(\epsilon_e/2)$.  Then
\begin{equation}
 U_{\bm\epsilon}=\sum_{a\in\Ftwo^n}
 (-i)^{|a|}\prod_{e\in a}q_e\prod_{e\notin a}c_e\,X(a).
 \label{eq:physical-support-expansion}
\end{equation}
Projection selects $\widetilde H_Za=s$; since $X$ operators commute and
$a+r_s\in\ker H_Z$ there,
\begin{equation}
 K_s(\bm\epsilon)=
 \sum_{a:\widetilde H_Za=s}
 (-i)^{|a|}\prod_{e\in a}q_e\prod_{e\notin a}c_e\,
 \overline X_{[a+r_s]}.
 \label{eq:multivariate-kraus}
\end{equation}
In $F_s=K_s^\dagger K_s$, recovery cancels between paired supports.  Writing
their difference $f\in\ker H_Z$ gives
$\beta_{s,f}=\sum_{a:\widetilde H_Za=s}\overline a_a a_{a+f}$ for scalar
coefficient $a_a$ in \cref{eq:physical-support-expansion}.  Insert
\begin{equation}
 \mathbf1_{\{\widetilde H_Za=s\}}
 =2^{-r}\sum_{z\in\Ftwo^r}(-1)^{z\cdot(s+\widetilde H_Za)}.
 \label{eq:syndrome-character}
\end{equation}
For one edge, the four sums over $a_e$, indexed by $(f_e,(h_z)_e)$, are
\begin{equation}
 1,\qquad \cos\epsilon_e,\qquad0,
 \qquad-i\sin\epsilon_e.
 \label{eq:four-local-factors}
\end{equation}
Their product gives \cref{eq:character-formula}.  For a nonzero Taylor
multi-index $\bm a$, sine coordinates are exactly those with odd $a_e$,
proving the support-parity argument used in the proof of
\cref{thm:record-onset} without a numerical-0 test.

For a covered weight-$2L$ circuit $g$, only $0$ and $g$ have zero boundary, so
the restriction of $\widetilde H_Z$ to $g$ has rank $2L-1$.  The nonempty
fiber of $z\mapsto h_z|_g$ over the all-one word therefore has size
$2^{r-(2L-1)}$, proving the selector factor in
\cref{eq:effect-coefficient}.

For the circuit count, orient a covered minimum circuit.  Equality in
$2L\ge L(|p|+|q|)$, together with the established $(1,1)$ winding parity,
gives $|p|=|q|=1$ and forces every step to have the corresponding winding
sign.  Thus its root, four sign pairs, and a word with $L$ horizontal and
vertical letters determine it.  Coincident
intermediate projected vertices require an intervening subword to change each
coordinate by a multiple of $L$; a proper repeated segment therefore contains
all letters of one type and none of the other.  The nonsimple words are exactly
the $2L$ cyclic shifts of $H^LV^L$.  Dividing by the $2L$ roots and two
orientations of each simple cycle yields
\begin{equation}
 M_L=\frac{L^2\cdot4\,[\binom{2L}{L}-2L]}{2L\cdot2}
 =L[\binom{2L}{L}-2L].
 \label{eq:ML-appendix}
\end{equation}

\section{Coefficient-level comparison with stabilizer effects and toric paths}
\label{sm:contents:18}
\label{app:prior-comparison}

\subsection{Generator coefficients and the cut constraint}

Hu, Liang, and Calderbank express syndrome probabilities as diagonal
generator-coefficient terms plus logical cross terms
\cite[Sec.~4.2, Eq.~(91)]{hu2022designing}.
To translate their diagonal-$Z$ presentation, exchange physical $X$ and $Z$
by Hadamard conjugation.  Their classical spaces then obey
$C_2=\row H_Z$, $C_1^\perp=\row H_X$, and
$C_2^\perp/C_1^\perp=\mathcal L_X$ in our notation.  All stabilizer signs
are positive here.  Their syndrome coset representative $\mu$ may be
chosen as our $r_s$, so $s=\widetilde H_Z\mu$; their logical coset label
$\gamma$ is our $c\in\Ftwo^2$.  With the physical rotation angle $\theta $ and
the half-angle amplitudes in \cref{eq:kappa-formula}, their
$A_{\mu,\gamma}$ is $\kappa_{s,c}$.  Changing the recovery representative
only relabels these coefficients and does not change the effect.

In particular their Eq.~(91), with the conjugation from $K_s^\dagger$
written explicitly, becomes
\begin{align}
 B_{s,d}&=\sum_{c\in\Ftwo^2}
       \overline{\kappa_{s,c}}\kappa_{s,c+d},
       \qquad d\in\Ftwo^2,\label{eq:hlc-effect-map}\\
 p_s(\rho)&=\sum_c|\kappa_{s,c}|^2+
       \sum_{d\ne0}B_{s,d}\Tr(\overline X^d\rho).
       \label{eq:hlc-probability-map}
\end{align}
Thus the basic criterion for input dependence is already present in that
work.  The additional issue is which physical cross terms cancel in this
toric specialization.  Write $u_a$ for the physical amplitude of $X(a)$ in
\cref{eq:physical-support-expansion}.  Expanding the first line gives
\begin{equation}
 B_{s,d}=\sum_{\substack{f\in\ker H_Z\\{}[f]=d}}
       \sum_{a:\widetilde H_Za=s}\overline{u_a}u_{a+f}
       =\sum_{[f]=d}\beta_{s,f}.
 \label{eq:hlc-support-map}
\end{equation}
The character transform of the inner sum is
\cref{eq:character-formula}: the local pair sum is zero when
$f_e=1$ and $(h_z)_e=0$.  This is the extra cut-containment cancellation.
At the first allowed degree $2L$, every surviving nonscalar support is a
simple circuit in class $P$; at zero syndrome all have the same sign and
the same selector factor $-2^{-(2L-1)}$.  This noncancellation establishes
an attained effect order, rather than just a lower bound from even logical
weight.  It also shows precisely how the support multiplicity enters the
observable coefficient.

\subsection{Seam-linked paths and distinct simple supports}

Iverson and Preskill count monotone paths with
$\binom{l_1+l_2}{l_1}$ in Eq.~(H.1), then use two seam linkings in
Eq.~(H.2) to estimate weight-$2L$ double-logical strings
\cite[Appendix~H]{iverson2020coherence}.  After exchanging $Z$ and $X$
and identifying the square lattice with its dual, their double-logical
class is our $P$.  Let $i,j=0,\ldots,L-1$ label the two chosen seam edges.
Denote their Eq.~(H.2) by
\begin{align}
 N_{ij}&=A_{ij}+A_{L-1-i,j},\label{eq:ip-seam-count}\\
 A_{ij}&=\binom{i+j}{i}
       \binom{2L-2-i-j}{L-1-i}.\nonumber
\end{align}
The two terms correspond to the two relative winding signs.  A simple
minimum circuit crosses each seam exactly once, so its seam-edge pair and
linking are unique; no choice of root or traversal orientation remains in
this representation.  Within each linking there is also one nonsimple
path pairing: a straight horizontal loop joined to a straight vertical
loop at their common vertex.  It has the selected seam edges but contains
two odd length-$L$ circuits, so cannot be contained in a cut.  The two
linkings represent this same excluded support twice.  These are the only
nonsimple cases, by the repeated-vertex argument in \cref{app:record}.
Consequently the conversion to our distinct tagged simple supports is
\begin{equation}
 M_L=\sum_{i,j=0}^{L-1}(N_{ij}-2)
     =L\binom{2L}{L}-2L^2.
 \label{eq:ip-exact-conversion}
\end{equation}
For the last equality, fix $k=i+j$.  Vandermonde convolution gives
$\sum_{i+j=k}A_{ij}=\binom{2L-2}{L-1}$ for every
$k=0,\ldots,2L-2$, with binomial coefficients outside their natural range
set to zero.  The second linking has the same double sum.  Hence
$\sum_{i,j}N_{ij}=2(2L-1)\binom{2L-2}{L-1}
=L\binom{2L}{L}$.

Equivalently, rooted oriented words have $L^2$ roots and four winding sign
pairs.  Removing the $2L$ cyclic shifts of $H^LV^L$ per root and sign pair,
then dividing by $2L$ roots and two orientations per simple support, yields
\cref{eq:ML-appendix}.  Translated supports are distinct because the edge
tags are fixed.  The two normalizations therefore agree exactly.
An independent integer-arithmetic check verifies the seam-sum identity and
enumerates the tagged simple supports at $L=3,5$.
Equation~(H.3) of Ref.~\cite{iverson2020coherence} uses the $4^L/\sqrt{\pi L}$ path-growth scale for its large-$L$ comparison, whereas the tagged-support multiplicity here scales as $M_L\sim\sqrt{L/\pi}\,4^L$.  The enumeration uses standard seam-linking and path-counting arguments.

The exact coefficient measures the full leading zero-syndrome parity
contrast, $-M_L/2^{2L-2}$, rather than the existence of one surviving
circuit.  It supplies a parameter-free normalization for perturbative
checks of the syndrome instrument and distinguishes an attained leakage
order from a support obstruction alone.  A single nonzero witness would
suffice for the exponent and the qualitative separation.  The contribution is the cancellation and common-sign survival in the effect, the identified parity observable, and the quantitative history bound.

\subsection{A related fixed-code perturbative privacy bound}

Shen and Zhong also obtain a fixed-$(d,T)$, weak-amplitude-damping bound
$O(\gamma^{d_Z})$ on the syndrome contrast between logical $Z$ eigenstates
\cite[Proposition~1]{shen2026transcript}.
That result excludes leakage below a permitted order without establishing
general attainment.  The present coherent-rotation calculation proves a
nonzero $\theta ^{2L}$ coefficient for the product parity $P$ of the odd torus.
The fixed-instrument calculation complements the established privacy comparison for growing-distance noisy memories.

\section{Leading-order line reduction and the channel remainder}
\label{sm:contents:19}
\label{app:drift}

Write $\kappa_{s,c}(\theta )=\sum_m \theta ^m\kappa_{s,c}^{(m)}$.  From
\cref{eq:kappa-formula},
\begin{equation}
 \kappa_{s,c}^{(m)}=(-i)^m r_{s,c}^{(m)},
 \qquad r_{s,c}^{(m)}\in\mathbb R.
 \label{eq:kappa-reality}
\end{equation}
A weight-$w$ support contributes only at powers congruent to $w$ modulo two.
For $c\ne0$, $a+r_s$ is nontrivial, and
\begin{equation}
 |a|\ge w(s),
 \qquad |a|\ge L-w(s)
 \label{eq:amplitude-order-bound}
\end{equation}
give the main proof's threshold; for class $(1,1)$, the second becomes
$2L-w(s)$.

To show that the following inequalities are equalities, consider an order-$L$ cross term with
nontrivial residual $g=a+r_s$ and identity-class support $a'$ of the same
syndrome.  Taylor degree dominates support weight, while
\cref{eq:amplitude-order-bound} gives $|a|+|a'|\ge L$; thus equality holds.
Minimality of $r_s$ gives $w(s)\le |a'|$, so
\begin{equation}
 L\le |g|\le |a|+w(s)\le |a|+|a'|=L.
 \label{eq:order-L-equality-squeeze}
\end{equation}
All inequalities are equalities.  Specifically, $|g|=L$, $|a'|=w(s)$,
$a\cap r_s=\varnothing$, $r_s\subset g$, and $a=g\setminus r_s$.  Also
$|a|=L-w(s)$ and, because $L$ is odd, $w(s)\le(L-1)/2$.  The cycle
$a'+r_s$ has weight at most $2w(s)<L$ and is stabilizer-trivial.  Every
nontrivial weight-$L$ $X$ cycle is a straight row or column.  The projection
below makes $r_s$ its unique boundary minimizer; hence the other minimizer
$a'=r_s$.  Exchanging the pair is identical.  Thus complementary subsets of
straight logical lines exhaust the order-$L$ coefficient.

The channel has the class-matrix form
\begin{align}
 \Lambda_L(\theta )(\rho)&=\sum_{c,c'\in\Ftwo^2}
 M_{c,c'}(\theta )\overline X^c\rho\overline X^{c'},\nonumber\\
 M_{c,c'}(\theta )&=\sum_s\kappa_{s,c}(\theta )\overline{\kappa_{s,c'}(\theta )}.
 \label{eq:channel-class-matrix}
\end{align}
At odd $L$, coefficients with both logical-class indices nonzero start above degree $L$; trace
preservation gives $\sum_cM_{c,c}(\theta )=1$ and
$M_{0,0}=1+O_L(\theta ^{L+1})$.  At degree $L$, only
$(m,m')=(L-w(s),w(s))$, whose phase from
\cref{eq:kappa-reality} is
\begin{equation}
 (-i)^m(i)^{m'}=(-1)^mi^L.
 \label{eq:pairing-phase}
\end{equation}
This phase is purely imaginary, so the leading generator is a sum of commutators.

For global line minimization, take a horizontal dual-graph logical row
$R_{y_0}$ and let
$\pi_0^{(y_0)}(x,y)=(x,y_0)$.  On edge basis elements, define
$\pi_1^{(y_0)}(e_{x,y}^{\parallel})=e_{x,y_0}^{\parallel}$ for edges parallel
to the row and $\pi_1^{(y_0)}(e_{x,y}^{\perp})=0$, extending linearly.
On basis edges, $\partial\pi_1^{(y_0)}=\pi_0^{(y_0)}\partial$; mod-two
collisions only decrease weight.  If $A\subset R_{y_0}$ has boundary $D$,
every boundary-$D$ chain projects to $A$ or $R_{y_0}\setminus A$.  Hence
\begin{equation}
 w(D)=\min\{|A|,L-|A|\}.
 \label{eq:line-minimum}
\end{equation}
If $|A|<L/2$ and $|r|=|A|$, projection must equal $A$ with no weight loss,
forbidding transverse edges and collisions.  The boundary excludes off-target
partial rows, and a full off-target row weighs $L>|A|$; thus $r=A$ uniquely.
Exchanging coordinates proves the column case.  Odd $L$ excludes half-line
ties, and winding equality places every active residual on such a line.

On a fixed line, complementary subsets share a syndrome; the lighter corrects
to identity and the heavier to the line class.  Their leading products sum to
\begin{align}
 \beta^{\rm line}
 &=\frac{i^L}{2^L}
 \sum_{k=(L+1)/2}^L(-1)^k\binom Lk\nonumber\\
 &=-\frac{i}{2^L}\binom{L-1}{(L-1)/2}.
 \label{eq:line-sum-detail}
\end{align}
Pascal induction proves the partial alternating-binomial identity.  There are
$L$ representatives per one-logical class and none of weight $L$ for the
product class, giving \cref{eq:alpha-result}.

For the remainder, equip finite Fourier polynomials with
$\|f\|_{\mathcal F,1}=\sum_\omega|\widehat f_\omega|$.  Each support amplitude
$\cos(\theta /2)^{n-|a|}[-i\sin(\theta /2)]^{|a|}$ has Fourier coefficient norm at most
one.  If $N_{s,c}$ counts supports in cell $(s,c)$, then
$\|\kappa_{s,c}\|_{\mathcal F,1}\le N_{s,c}$, and submultiplicativity gives
\begin{align}
 \|M_{c,c'}\|_{\mathcal F,1}
 &\le\sum_sN_{s,c}N_{s,c'}\nonumber\\
 &\le\Bigl(\sum_sN_{s,c}\Bigr)
       \Bigl(\sum_sN_{s,c'}\Bigr)\nonumber\\
 &\le2^{2n}=2^{4L^2}.
 \label{eq:fourier-l1-bound}
\end{align}
The frequencies of $M_{c,c'}$ have magnitude at most $n=2L^2$, so every
$(L+1)$st derivative is bounded by $n^{L+1}2^{2n}$.  Integral Taylor remainder
therefore gives, for real $\theta $,
\begin{equation}
 \left|M_{c,c'}(\theta )-\sum_{j=0}^{L}[\theta ^j]M_{c,c'}(\theta )\theta ^j\right|
 \le\frac{(2L^2)^{L+1}2^{4L^2}}{(L+1)!}|\theta |^{L+1}.
 \label{eq:scalar-remainder}
\end{equation}
The $16$ maps that multiply an operator by logical Paulis on the left and right have diamond norm one; the triangle inequality lifts \cref{eq:scalar-remainder} to
\cref{eq:channel-remainder}.

\section{Product limit and null-prefix histories}
\label{sm:contents:20}
\label{app:sequential}

For a history prefix $z$, compose \cref{eq:history-law}'s branch maps to obtain
its unnormalized conditional state.  At positive trace, divide by it; at zero
trace, every child has zero probability regardless of the filler kernel.  Thus
the recursion for probabilities of partial histories is exact and $\eta_L(\theta )$ controls every next-round kernel
pair.  First-disagreement maximal coupling proves \cref{eq:history-bound}
without common-support or likelihood-ratio assumptions.

Second, define $\mathcal V_h=e^{\sigma h\mathcal G_L}$, a unitary channel.
For any $R$,
\begin{equation}
 \Lambda_L(\theta )^R-\mathcal V_h^R
 =\sum_{j=0}^{R-1}\Lambda_L(\theta )^{R-1-j}
 [\Lambda_L(\theta )-\mathcal V_h]\mathcal V_h^j.
 \label{eq:mixed-telescoping-sum}
\end{equation}
Each channel factor has diamond norm one, avoiding powers of the generally
nonpositive Euler map $\id+\sigma h\mathcal G_L$.  Exponential remainder gives
$\|\mathcal V_h-\id-\sigma h\mathcal G_L\|_\diamond
\le h^2g_L^2e^{hg_L}/2$, and
$R\le T/h$ gives the first two terms of \cref{eq:product-bound}.  For
$R=\lfloor\tau/h\rfloor$, $0\le\tau-Rh<h$, and Duhamel's formula bounds the remaining time interval by $hg_L$, proving the uniform limit including floor discontinuities.

\section{Exhaustive numerical protocol}
\label{sm:contents:21}
\label{app:numerics}

For $L=3$ we use the first 8 plaquette rows in the coordinate order of
\cref{eq:edge-labels} as the independent syndrome basis.  For each support
$a\in\Ftwo^{18}$, we compute its syndrome, Hamming weight, and seam
parities.  Comparing weight and then the ascending edge-index tuple over
all supports gives exactly the declared recovery $r_s$.  Let
$N_{s,c,w}$ count supports of weight $w$ with syndrome $s$ and residual
logical class $c=[a+r_s]$.  Then
\begin{equation}
 \kappa_{s,c}(\theta )=\sum_{w=0}^{18}N_{s,c,w}
 \cos(\theta /2)^{18-w}[-i\sin(\theta /2)]^w.
 \label{eq:numeric-counts}
\end{equation}
The integer counts are evaluated exhaustively. Independent decimal-arithmetic checks assess numerical precision; the plotted values follow directly from these counts.

For a simultaneous logical-$X$ eigenstate labeled by $x\in\Ftwo^2$, let
$z_s(x)=\sum_c(-1)^{x\cdot c}\kappa_{s,c}$.  The averaged channel multiplies
$|x\rangle\langle y|$ by
\begin{equation}
 C_{xy}(\theta )=\sum_s z_s(x)\overline{z_s(y)}.
 \label{eq:numeric-schur}
\end{equation}
Consequently repeated rounds multiply this matrix element by $C_{xy}^R$,
without constructing any histories.  Taking $x=++$, $y=--$ gives
$d_b=|1-C_{xy}^{R_*}|$ and
$\ell_b=(1-\operatorname{Re}C_{xy}^{R_*})/2$.
The plotted one-round channel coefficient is
$|\sum_s\kappa_{s,(1,0)}\overline{\kappa_{s,(0,0)}}|$, tending to
$\alpha_3\theta ^3$.  This quantity is the indicated channel-expansion coefficient; the diamond norm is evaluated separately.  The exact identity \cref{eq:eta-exact} supplies the record
curve without optimization over input states.

The checks verify support-count completeness, trace preservation,
absence of single-logical-$X$ terms in every effect, sign symmetry,
and selected direct two-round compositions against
\cref{eq:exact-history-mixture}.  The two leading coefficients are checked
with rational arithmetic on integer support counts.  Independent
60-digit trigonometric sums and repeated complex multiplication at
$\theta =0.015,0.025,0.05,0.1,0.15$ check floating-point accumulation.  No
normalization is imposed to force the trace-preservation check to pass.
The grid evaluates finite-angle behavior over $0.015\le \theta \le0.15$; the analytical results establish the asymptotic exponent and stated all-angle identities.

\section{Recovery calculations and validation}
\label{sm:contents:22}
\label{app:recovery-protocol}

The calculation regenerates all $2^{18}$ support counts and recovery representatives before use.  An
independent 18-qubit computational-basis construction applies every physical
$R_X(\theta )$ gate, projects onto all 256 syndrome sectors, and recovers their
logical Kraus matrices at $\theta =0,0.015,0.05,0.15,-0.05,0.3,0.6$.
The tolerance is $10^{-12}$, using complex128 and float64 without probability
clipping or forced renormalization.  Unit tests check Kraus versus full
superoperator composition, parity overlap, Pauli optimization, feasible fixed
control, the spectral bound, and enumeration of all randomized actions.
The additional terminal check explicitly sums every two-round history and
tests the analytical precision bound independently.

The fixed comparison uses the principal eigenvector of $C^{\circ R}$,
projects its entries to unit modulus, and distributes the resulting phase
correction over $R$ rounds.  The eigenvector phase is fixed before choosing
angle branches.  At the anchor, the feasible policy lies approximately $2.41\times10^{-8}$ above the spectral lower bound, quantifying its proximity within the stated comparison.  The Pauli policy
maximizes one-round fidelity.  The optimization here selects the best one-round Pauli correction.

For grid spacing $\Delta=2\pi/2^b$, write the desired phase as
$\phi=(n+f)\Delta$, with $n\in\mathbb Z$ and $0\le f<1$.  Independently choose
$n\Delta$ or $(n+1)\Delta$ with probabilities $1-f$ and $f$.  With
$w_s(x)=\mathbb E e^{i\phi_s(x)}$, the corrected Schur multiplier is
\begin{equation}
 C^{\rm rand}_{xy}=\begin{cases}
 \sum_s z_s(x)z_s(y)^*w_s(x)w_s(y)^*,&x\ne y,\\
 \sum_s|z_s(x)|^2=1,&x=y.
 \end{cases}
 \label{eq:randomized-rounding-channel}
\end{equation}
The diagonal uses $\mathbb E|e^{i\phi}|^2=1$, not
$|\mathbb E e^{i\phi}|^2$.  Independently sampled controller randomness is
assumed each round.  This applies the established benefit of mixing nearby
unitaries to phase quantization \cite{campbell2017mixing}; it differs from
the Pauli-twirling construction of randomized compiling
\cite{wallman2016noise}.  

For complete logical depolarization with probability $\epsilon$ after each
additional control round, unitality gives
\begin{equation}
 F_e^{(\epsilon)}(R)=(1-\epsilon)^R F_e(R)
                  +\frac{1-(1-\epsilon)^R}{16}.
 \label{eq:recovery-cost}
\end{equation}
The phenomenological budget treats the correction as a whole action; pulse synthesis, measurement faults, and latency-dependent gate noise require a resolved implementation model.  It also charges per-round control that
can be deferred in this commuting storage example.

\begin{figure*}[t]
 \centering
 \includegraphics[width=0.98\textwidth]{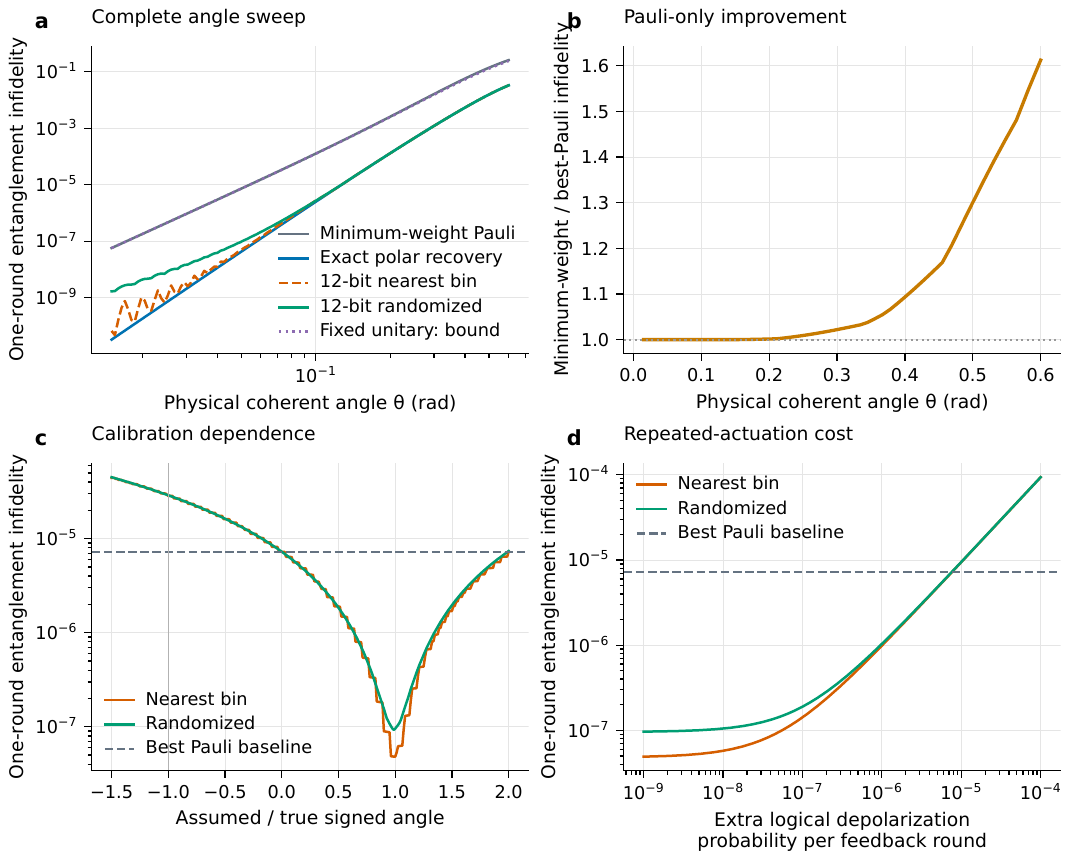}
 \caption{Complete recovery robustness sweeps.  (a) One-round infidelity
 across 122 angles including the reference angle $\theta=0.05$.  (b) Relative Pauli-only
 improvement occurs at stronger noise, where logical infidelities are high.
 (c) Assumed calibration divided by true $\theta =0.05$; $-1$ is the wrong-sign case.
 (d) Extra per-round logical depolarization erases the benefit.  Dashed and dotted curves labeled as bounds are analytical comparators.  Every curve is an exhaustive instrument calculation, so sampling
 error bars are inapplicable.}
 \label{fig:recovery-robustness}
\end{figure*}

The parameter sweeps contain 122 angles from $0.015$ to $0.6$, 124 distinct
round counts through 4188, 281 signed calibration ratios from $-1.5$ to $2$,
162 actuation-noise values including zero, and every integer precision from
4 to 24 bits.  Six pilot angles were inspected before the protocol was
written.  The precision/randomization extension followed the initial
rounding failure; both the initial run and this exploratory extension are
retained.  The study was conducted without preregistration, with its decisions and failed checks retained in the research record.  This calculation evaluates the instrument directly through exhaustive deterministic sums.

\section{Exact channel evaluation and continuous bounds}
\label{sm:contents:23}

\label{cal:app:numerics}
Write the diagonal entries of $V_s(u)K_s(\theta )$ as $q_{s,i}(\theta ,u)$. Equation~\eqref{cal:eq:channel} is a Schur channel, with
\begin{align}
 C_{ij}(\theta ,u)&=\sum_s q_{s,i}(\theta ,u)q_{s,j}(\theta ,u)^*,\\
 [\mathcal E_{\theta ,u}(\rho)]_{ij}&=C_{ij}(\theta ,u)\rho_{ij}.
\end{align}
For fixed $u$, channel composition multiplies entries, and
\begin{equation}
 F_e(\mathcal E_{\theta ,u}^{H})=\frac{1}{16}\sum_{i,j=1}^4 C_{ij}(\theta ,u)^H.
 \label{cal:eq:schur}
\end{equation}
The sum is real; the implementation takes its real part and checks the channel conventions against direct evaluation of the instrument. The reference uses float64/complex128 arithmetic. Physical support counts come from exhaustive enumeration of the $18$-qubit model. The original 78 tests, including a separate physical-state construction, were rerun before reuse.

Instrument amplitudes are degree-$N$ homogeneous polynomials in $\cos(\theta /2)$ and $\sin(\theta /2)$. Their products with complex conjugates contain integer Fourier frequencies at most $N$. The fixed phase tables do not change this degree. Consequently 37 Fourier samples recover each $C_{ij}$ at $N=18$ in exact arithmetic. For coefficients $c_{f,ij}$,
\begin{equation}
 |C''_{ij}(\theta )|\leq M_{ij}:=\sum_{f=-N}^{N}f^2|c_{f,ij}|.
\end{equation}
At a cell center $x$ and half-width $h$, this gives $\sup_{|\theta -x|\leq h}|C'_{ij}(\theta )|\leq |C'_{ij}(x)|+hM_{ij}$. Complete positivity and trace preservation imply $|C_{ij}|\leq1$, so differentiation of Eq.~\eqref{cal:eq:schur} yields
\begin{equation}
 |F'_e(\theta )|\leq\frac{H}{16}\sum_{ij}|C'_{ij}(\theta )|.
\end{equation}
The numerical bound on $D_u$ includes derivative allowances for \emph{both} candidate and incumbent. It considers every cell intersecting $C_{\rm dep}$, including boundary cells, and adds the corresponding half-width allowance to the grid excess loss. The plotting grid alone is never treated as a certificate.

Comparisons at additional angles check Fourier reconstruction and derivatives against direct Kraus evaluation. An additional $10^{-8}$ allowance covers observed numerical residuals. The tests assess the numerical tolerance. Proposition~\ref{cal:prop:conditional} requires valid numerical bounds; a formal floating-point enclosure would additionally certify rounding error beyond the tested cases.

A separate audit enumerates all 8193 possible counts and all 13 actions at 1692 physical/nuisance settings. An offline audit flags a count if \emph{any} admitted action violates the promised margin. The largest integrated probability is 0.00281349, below $\alpha=0.01$. The audit exhausts counts and actions on a finite physical-parameter grid, using evaluator knowledge to compute worst-selection risk. This zero-drift audit is the 1692-setting subset of the extended audit in Supplemental Material, Sec.~\ref{sec:completion}. Adding drift radii $0.005$ and $0.02$ rad, with the corresponding allowed deployment shifts, increases the grid to 5076 settings and gives maximum margin-violation probability $0.00378720114$; the zero-drift results are unchanged.

\section{Information and duration references}
\label{sm:contents:24}
\label{cal:app:drift}
For two equal-prior sign hypotheses, let $C$ and $W$ be the probabilities of correct-sign and wrong-sign activation, with abstention allowed, and $q=C+W$. If the observation laws have total-variation distance $T$, a binary decision rule obeys $C-W\leq T$. Indeed, the difference is one half of the expectation difference of a function taking values in $[-1,1]$, bounded by total variation. Hence
\begin{equation}
 W\geq\max\{0,(q-T)/2\}.
\end{equation}
Requiring $q\geq0.8$ and $W\leq0.01$ therefore requires $T\geq0.78$. With known offset zero, contrast $0.99$, and $|\theta |=0.1$, numerical evaluation of the binomial likelihoods first meets that necessary condition at 153 shots. This optimistic sign-only reference excludes nuisance uncertainty, magnitude accuracy, and recovery-margin requirements. The 153-shot estimate concerns sign discrimination; QEC certification additionally constrains magnitude, nuisance parameters, and the recovery margin. Actions beneficial at both signs require a different decision criterion from sign discrimination.

For changing angles, product rotations satisfy
\begin{equation}
 \|R_X(\theta )^{\otimes N}-R_X(\theta _0)^{\otimes N}\|
 \leq\tfrac{N}{2}|\theta -\theta _0|.
\end{equation}
The corresponding channel diamond distance is at most $N|\theta -\theta _0|$. Fixed syndrome extraction, recovery, and classical record retention are completely positive trace-preserving postprocessing and cannot increase it. Telescoping $H$ rounds bounds one policy's loss change by $N\sum_r|\theta _r-\theta _0|$, using a conservative fidelity bound. Adding the bounds on candidate and incumbent losses gives $2N\sum_r vr\Delta=Nv\Delta H(H+1)$, Eq.~\eqref{cal:eq:path}. The physical model specifies the pulse-error and channel assumptions; unknown pulse errors and policy-dependent adversarial channels require an extended model.

The executed extension assumes $H=300$, a $10^{-6}T_0$ round, a $10^{-6}T_0$ probe shot, and a $0.001$ terminal logical-depolarization cost. Full acquisition duration enlarges the age allowance; correction cost is charged once terminally. These stipulated durations define the timing ratios for this simulation. Of 81 tested drift rates, 47 admit the chosen table, and every independently composed changing-angle loss is below the bound (Fig.~\ref{cal:fig:duration}). The sweep evaluates the decision boundary for one observed count; coverage under varying or dependent outcomes requires a sampling model for those outcomes.

\begin{figure*}[t]\centering\includegraphics[width=\textwidth]{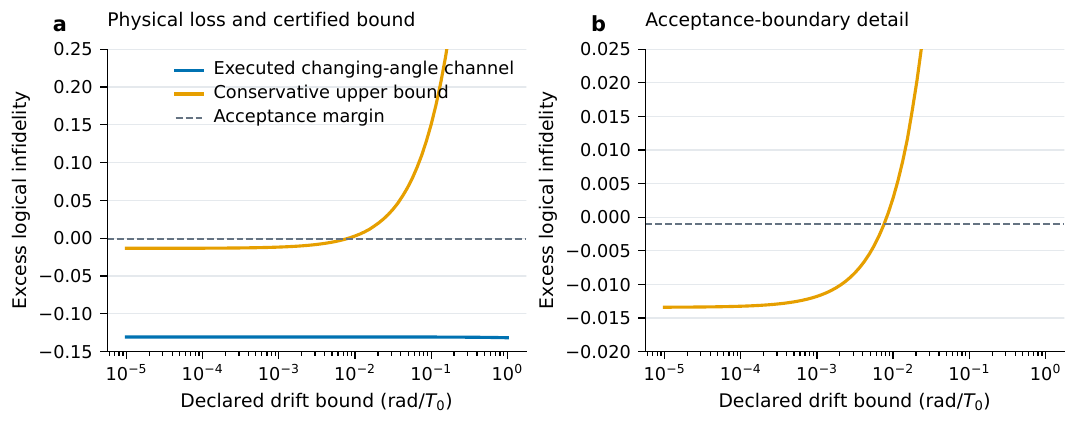}\caption{Within-block drift extension over 81 evaluated rates. (a) Numerically evaluated changing-channel excess risk and the conservative upper bound. (b) An expanded view of the acceptance boundary at excess risk $-10^{-3}$. Curves outside each displayed vertical range are clipped. The simulation specifies acquisition durations, terminal-action noise, and physical drift bounds.}\label{cal:fig:duration}\end{figure*}

\section{Training, splits, and workflow details}
\label{sm:contents:25}
\label{cal:app:learning}
The training set has 600 paths, equally divided between stationary and ramp families, with signed base magnitudes $0.04,0.07,0.10,0.13$ and magnitude jitter $\pm0.002$. Development has 200 paths at bases $0.055,0.085,0.115$. The initial test has 600 paths, equally divided among stationary, ramp, and periodic families, at bases $0.05,0.08,0.11,0.14$. Periodic drift is held out from both training and development. Whole paths and base blocks, rather than individual trajectories from a common source pool, define splits.

Each path has 25 measurement times. The predictor sees the first 24 means and an elapsed-time input feature. GRUs have hidden size 64; MLPs process per-time two-component inputs through two 64-unit hidden layers. The MLP uses the specified summary representation; a flattened full-history feed-forward model would be an additional comparator. Both families use 5 seeds, $11,23,37,51,79$, 40 Adam epochs, learning rate $0.003$, and select the checkpoint with the lowest binary cross-entropy on the development measurements available to the model. The trained target is an observed next-outcome frequency, not a hidden angle or logical risk.

The Bayesian reference uses a grid random walk, choosing its transition scale from $0.001,0.003,0.01,0.03$ on development observations. Sliding-window lengths are selected from $1,2,4,8,12,24$. In the initial 4-look test the Bayesian posterior is updated directly, whereas the other point predictors blend the forecast with the current observed mean, assigning the forecast the weight of 512 observations and the observed mean the weight of the current shot count. This limits attribution of differences to model architecture alone. The robust catalog reference was added in a documented amendment and accepts 323 of 600 paths; its mean risk is 0.0893974. On the public development distribution and tested average-risk objective, retention is the best fixed policy in the catalog.

The cost-aware extension fixes new bases $0.047,0.077,0.107,0.137$ with $\pm0.001$ jitter and 600 new paths. Its 5 prices are $0$, $2.5\times10^{-7}$, $10^{-6}$, $4\times10^{-6}$, and $1.6\times10^{-5}$. Acquisition integration uses 181 angle points and all possible binomial counts. Forecast spreads are fitted only to public development targets, subtracting estimated binomial variance and imposing a minimum standard deviation of $0.005$ in units of the observed outcome mean. No test target is used to fit uncertainty or choose a favorable price afterward.

\begin{figure*}[t]\centering\includegraphics[width=\textwidth]{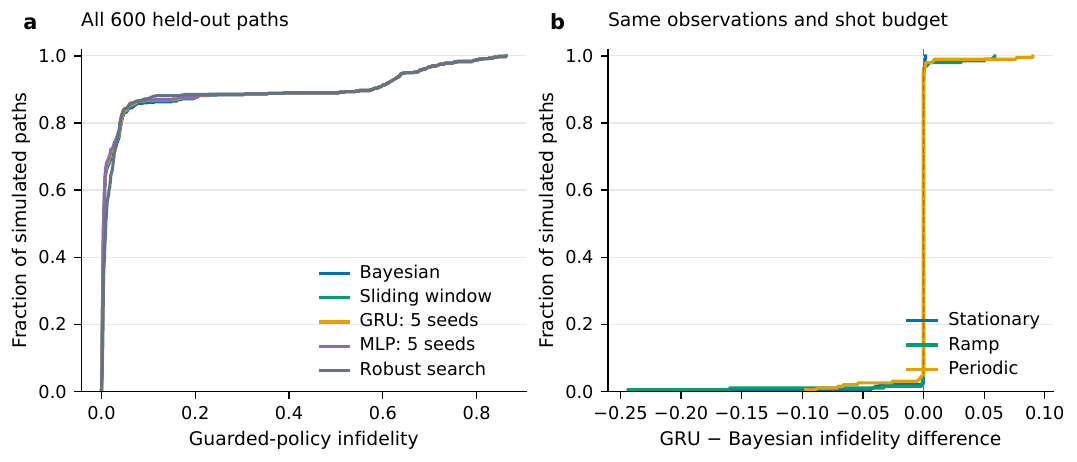}\caption{Empirical distributions on 600 held-out simulated noise paths. (a) Cumulative distributions of guarded-policy infidelity. Neural methods are gated recurrent units (GRUs) and multilayer perceptrons (MLPs); risks are averaged over 5 training seeds within each path. (b) Paired GRU-minus-Bayesian guarded-risk differences, with 200 paths per drift family; negative values favor the GRU. The curves display empirical cumulative distributions; inferential intervals are reported separately where available. All methods have the same observations, action catalog, and budget for new calibration shots.}\label{cal:fig:learningdetail}\end{figure*}

Ablations remove measurements, permute history order, or zero the elapsed-time input feature. Their output is saved as forecast sensitivity. With constant training-time gaps, zeroing that feature probes an out-of-distribution input change; learned age reasoning would require variation in training and a targeted evaluation. These predictors receive numerical measurements and time features, without the untrusted text.

The preliminary interface tests use temperature-zero Qwen inference, seed 17, a 4-turn cap and at most one additional probe. The expanded multistep protocol and its matched observation budget are specified in Supplemental Material, Sec.~\ref{sec:completion}; the central single-proposal toric protocol is specified in Sec.~\ref{sm:unified-toric}.

\section{Paired circuit risks and confirmation uncertainty}
\label{sm:contents:26}
\label{cal:app:surface}
\begin{table}[t]
\caption{Circuit confirmation. Mean plug-in worst-basis excess risk and the count of unresolved individual effects under independent final confidence boxes. Each row has 200 paths. The two deployment rows at each distance reuse those paths; distance 3 uses local-gate noise and distance 5 uses idle-$Z$ noise. Plug-in means retain finite-sample maximum bias.}
\label{cal:tab:surface}
\centering\small
\begin{tabular}{llrr}
\toprule
$d$ & Deployment & Mean excess & Unresolved\\\midrule
3 & Same channel & $-0.091455$ & 0\\
5 & Same channel & $-0.032284$ & 54\\
3 & Return to base & $-0.002561$ & 200\\
5 & Return to base & $+0.002532$ & 200\\\bottomrule
\end{tabular}
\end{table}
\begin{figure*}[t]\centering\includegraphics[width=\textwidth]{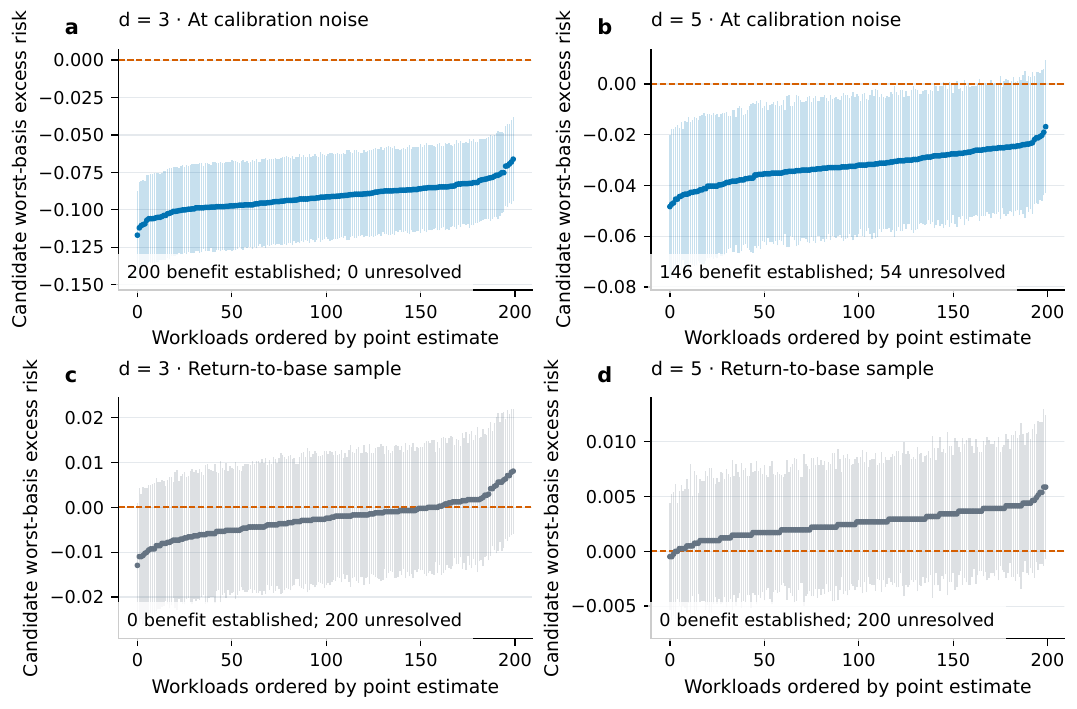}\caption{Circuit confirmation with all 200 paths per distance and deployment stage. Each point is an estimate from an independent deployment sample of worst-basis excess risk; vertical segments are the saved pointwise 95\% final risk boxes. Each panel orders paths independently by its own point estimates; paired stages must be compared through their workload identifiers. Benefit is established for 200 distance-3 and 146 distance-5 paths at the calibration settings; all 400 stale effects remain unresolved. Point estimates retain finite-sample maximum bias; harmful staleness requires a resolved adverse risk difference.}\label{cal:fig:circuit}\end{figure*}
For independent memory records in basis $b$, let $L_u,L_0\in\{0,1\}$ denote candidate and incumbent failures on the same record. Their risk difference is
\begin{equation}
 r_{u,b}-r_{0,b}=\Pr(L_u=1,L_0=0)-\Pr(L_u=0,L_0=1).
\end{equation}
Separate one-sided exact binomial limits on the two discordance probabilities yield an upper bound $d_b$. Let $r_{0,b}\in[\ell_b,h_b]$ simultaneously. Then
\begin{align}
 &\max_b r_{u,b}-\max_b r_{0,b}\nonumber\\
 &\quad\leq\max_{b\in\{X,Z\}}
 \left[d_b+\min\{0,h_b-\ell_{-b}\}\right].
 \label{cal:eq:surfacebound}
\end{align}
To see this, subtract $\max(r_{0,b},r_{0,-b})$ from each $r_{0,b}+d_b$, and maximize the remaining difference over the baseline box. The calibration error allocation covers both bases and every predeclared alternative simultaneously, with total $\alpha=0.01$. A single record supports rescoring classical matching priors; evaluating a different quantum feedback policy would require its corresponding physical trajectory. The endpoint is the maximum of the two memory-basis failure risks.

Pilot simulated base rates are $0.0017$ and $0.0023$; the prior base is $0.002$. Four families amplify measurement, idle-$Z$, local-gate, or paired faults. Independent final batches use 32768 shots per basis. The selected confirmation regimes draw simulated base rates from $[0.0018,0.0022]$ with new physical seeds, retaining the prior catalog and thresholds. The calibration test and final uncertainty calculation remain separate.

Pointwise 95\% final boxes cover the four candidate/incumbent basis risks for each path. The individual return-to-base effects remain statistically unresolved (Table~\ref{cal:tab:surface}). The estimated mean primary excess risks have bootstrap intervals $[-0.003108,-0.002010]$ at distance 3 and $[0.002362,0.002704]$ at distance 5, but the maximum of finite-sample risks is biased. The maximum-estimation bias persists under path resampling.

As a separately labeled secondary endpoint, the equally weighted $X/Z$ mean-risk difference avoids the maximum operation. Its return-to-base estimates are $+0.0004706$ (95\% interval $[0.0001977,0.0007361]$) and $+0.0015100$ ($[0.0014172,0.0015991]$). These quantify an additional average-risk diagnostic alongside the prescribed worst-basis endpoint. The full stale rule uses the valid but uninformative total-variation bound one and rejects all cases. The point estimates and universal refusal do not establish a recovery-improving update justified by old evidence in this challenge.

\section{Numerical precision and validation}
\label{sm:contents:27}
\label{cal:app:resources}
Physics validation uses float64/complex128 arithmetic. Neural inference uses float32; a comparison of numerical implementations gives a maximum logit discrepancy of $5.96\times10^{-8}$ for the same GRU checkpoint and a 200-path, 24-observation batch.

Validation covers the physical probe, continuous bounds, confidence construction, evidence binding, budgets, malformed input, circuit conversion, acquisition, and permitted maintenance operations. A separate rerun reproduces the scientific rows of five stages using fixed checkpoints, at absolute tolerance $10^{-12}$ and relative tolerance $10^{-10}$. Figures and tables are generated from numerical results with consistency checks.

\section{Theoretical context and prior-work comparison}
\label{sm:contents:28}
\label{app:theory-context}
\begin{table*}[t]
\caption{Positioning under the declared fixed-odd-$L$ toric instrument.}
\label{tab:theorem-literature}
\footnotesize
\renewcommand{\arraystretch}{1.05}
\begin{tabular}{p{0.23\textwidth}p{0.29\textwidth}p{0.40\textwidth}}
\hline
\textbf{Result} & \textbf{Closest prior work} & \textbf{Contribution and scope} \\
\hline
$[[2L^2,2,L]]$ parameters and parity-dependent $d_{\rm aff}$
& Toric geometry \cite{gottesman1997stabilizer,dennis2002topological} and
weight-$2L$ paths \cite{iverson2020coherence,wichette2026partition}
& \textbf{Standard geometry.}  $d_{\rm aff}$ denotes the cut-containment requirement for a nontrivial support contributing to a syndrome effect defined here.  The parity pattern follows established geometry. \\
Order-$\theta ^L$ corrected channel and $\alpha_L$
& Iverson--Preskill \cite{iverson2020coherence}
& \textbf{Reproduced.}  A common frame, sign, half-angle, and recovery make the
channel directly comparable with the effects. \\
Parity-controlled record dependence in one-logical-qubit codes
& Square-lattice blindness and parity-dependent leakage
\cite{bravyi2018coherent,huang2019coherent,venn2020thresholds,
behrends2025beyond}
& \textbf{Prior contrast.}  Two logical qubits here permit $\Xx\Xy$
dependence at finite order. \\
Sign ambiguity in coherent-error estimation
& Graph-state stabilizer marginal probabilities, Eq.~(6) and following discussion of Ref.~\cite{orsucci2016estimation}
& \textbf{Prior ambiguity.} Their cosine dependence loses rotation sign and period. Here the symmetry holds for complete passive toric histories. \\
Noise or channel inference from syndrome records
& Noise, detector-model, and circuit-level inference
\cite{wagner2021optimal,wagner2023logical,takou2025graphs,
takou2026nonpauli,xiao2026insitu,zheng2026learning}
& \textbf{Different target.}  They infer faults, priors, or channels.  We
distinguish encoded inputs at fixed signed $\theta $. \\
$b_{\rm rec}=2L$, exact $M_L$, and
$\eta_L(\theta )=\Theta_L(|\theta |^{2L})$
& $2L$ geometry \cite{iverson2020coherence,wichette2026partition} and effect
criteria \cite{huang2019coherent,venn2020thresholds,hu2022designing}
& \textbf{Exact specialization.} Cut cancellation and common-sign survival
attain the effect order; \cref{app:prior-comparison} maps coefficients and
path multiplicities explicitly. \\
Conditional polar recovery and phase randomization
& Resolved-instrument restoration \cite{gregoratti2004feedback} and
mixing nearby unitaries \cite{campbell2017mixing,wallman2016noise}
& \textbf{Applied methods.} Exact parity residual and model-specific application of established recovery theory, with deferred actuation available. \\
Parity-product history law
& Repeated QND statistics \cite{bauer2013repeated}
& \textbf{Applied framework.} Identify the toric pointer sectors
$\Pi_\pm$; the product mixture and count sufficiency then follow. \\
Order-one dynamics with vanishing complete-record distinguishability
& Correctability--privacy duality \cite{kretschmann2008complementarity}
& \textbf{Attained model-specific scaling.} $\Theta_L(|\theta |^L)$, with a rare-syndrome lower bound, improves on the $O_L(\sqrt{|\theta |})$ baseline from the same channel estimate. \\
Logical-input transcript privacy
& Noisy local memories \cite{shen2026transcript}
& \textbf{Different regime.} Their distance-dependent upper bounds use
assumptions on backbone blindness, locality, and noise strength; our one-round order is attained
at fixed odd $L$ under coherent rotation noise. \\
\hline
\end{tabular}
\end{table*}

The corrected logical channel and the syndrome effects address distinct observables of a coherent-error instrument.  Pauli amplitudes in one syndrome
sector interfere, and recovery maps their relative logical classes to a
syndrome-conditioned logical map.  Coherent logical channels under quantum
error correction have been studied for
repetition, stabilizer, surface, and toric codes
\cite{greenbaum2018modeling,bravyi2018coherent,beale2018decoheres,
huang2019coherent,iverson2020coherence,venn2020thresholds,
marton2023coherent}.  Stabilizer records support in-situ noise characterization
\cite{combes2014insitu,orsucci2016estimation,wagner2021optimal,
wagner2022pauli,wagner2023logical}, detector-model inference
\cite{blumekohout2025estimating,takou2025graphs,takou2026nonpauli}, and
circuit-level logical-noise estimation
\cite{xiao2026insitu,zheng2026learning,takou2026logical}.
Syndrome-conditioned evolution can also generate logical unitary ensembles
\cite{cheng2025designs}.  Odd-distance rotated square-lattice planar surface codes with
even-weight relevant stabilizers and odd-weight logical representatives
have input-independent syndrome probabilities despite conditioned
logical rotations
\cite{bravyi2018coherent,cheng2025designs,behrends2025beyond}.  Other parity
classes reveal logical polarization \cite{huang2019coherent,venn2020thresholds},
while diagonal-gate methods quantify state--syndrome correlations
\cite{hu2022designing}.  Input-agnostic branches also support repeated
syndrome-resolved logical-gate protocols
\cite{huang2025robust,yoshioka2025syndrome}.  We instead determine when the
complete syndrome record of the periodic instrument in Sec.~\ref{sec:model} first depends on
the encoded input.

Orsucci, Tiersch, and Briegel explicitly identify the ambiguity
$\lambda_a\mapsto\pm\lambda_a+2k\pi$ in the probability differences of individual graph-state stabilizer outcomes, since these depend on $\cos\lambda_a$~\cite{orsucci2016estimation}, Eq.~(6) and the paragraph following it. This is a direct antecedent of the parameter-identifiability question. Our statement concerns a different observation object: the full joint law of arbitrarily long passive histories for the declared toric instrument. Its proof uses the corrected Kraus operators and their common eigenbasis; marginal symmetry alone would not establish that history statement.

We answer for a periodic square toric code with full-support coherent $X$
rotation, complete ideal $Z$-syndrome extraction, and deterministic
minimum-weight $X$ recovery.  Stochastic toric-code work found the parity of
the minimum even logical representative \cite{wichette2026partition}; here
that scale controls a coherent, input-dependent syndrome effect.  Earlier work establishes the underlying $L$-versus-$2L$ pattern.  For essentially this setting,
Iverson and Preskill derived the leading order-$\theta ^L$ corrected-channel
coefficient (our $\alpha_L$ in magnitude) and weight-$2L$ double-logical paths
\cite{iverson2020coherence}.  We reproduce the established channel onset, coefficient, and double-logical paths in common conventions to compare them with the same instrument's effects.
Our central quantitative result is the exact record-effect onset and its
nonzero coefficient.  The CSS cross terms of Hu, Liang, and Calderbank
\cite[Eq.~(91)]{hu2022designing} become our effect-character expansion;
the toric cut constraint then cancels forbidden supports and makes all
minimum surviving contributions have a common sign.  Our exact support
count is a specialization of the seam-linked path enumeration in
Ref.~\cite[Appendix~H]{iverson2020coherence}, with nonsimple supports
removed and multiplicities fixed (\cref{app:prior-comparison}).
The count specializes standard enumeration to the tagged supports used in the effect calculation.  
Recent adjacent work analyzes logical-noise coherence and information-theoretic
performance, projected logical ensembles, syndrome-distribution phase
structure, correlated coherent noise, and decodability transitions with
monitored-dynamics duals
\cite{behrends2025beyond,bejan2026projected,liu2026transition,
rajmohan2026correlated,yang2026decoding}.  These works supply operational and
many-body context for the fixed-instrument calculation.
The recent transcript-privacy preprint of Shen and Zhong
\cite{shen2026transcript} addresses the same logical-input target for noisy
single-logical-qubit stabilizer memories, including rotated surface codes.
Its Theorem~1 bounds logical components under explicit backbone blindness,
locality, and smallness hypotheses with growing distance and
$T=\Theta(d)$.  Here the geometry is a two-logical-qubit torus, the noise is
a coherent rotation, and $L$ is fixed before $\theta \to0$; we attain the
one-round order and bound the growing-history distance.  Each comparison applies under its own limiting regime and hypotheses.

The physical context also includes the structure of the stored quantum information and its coupling to noise. Otten et al.~\cite{otten2021impacts} compare information lifetimes in qubit- and qudit-based memories under different physical noise models. At the control level, spectator modes can reduce the fidelity of pulses optimized for an isolated qudit~\cite{ozguler2024spectator}. These dependencies motivate specifying the memory and control environment when interpreting a calibrated recovery model.

\subsection{Decoding cost, changing records, and measurement advantage}
\label{sm:decoding-context}
\emph{Scope of the computational comparison.}
Finding a likely physical error, selecting the most likely logical class after summing degenerate errors, simulating a memory, and certifying an update are different tasks. Optimal degenerate stabilizer decoding has counting-complexity hardness under explicit input and promise conditions~\cite{iyer2015hardness}. Surface-code hardness also holds with arbitrary qubit-dependent Pauli probabilities as inputs~\cite{fischer2024hardness}. These worst-case results coexist with efficient structured instances; establishing hardness under a particular drift model requires a separate analysis.

The two efficient results of Bravyi and collaborators also concern different tasks. Exact maximum-likelihood decoding in Ref.~\cite{bravyi2014efficient} has $O(n^2)$ cost for independent bit and phase flips with noiseless checks. Its more general matrix-product approximation has $O(n\chi^3)$ cost at bond dimension $\chi$; accuracy must be checked as $\chi$ changes. The coherent-error simulation of Ref.~\cite{bravyi2018coherent} treats single-axis rotations and noiseless checks and recovery. Its runtime includes the chosen classical decoder's cost in addition to the polynomial simulation cost. Efficient simulation therefore does not supply a general optimal decoder.

Our toric evaluation exploits two-term single-qubit Pauli expansions and a shared eigenbasis of corrected logical Kraus operators. Channel composition multiplies their matrix-element factors. Constructing those factors requires syndrome sums; the present exhaustive enumeration grows exponentially when extended directly. The evaluator ranks actions for both controllers, enabling independent outcome checks. Scaling the certificate requires an evaluation method with controlled approximation error, while the observation obstruction determines when further measurements are needed.

\emph{Noisy records and drift.}
Noisy syndrome extraction requires inference over a history that contains both data faults and measurement faults. Correlations and drift can change the appropriate history model. A decoder can remain fast while its assumed probabilities become inaccurate; selected non-Markovian noise models exhibit degraded memory performance with a decoder using an independent detector-error approximation~\cite{kam2025nonmarkovian}. This is a statement about the tested noise--decoder pair. Longer records can improve estimation when the model remains applicable, while temporal averaging can obscure changes. An observation ambiguity, such as the sign symmetry proved here, persists even with exact inference in the declared model.

Recent adaptive approaches address different parts of this problem. Bhardwaj et al. infer drifting detector-model parameters from syndrome data using overlapping windows, exposing the filtering and sampling tradeoff~\cite{bhardwaj2026adaptive}. Sivak et al. demonstrate reinforcement-learning control of QEC under experimental drift~\cite{sivak2026reinforcement}. Their detector-based controller steering and logical-error-rate-based decoder steering use different feedback: the latter requires the additional outcome information used in their memory experiments. Calibration-conditioned neural decoding has also been evaluated on repetition-code experiments~\cite{stein2026calibration}. QAdapt combines neural pre-decoding with matching under varying noise~\cite{miao2026qadapt}; its reported residual-matching timings exclude the neural and transfer stages, which must be included in a complete latency assessment.

Pancotti, Saravanan, and Svore study exact tensor-network likelihoods for learning quantum noise~\cite{pancotti2026}. Their demonstrations use synthetic samples, including samples from a hardware-characterized detector model, with both syndromes and terminal logical-outcome labels. Their drift study uses synthetic changing noise. Exact contraction has a treewidth-dependent cost, and an expressive model's population optimum is distinct from finite-data identifiability and reliable optimization. These observation and computation requirements matter when comparing their task with passive syndrome-only inference.

Fast decoding likewise needs a precise timing metric. Sparse Blossom provides efficient matching~\cite{higgott2025}. AlphaQubit~\cite{bausch2024} and AlphaQubit~2~\cite{senior2026scalable} establish learned decoding as a substantive task distinct from language-model maintenance advice. AlphaQubit~2's per-cycle streaming throughput and final-result latency are different quantities; its smaller-distance real-time results and larger-distance accuracy results are different operating points. The recent parallel Sparse-Blossom analysis~\cite{mikami2026sparse} obtains subconstant average time per round by amortizing parallel work over a distance-sized window below threshold. It uses growing parallel resources and a circuit-depth model without geometric communication overhead. For a timed acceptance rule, acquisition, inference, communication, and actuation all contribute to evidence age, regardless of which decoding algorithm is used.

\emph{What the syndrome-aware quantum advantage establishes.}
Tsubouchi et al.~\cite{tsubouchi2026advantages} define an effective error rate through the information available for estimating a logical observable. Their classical comparison fixes the logical measurement basis and uses the syndrome in subsequent processing; the quantum protocol conditions the logical measurement on the syndrome. The quantum protocol thus changes the accessible measurement rather than merely the algorithm operating on a fixed classical transcript. Their Pauli-channel framework accommodates complete syndrome histories of Clifford circuits, but the explicit asymptotic advantage has narrower hypotheses. Theorem~2 concerns independent local Pauli noise, even-distance blocks, a low-noise limit, and a Haar average over pure logical states; the improvement then scales with the number of encoded blocks. Odd-distance finite-noise examples are assessed separately. The protocol requires state characterization and generally nonstabilizer logical measurements, with corresponding control resources.

Extending that advantage to unknown drift requires a time-dependent model, an estimation target, and calibration and measurement costs. Efficient classical decoding is compatible with quantum computational advantage, and syndrome-conditioned measurements can improve estimation despite efficient classical processing. Our security evaluation addresses a different operational question: which physical evidence supports a recovery update when observations omit relevant information or become outdated?

\section{Calibration, activation, and circuit validation}
\label{sm:contents:29}
\label{sec:completion}
The following tests examine three requirements for a recovery update: a calibration that constrains the physical noise, a risk bound valid until deployment, and application of the recovery table that was evaluated. We first measure readout uncertainty, then audit acceptance over every count and catalog action. Circuit tests and controller workloads assess the resulting recovery improvement and completion rates.

\subsection{Why the proposer need not be accurate}
Let $E_j$ be the event that confidence region $C_j$ contains every physical condition relevant to deployment $j$. Suppose $\Pr(E_j^c)\le\alpha_j$, uniformly over the allowed nuisance parameters, and the record and activated action are correctly bound. On $E_j$, any action satisfying $U_{j,u}\ge\sup_{\theta \in C_j}D_u(\theta )$ and $U_{j,u}\le-\delta$ has $D_u(\theta _j)\le-\delta$. This implication holds simultaneously for every allowed action, including an action selected after seeing the data. Therefore an accepted violation implies $\bigcup_jE_j^c$, whose probability is at most $\sum_j\alpha_j$. Repeated queries against one region do not require an action-by-action statistical penalty; new acquisitions consume the declared family allowance. A new confidence construction after data-dependent stopping would require its own justification.

The argument imposes no accuracy or calibration condition on the proposer. It requires correctly associated measurement records, valid physical bounds, a correct risk evaluation, and application of the evaluated action. It controls the unconditional probability of violating the promised excess-loss margin across the declared experiment family. Beneficial completion is evaluated separately from protection of accepted actions. Repeated invalid proposals, refusal to acquire evidence, and budget exhaustion can leave the incumbent in place. The implementation caps both evidence acquisition and proposal work so that rejected requests are not a free, unbounded service resource.

\subsection{Measuring readout nuisance parameters}
The signed sentinel uses mean $b-a\sin \theta $ with stipulated contrast and offset bounds. To measure these quantities in the simulated observation model, prepare known $+Y$ and $-Y$ references and measure them through the same stationary readout channel. Their plus-outcome probabilities satisfy
\begin{equation}
 r_+=(1+b+a)/2,\quad r_-=(1+b-a)/2.
\end{equation}
For independent reference counts with Clopper--Pearson intervals $[l_+,u_+]$ and $[l_-,u_-]$, a simultaneous outer rectangle is
\begin{align}
 a&\in[l_+-u_-,u_+-l_-],\\
 b&\in[l_++l_--1,u_++u_--1].
\end{align}
Intersect contrast with its physical upper limit one and reject if a strictly positive contrast is not established. The rectangle is conservative because it discards correlations between the inferred nuisance parameters. For probe interval $[l,u]$, invert $2[l,u]-1=b-a\sin \theta $ over all rectangle corners, then intersect only with the declared identifiable angle domain using the inversion procedure in Sec.~\ref{cal:sec:rule}. Invalid probabilities and unsupported domains fail explicitly. Divide the family allowance equally among the two reference intervals and the signed-probe interval. By the union bound their composite capture set has failure probability at most $\alpha$. Transfer and age expansion are subsequently applied as before.

The calibration estimates readout nuisance parameters using known reference states; preparation quality requires independent characterization. The construction assumes known reference states, a common readout channel, conditionally independent shots, and stationarity across all 3 acquisition blocks. Real acquisition durations require the additional physical allowances already described in the calibration protocol. Here the blocks are simulated at one capture angle.

The simulation design uses contrasts/offsets $(0.99,0)$, $(0.97,0.01)$, and $(0.94,-0.02)$; magnitudes $0.06,0.10,0.14$ rad with alternating signs; 200 independent paths per magnitude, setting, and budget; and reference budgets $512,2048,8192,32768$ shots per state. Every signed probe has 8192 shots. Thus there are 7200 composite records with total costs $9216,12288,24576,73728$ shots. All candidate tables use $H=300$ and $\delta=0.001$. The true simulated angle is used to generate measurements and assess recovery outcomes. Action selection and acceptance use only the measurement-derived uncertainty region.

Nuisance calibration from simulated reference counts yields 2323 beneficial activations, 0 observed harmful activations, and 3 calibration-angle noncoverage events. By readout setting the beneficial-activation counts are $994,740,589$ out of 2400 each. A paired legacy comparator uses the same signed counts and the original fixed nuisance rectangle: it gives $1498,1537,1449$ beneficial activations and $16,83,470$ noncoverage events, respectively. The latter two settings violate its nuisance premises. No harmful activation is observed in these particular legacy workloads, but the lost confidence coverage prevents interpreting their higher acceptance as justified protection. Measuring nuisance parameters costs shots and can sharply reduce acceptance of beneficial operations. Figure~\ref{fig:completion-calibration} quantifies this tradeoff.

\subsection{An audit allowing arbitrary catalog advice}
For one 8192-shot capture and $\alpha=0.01$, enumerate every count $k=0,\ldots,8192$. For each count and all 13 actions, including retention, recompute the permitted deployment region and certification rule. For each physical setting, mark the count if \emph{any} certified action would violate the promised margin. Numerically summing the binomial probabilities of all such counts gives the worst-selection violation probability, without fitting a particular attacker or AI model.

The audit uses 141 capture angles from $-0.14$ to $0.14$, contrast endpoints $0.98,1$, offset endpoints $\pm0.001$, and drift radii $0,0.005,0.02$ rad. For each radius the simulated transfer/drift shift takes three values: zero and $\pm(\text{radius}+0.001)$. There are 5076 physical settings. The maximum violation probability is $0.003787201138615698$, below the declared $0.01$. A violation is any excess loss above the promised $-0.001$ margin, including a smaller-than-promised improvement. This audit exhaustively sums binomial probabilities and the specified numerical channel bound, so sampling confidence intervals are inappropriate for these curves. The audit exhausts counts and catalog choices on its finite physical grid using the stated floating-point treatment; the analytical guarantee separately requires valid continuous risk bounds.

Figure~\ref{fig:completion-calibration}(a,b) uses 600 independent records per cost point and 2400 per readout setting. Its pointwise 95\% Hoeffding intervals cover averages of independent, possibly nonidentical Bernoulli outcomes. Panel (c) maximizes the numerically evaluated binomial margin-violation probability across the specified nuisance parameters and deployment shifts at each calibration angle, allowing selection among all certified catalog actions.

\subsection{Binding the action at activation time}
The software represents the certified validity period as a single-use activation authorization. This authorization binds a workload, epoch, evidence identifier and digest, reference digest, action identifier and digest of the action array to be applied, issue time, and expiry. Evidence records are service-owned; sampled composite calibration records include both reference counts, signed-probe counts, and the derived interval. The reference digest covers the physical configuration, enumerated support counts, channel coefficients, derivative bounds, cached risks, and all other reference arrays. The digests identify the evidence, reference arrays, and recovery table checked within the evaluator process; signatures and post-quantum execution proofs would require additional protocols.

Issuance verifies a risk bound through the requested deployment deadline. Activation takes a lock, checks that the authorization remains unused, checks workload and epoch, verifies current evidence/reference/action identity, checks clock order and expiry, reruns the current-time risk test, and copies the evaluated action array while still inside the critical section. This prevents the in-process race in which a checked catalog index points to a replaced action by activation time. A concurrent duplicate request permits 1 activation only. Epoch changes invalidate outstanding authorizations but do not reset cumulative acquisition allowance. Each new acquisition is charged its share of the family budget, including a measured reference block that fails to establish positive contrast; such a failure cannot fall back to the fixed-nuisance rule. Proposal work also has a finite cap.

The 13 deterministic conformance cases exercise honest activation, benign waiting, expiry, forged action or reference fields, replacement of action/reference/evidence, replay, epoch change, clock rollback, concurrent replay, and measured-nuisance evidence. The 13 cases are deterministic functional conformance checks. Their scoring evaluates the channel formula using the returned action array. The prototype enforces activation within a process and clock assumed uncompromised; external identity and acknowledgment of physical actuator execution require additional interfaces. The noise is stationary during the coherent recovery block after activation. 

\subsection{A sufficient validity period for a fixed circuit-noise family}
Consider a fixed stochastic fault circuit and noise family parameterized by a scalar rate $p$. Each elementary fault has distribution $Q_i(p)$ satisfying
\begin{equation}
 \TV(Q_i(p),Q_i(p'))\le c_i|p-p'|.
\end{equation}
For a Bernoulli Pauli fault with probability $s_i p$, $c_i=s_i$. For a depolarizing channel with total nonidentity probability $s_i p$, the same coefficient applies, provided its conditional nonidentity distribution remains fixed. A one-qubit Pauli channel contributes the sum of its Pauli-probability slopes. Independent targets count separately; one explicitly correlated multi-target fault is one random event. Coupling each fault, taking the union bound on any mismatch, and then applying the deterministic circuit/decoder map gives
\begin{equation}
 \TV(P_{\rm dep},P_{\rm cap})\le
 \min\{1,\sum_i c_i|p_i^{\rm dep}-p_i^{\rm cap}|\}.
 \label{eq:completion-coupling}
\end{equation}
The guarantee assumes the conditional fault distributions and stochastic independence structure of the declared family.

If a common drift bound $|\dot p|\le v$ holds from calibration to the end of a block of duration $T$ starting at age $A$, each difference is at most $v(A+T)$. Define $K=\sum_i c_i$. Every fixed decoder's binary loss expectation changes by at most $\varepsilon=\min\{1,Kv(A+T)\}$. The maximum over the two basis risks is 1-Lipschitz in their sup norm. Applying the loss bound separately to candidate and incumbent therefore gives a $2\varepsilon$ allowance for worst-basis excess risk. When $v>0$ and a nonempty certified interval exists, the maximum certified age at deployment is
\begin{equation}
 A_{\max}=\frac{-\delta-U_{\rm cap}}{2Kv}-T.
 \label{eq:completion-max-age}
\end{equation}
A negative value means no deployment window. The total-variation function also handles zero drift and saturation at one explicitly; the finite-age solver requires strictly positive drift. Standard coupling supplies this conservative freshness allowance.

The fixed six-decoder catalog and 30-round circuits are held fixed. The confirmed distance-3 family amplifies local gate noise by 8; distance 5 sets idle-$Z$ fault probability to $8p$, which is 24 times the base $p/3$ idle-$Z$ probability. Full-precision generated circuits give worst-basis coefficients $K_3=3911.5$ and $K_5=10101.5$. Each coefficient sums the slopes over all fault locations in the complete memory circuit. These coefficients differ in both code distance and noise family, so their comparison does not isolate distance scaling. We independently reconstruct all 400 calibration-time candidate selections and bounds from 800 archived packed failure arrays; all selections match and the maximum bound discrepancy is zero.

The 4 declared rates, $10^{-9},10^{-8},10^{-7},10^{-6}$ per $T_0$, are evaluated at 161 ages for every calibration record. At $v=10^{-8}T_0^{-1}$ and $T=30T_0$, all 200 distance-3 records and 199 of 200 distance-5 records admit a nonempty deployment window. At $10^{-7}$ none of the distance-5 records does; at $10^{-6}$ neither family does. The simulation treats these rates as declared scenario parameters. The scalar-family extension establishes freshness within that family; the original profile-changing challenge retains its vacuous bound.

Outcome validation with new samples selects the first 8 recorded workloads at each distance and tests ages $0$, $A_{\max}/2$, and $2A_{\max}$ using $v=10^{-8}$. At deployment, $p_{\rm dep}=p_{\rm cap}+vA$; each new block is stationary. There are 8192 shots in each of the $X$ and $Z$ bases at each of 48 instances: 786432 new complete memory shots. Candidate and incumbent decode identical records, and a separately implemented measurement-to-detector parity conversion is checked on every sample. Four Clopper--Pearson risk boxes with allowance $0.05/4$ yield a conservative pointwise 95\% interval for the difference of worst-basis risks. All 32 within-window instances are accepted and their independent interval lies below $-0.001$. All 16 beyond-window instances are rejected, although their independent intervals also establish benefit. The result demonstrates a sufficient, conservative freshness condition. Expiry marks loss of sufficient certification, while an expired action may remain beneficial. Figure~\ref{fig:completion-circuit} displays all validation intervals.

The age curves in Fig.~\ref{fig:completion-circuit} deterministically reevaluate the same 200 calibration records per distance at each age. The maximum-age distribution keeps denominator 200, including the distance-5 record with no certified window. The 4 drift rates are stipulated simulation parameters in probability per $T_0$.

\subsection{Language-model and deterministic controllers}
The controller evaluation comprises 192 independently generated workload blocks, 24 per condition: honest, insufficient initial evidence, misleading advice, bounded staleness, benign delay, workload substitution, nonce replay, and a physical step violating the drift bound. Horizons are $100,300,600$ rounds; initial budgets are $4096,8192,16384$ shots, except for the 512-shot insufficient-evidence condition. Each agent may request 1 additional 8192-shot acquisition. Both controllers receive the same public evidence, certified-action menu, untrusted report, and action/acquisition budget in each paired block. The 4 malicious-report variants are allocated across the advice attacks. Hidden angles and evaluator losses are unavailable to the agent.

Each block has Qwen3:8b and deterministic state-machine rollouts with authorization alone and with the full gate. This is 384 local-model rollouts, 384 deterministic rollouts, and 668 model turns. There are no Hypertext Transfer Protocol (HTTP) or response-parsing errors in this expanded run. Among 168 workloads satisfying the assumptions the full-gate large language model (LLM) gives 159 beneficial updates, 0 observed harmful activations, 6 incomplete workflows, and 3 completed retentions. Authorization alone gives 142 beneficial updates and 26 harmful activations. The state machine gives 162 beneficial updates and 0 observed harmful activations under either gate. Additional acquisition costs over these blocks are 1,007,616 shots for the full-gate LLM, 851,968 for authorization-only LLM, and 860,160 for each deterministic variant. The AI controller completes beneficial updates, while the deterministic reference completes more such updates at lower acquisition cost.

The 24 intentionally invalid drift-premise blocks produce 16 harmful full-gate LLM activations and 21 harmful full-gate deterministic activations. These are retained as failures. Pointwise 95\% Chernoff--KL intervals in Fig.~\ref{fig:hero} allow independent, nonidentical Bernoulli means across the fixed horizon and budget allocation. Each independent workload block supplies 4 paired rollouts. The harm intervals and outcomes characterize the specified attack set under its stated physical premises. 

\subsection{Numerical validation}
The numerical evaluation uses fixed random seeds. Validation checks confidence inversion, invalid inputs, action/evidence replacement, timing, concurrency, budget use, elementary coupling, and age-bound solving.

Circuit coefficients and independent validation samples use full-precision parameters. A separate comparison verifies instruction identities and target ordering; numerical arguments agree up to the expected rounding in a lower-precision representation. Evaluation of the channel formula uses the stated numerical tolerance.

\subsection{Scope of deployment}
The computational result applies to the stated observation and noise models, ideal toric actuation, and an in-process evaluator whose memory and clock are assumed uncompromised. Deploying it requires characterized reference preparation, readout, probe transfer, and drift, together with verification that the physical actuator applies the evaluated correction. Wider noise families and within-block coherent drift require corresponding risk bounds. Cryptographic authentication and formal verification of numerical enclosures would strengthen the implementation boundary. The evaluated controllers choose among fixed probes and recovery tables; learned probe selection and matched neural/Bayesian circuit control remain possible extensions.

\section{Detailed empirical panels supporting the quantitative overview}
\label{sm:contents:30}
\label{sec:hero-support}
The following figures expand the information checks, calibration-budget maps, and controller comparisons in Fig.~\ref{fig:hero}. The paired-comparison methods below also support Fig.~\ref{fig:support-security}.

\noindent\textit{Information and recovery.}
Figure~\ref{fig:support-risk} uses the exact $L=3$ square toric instrument with $R_x(\theta )=e^{-i\theta X/2}$ and the syndrome measurement and recovery of Sec.~\ref{sec:model}. Panel (a) orders all 256 syndrome probabilities for $\rho=I/4$ using the same ranking at both signs. Every probability is evaluated and saved; every fifth negative-sign value is marked for legibility. Panel (b) independently evaluates the effects at 361 angles and four logical $X$ sectors, with normalization checked separately. These residuals test the implementation of the exact effect and history identities. The phase tables in panels (c,d) stay fixed at their $\pm0.10$ calibration angles; logical actuation follows the assumptions of Sec.~\ref{sec:model}.

\noindent\textit{Decision maps.}
In Fig.~\ref{fig:support-contract}, the proposal is the $+0.10$ table at $H=300$ and the required infidelity improvement is $\delta=0.001$. Each budget $n$ uses the representative count $k=\operatorname{round}\{n[1-0.99\sin(0.10)]/2\}$. Figure~\ref{fig:hero} uses $k=3691$ and $n=8192$. A 99\% Clopper--Pearson count interval is inverted over contrast $[0.98,1]$ and offset $[-0.001,0.001]$, then expanded by $0.001+vA$ rad. Escape from the supported domain rejects the proposal. Certification uses a continuous numerical bound with the tested $10^{-8}$ allowance. Four $121\times121$ maps vary age and drift; a slice at fixed drift rate evaluates 161 budgets and 121 ages at $v=0.08$. The evaluator classifies rejected operations along the illustrative bounded path $\theta _{\rm dep}=0.10-\min(vA,0.20)$, with excess risk within the chosen tolerance interval $[-\delta,\delta]$. This path supplies the simulated outcome colors, not information available to the controller. Each map is a decision boundary conditional on its representative record, rather than an acceptance probability. Age is expressed in $T_0$ and angular drift in rad/$T_0$ for this supporting scenario.

\noindent\textit{Paired mechanism comparison.}
Figure~\ref{fig:support-security} uses saved paths with the same measurements and proposal for every baseline. Authorization checks their binding. Confidence-only additionally tests the calibration-time risk; fixed expiry uses a $T_0$ age cutoff without that risk test; the full safeguard propagates bounded drift. Each full-gate decision either applies the proposed operation or retains the incumbent. The scatter includes overlapping points without jitter, and the empirical cumulative distributions are unsmoothed. Harm intervals are pointwise two-sided 95\% Clopper--Pearson intervals over 200 independent paths per condition. Aggregate retained gain is the ratio of total deployed to total candidate gain; its interval resamples matched paths with 10,000 bootstrap replicates. These condition-specific intervals differ from the family allowance in the acceptance theorem.

\noindent\textit{Controller comparison.}
Figure~\ref{fig:support-agent} resolves the 192 independent workloads into 8 conditions, with 24 workloads each. Each supplies 4 paired rollouts: Qwen3:8b or a state machine, under authorization alone or the full gate. Thus the evaluation contains 384 local-LLM and 384 deterministic rollouts. Horizons are 100, 300, and 600 rounds; initial budgets are 4096, 8192, or 16384 shots, except for 512 in the small-sample condition. Each controller may request one additional 8192-shot block. Both receive identical public evidence, certified-action menus, and untrusted reports. The 4 malicious-report variants are distributed across the 24 advice attacks. Beneficial completion and harmful activation use all workloads as denominators, including retention and incomplete episodes. Conservative pointwise 95\% Chernoff--KL intervals allow independent, nonidentical Bernoulli means across the fixed horizon/budget allocation. The evaluator computes recovery-risk bounds and checks the acceptance criterion; the controller manages acquisition and action selection.

\begin{figure*}[t]
\centering
\includegraphics[width=\textwidth]{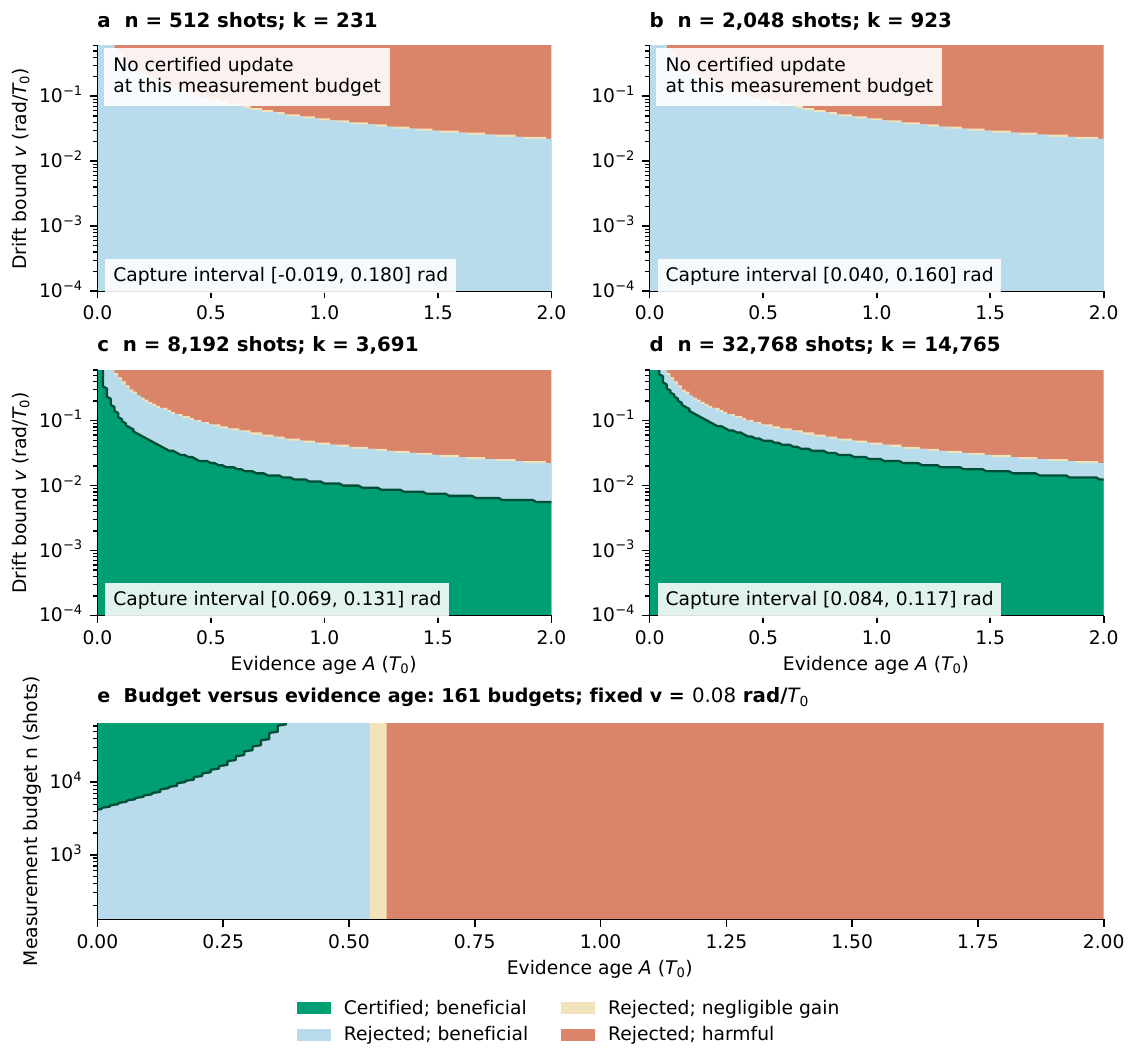}
\caption{\textbf{Certification depends on measurement budget, evidence age, and drift.}
(a--d) Conditional decision maps at 4 shot budgets. (e) Budget--age section at $v=0.08$. Colors distinguish certification from the evaluator's simulated outcome classification. Rejected beneficial updates quantify foregone opportunity: the controller must recalibrate, select another justified action, or retain the incumbent. The methods above specify the representative records and deployment path.}
\label{fig:support-contract}
\end{figure*}

\begin{figure*}[t]
\centering
\includegraphics[width=\textwidth]{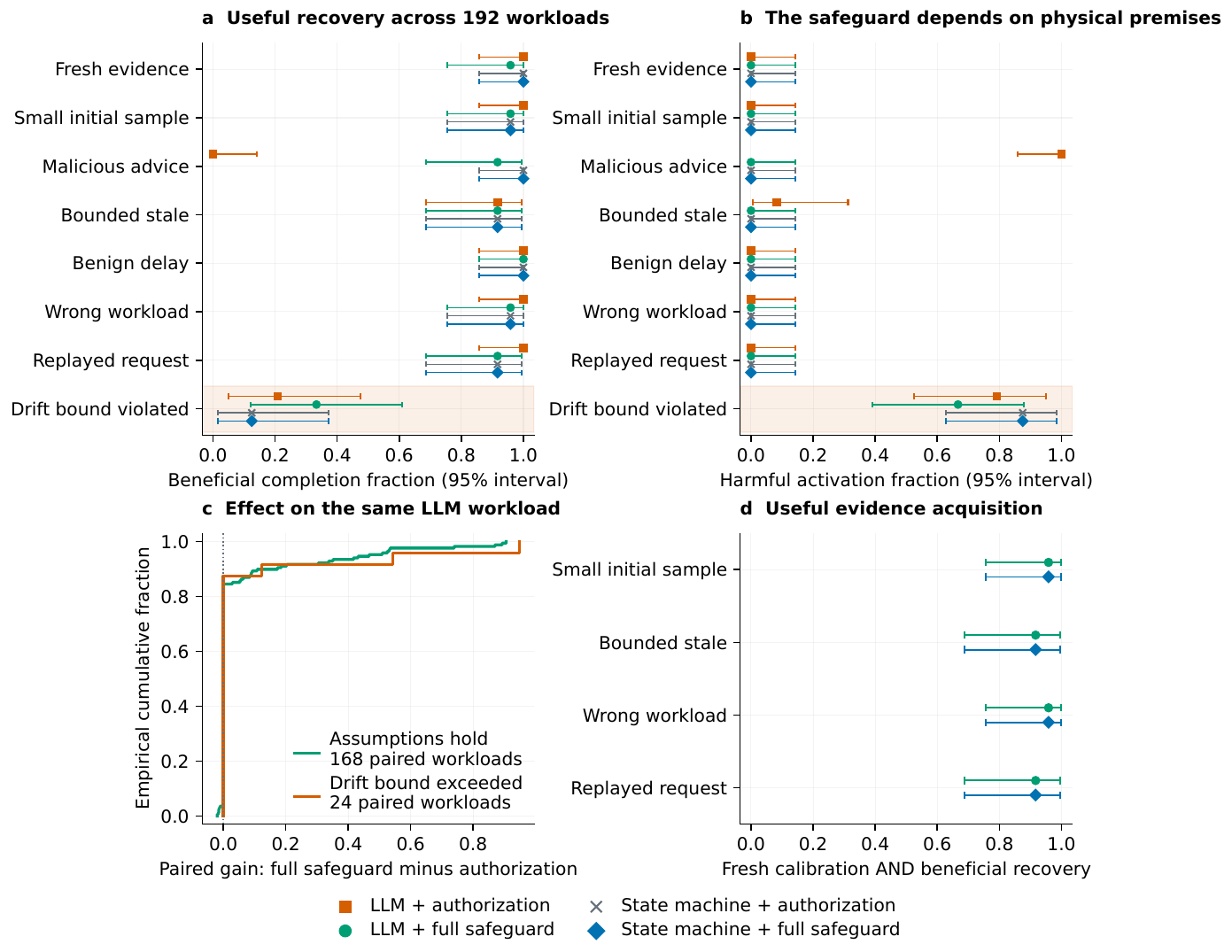}
\caption{\textbf{Large language model (LLM) participation and a matched deterministic baseline.}
(a,b) Beneficial completion and harmful activation by condition, with pointwise 95\% intervals. (c) Paired change in LLM recovery gain. (d) Additional acquisition followed by beneficial recovery under the full gate. In 168 workloads satisfying the assumptions, the full-gate LLM gives 159 beneficial completions and 0 observed harmful activations; authorization alone gives 142 and 26, respectively. The state machine gives 162 beneficial completions and 0 observed harmful activations under either gate. Exceeding the drift bound causes 16/24 harmful full-gate LLM activations.}
\label{fig:support-agent}
\end{figure*}

\section{Circuit definitions, measurement conventions, and shared schedules}
\label{sm:circuits}
\label{sm:contents:31}
The circuits in Fig.~\ref{fig:intuition} distinguish the coherent toric instrument from the stochastic surface-code extension. Figure~\ref{fig:circuit-definitions} gives the parity-measurement and probe conventions. The following protocol specifies the memory circuits and their noise parameters.

\subsection{Toric cycle and logical actuation}
Each round begins in the codespace of the periodic square toric code with $2L^2$ edge qubits and two logical qubits. Apply $R_X(\theta )=\exp(-i\theta X/2)$ on every edge, measure a complete independent set of $Z$ plaquettes, and apply $C_s=X(r_s)$ with the minimum-weight and lexicographic tie-breaking rule of Sec.~\ref{sec:model}. A weight-four plaquette can be read using a newly initialized $|0\rangle$ ancilla, four controlled-NOT (CNOT) gates with data qubits as controls and the ancilla as target, and a computational-basis measurement. For ancilla result $m$, the data projector is $[I+(-1)^m Z_1Z_2Z_3Z_4]/2$. Measuring all plaquettes adds one redundant parity; the model retains $L^2-1$ independent bits. The star sector is initially $+1$ and remains so under the stated $X$ operations.

With error-free gates and readout, this ancilla circuit realizes the stated projector. The exact toric calculation uses the projectors directly; it assigns no duration or noise to the individual extraction gates. The logical operators $K_s(\theta )$ include $C_s$ and the return to the fixed logical frame. The optional $V_s(u)$ then applies the inverse-polar phase table determined by candidate angle $u$. Its dashed box denotes assumed ideal logical actuation. The calculation evaluates the logical channel directly; a physical decomposition of this syndrome-conditioned logical operation is outside that implementation. The passive-sign theorem concerns the fixed instrument before this additional phase correction.

\subsection{Signed probe and measured readout}
The separate sentinel probe starts in $|0\rangle$ and undergoes $R_X(\theta )$. To measure $Y$, apply $S^\dagger$ followed by $H$, then measure $Z$, with $S=\operatorname{diag}(1,i)$. Outcome $m=0$ represents $Y=+1$, and $m=1$ represents $Y=-1$. With error-free preparation and readout, the mean is $-\sin \theta $. Characterized readout changes the mean to $b-a\sin \theta $, so the probability of the reported positive outcome is $(1+b-a\sin \theta )/2$.

Reference states $|+Y\rangle$ and $|-Y\rangle$ are prepared by $H$ followed by $S$ or $S^\dagger$, respectively. The same readout gives positive-outcome probabilities $r_\pm=(1+b\pm a)/2$. These reference preparations and the physical relation between probe and protected data are assumptions to characterize experimentally. The simulations sample the stated response model. Arbitrary-angle $R_X(\theta )$ is evaluated directly with complex-valued rotation matrices, rather than approximated by a Clifford gate in Stim. In Stim syntax, \texttt{RX} means an $X$-basis reset, not this coherent rotation.

\begin{figure*}[t]
\centering
\includegraphics[width=\textwidth]{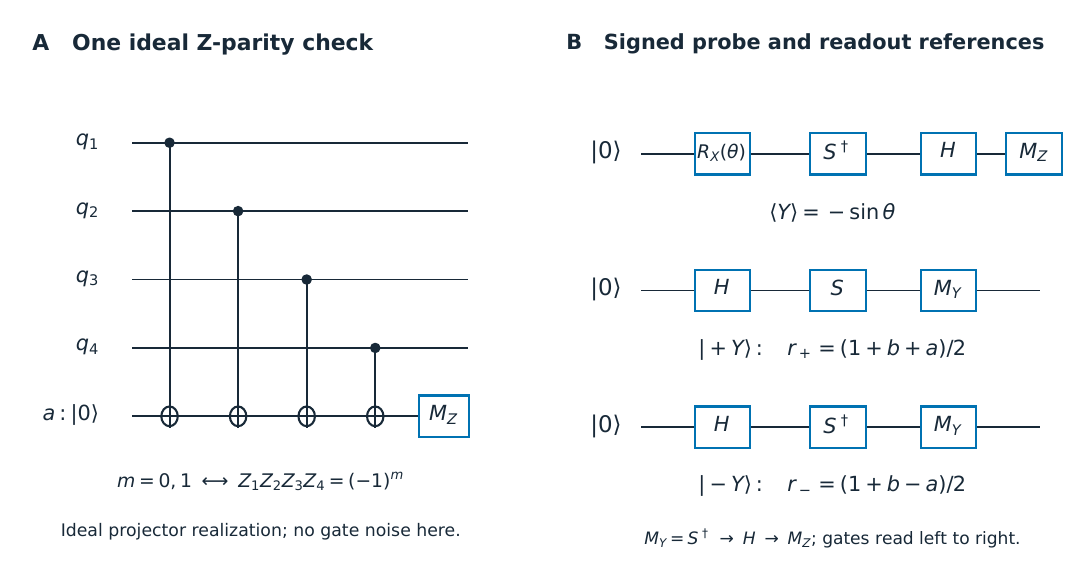}
\caption{\textbf{Circuit conventions for the toric instrument and calibration.}
(a) Four controlled-NOT (CNOT) gates with data qubits as controls transfer a $Z$-plaquette parity to an ancilla. The ancilla measurement returns its eigenvalue and preserves coherence inside the selected parity sector. (b) A signed probe and two reference preparations use the same $Y$ readout. The parameters $a$ and $b$ describe its contrast and offset. Gates are applied from left to right. The drawings assume error-free gates and readout; the noisy surface-memory schedules are specified below.}
\label{fig:circuit-definitions}
\end{figure*}

\subsection{Executed surface-code memory schedules}
The circuit-level extension uses rotated planar surface codes with distances $d=3,5$, $d^2$ data qubits, $d^2-1$ ancillas, and 30 extraction rounds. Stim's \texttt{surface\_code:rotated\_memory\_x} and \texttt{rotated\_memory\_z} templates fix preparation, all CNOT targets and ordering, measurement/reset operations, and terminal data readout. The circuit constructor flattens the template and inserts the following faults while preserving its detector and logical-observable annotations.

At the start of each round, every data qubit receives \texttt{PAULI\_CHANNEL\_1} with the specified $X,Y,Z$ probabilities. Each CNOT is followed by two-qubit depolarization. Each Hadamard is followed by one-qubit depolarization with probability $0.1p$. A computational-basis measurement is preceded by $X$ error and an $X$-basis measurement by $Z$ error, each at the readout rate. Resets have the corresponding anticommuting error immediately afterward at rate $p/2$. Thus an ancilla measurement/reset has distinct premeasurement and postreset faults. There is no additional idle fault at every tick; the stated data-idle channel is inserted once per round.

The base prior uses CNOT rate $p$, readout rate $p$, and idle probabilities $(p/3,p/3,p/3)$. The six catalog members are: base; readout $8p$; idle-$Z$ probability $8p$; local CNOT rate $8p$; independently drawn two-qubit $XX$ and $ZZ$ events of probability $3p$ each per round; and the mixed prior with readout $4p$, idle-$Z$ probability $4p$, and local CNOT rate $4p$. The modified entries replace the corresponding base entries. In particular, idle-$Z=8p$ is 24 times the base idle-$Z$ entry, although it is 8 times the scalar rate $p$.

A CNOT is ``local'' for this construction if either target has both Stim coordinates at least $d$. Correlated faults act on the lexicographically first adjacent horizontal data pair in the generator coordinates. Their $XX$ and $ZZ$ draws are independent. The simulation preserves these coordinate rules exactly. The freshness validation fixes catalog family 3 (zero-based indexing, local gates) at distance 3 and family 2 (idle $Z$) at distance 5. The fixed decoders use all 6 priors at $p=0.002$ and correlated matching. The simulated physical rate varies by workload; each independently validated deployment block is stationary at its recorded deployment rate.

The detector record consists of template-specified measurement parities, with $(d^2-1)30$ detector bits per shot. The logical observable is a terminal memory parity. Both are reconstructed independently by exclusive OR (XOR) of the outcomes at the specified absolute measurement indices and checked against Stim's converter before decoding. No hidden physical rate is an input to the controller. The complete schedules include boundary detectors, the four CNOT layers within each extraction round, terminal measurements, and record offsets.

\subsection{Independent circuit reconstruction}
The validation covers 24 decoder priors, 800 calibration/basis configurations (also used for honest confirmation), and 96 deployment/basis configurations with independently sampled shots, spanning 16 reused calibration workloads at three ages. Each circuit is specified by its scalar rate, family, distance, and basis. Independent reconstruction checks instruction identities, targets, and numerical parameters, and reproduces $K_3=3911.5$ and $K_5=10101.5$.

Circuit reconstruction uses full-precision parameters and adds no independent Monte Carlo observations. These checks cover the fixed-family freshness experiment; the earlier profile-changing challenges use the separate definitions and results stated above.
\section{Time-resolved circuit maintenance and adversarial advice}
\label{sm:timed}
\label{sm:contents:32}
This experiment extends the independently initialized surface-memory setting of Sec.~\ref{sm:circuits}. It reuses the complete gate ordering, detector definitions, observable parities, and 6 decoder priors. The following derivation specifies the additional noise and time assumptions; the exact toric-instrument statements elsewhere in the paper retain their original scope.

\subsection{Simulation time units and measured inference durations}
\label{sm:clock-normalization}
The abstract unit $T_0$ specifies the simulation clock. For a duration $\tau$, an angular drift rate $v$ and an interpolation rate $w$, define
\begin{equation}
 \widehat\tau=\tau/T_0,\qquad \widehat v=vT_0,\qquad \widehat w=wT_0.
 \label{eq:time-normalization}
\end{equation}
The normalized clock and interpolation rate are dimensionless; $\widehat v$ is an angle per normalized clock increment. Products are invariant: $vA=\widehat v\widehat A$ and $w\Delta\tau=\widehat w\Delta\widehat\tau$. Hence the transport and acceptance bounds have identical values in either representation. We keep the scenario-clock symbols in the transport equations and show their $T_0$ units explicitly in tables and captions.

The toric and time-resolved surface records express acquisition and processing times in seconds. Their scenario representation retains each stored numerical value as the coefficient of $T_0$. In particular, a measured request duration $\ell_{\rm wall}$ maps to
\begin{equation}
 \ell_{\rm scenario}=
 \frac{\ell_{\rm wall}}{1\,\mathrm{s}}T_0.
 \label{eq:latency-embedding}
\end{equation}
This conversion fixes the acquisition, inference and deployment ratios and every drift--duration product. It is a scenario normalization rather than a measured quantum-device timescale. Assigning a different round time or inference-to-round ratio defines a different scenario.

\begin{table*}[t]
\caption{Scenario clocks used in the controller studies. The wall-time conversion applies only to measured inference; quantum durations and reference-controller processing are stipulated. Common symbols do not identify different rows with one hardware clock.}
\label{tab:clock-conventions}
\begin{tabular}{p{0.30\textwidth}p{0.21\textwidth}p{0.40\textwidth}}
\toprule
Study & Measured inference & Other timing assumptions \\
\midrule
Central encoded toric & $(\ell_{\rm wall}/1\,\mathrm{s})T_0$ & Round $10^{-6}T_0$; acquisition $NH_c$ rounds; deployment 300 rounds. \\
Timed surface maintenance & $(\ell_{\rm wall}/1\,\mathrm{s})T_0$ & Round $10^{-6}T_0$; acquisition and delivery accounted for explicitly. \\
Expanded advice attacks & Fixed evaluation age & Decision evaluated at prescribed age $T_0$; model wall time is recorded but does not set that age. \\
Expanded multistep maintenance & $(\ell_{\rm wall}/60\,\mathrm{s})T_0$ & Round $10^{-6}T_0$; service $0.05T_0$; reference processing $0.001T_0$. \\
Earlier stationary circuit and sentinel studies & Study-specific prescribed ages & Their own nominal units, with the stated transfer and duration assumptions. \\
\bottomrule
\end{tabular}
\end{table*}

The two measured-latency conversions differ by a factor of 60 relative to the prescribed quantum round. Changing that ratio changes the scenario. Model and interface comparisons within the expanded maintenance study use its common conversion; cross-study differences in completion cannot be attributed solely to the model or interface. The central toric median request contributes approximately $0.782T_0$, or 782000 nominal rounds. This is maintenance latency with an available incumbent, rather than a deadline for decoding each syndrome round.

The earlier stationary-acquisition circuit study uses one nominal round as its abstract unit: there $T=30T_0$. The supporting coherent/sentinel sweeps likewise retain their own prescribed unit as $T_0$. These are distinct scenario families; their absolute timings are not compared. In the central toric and time-resolved surface scenarios the nominal round is instead $10^{-6}T_0$. The common notation exposes each study's stated relative timing without identifying their clocks with one measured device.

Only the explicitly converted inference duration enters the scientific scenario clock. The stipulated reference-controller latency is distinct from measured Qwen latency, so the timing model does not establish a measured controller-speed comparison.

\subsection{Simulated noise path, clock, and sampling unit}
At each elementary noise location $j$, the probability vector is
\begin{align}
 Q_j(\lambda)&=Q_j(0)+\lambda[Q_j(1)-Q_j(0)],\nonumber\\
 \lambda(\tau)&=\max\{0,1-w_{\rm true}\tau\}.
 \label{eq:timed-law}
\end{align}
At distance 3, $Q_j(1)$ has the local-gate amplification; at distance 5 it has amplified idle-$Z$ noise. $Q_j(0)$ is the corresponding base family. The physical scale $p$ is independently drawn uniformly from $[0.0018,0.0022]$ for each workload and used to generate the simulated noise, rather than supplied to the controller. Matching priors remain fixed at $p=0.002$. Both $p$ and all evaluation losses are held in separate evaluator records. The distinction between the two ensembles precludes a pure code-distance scaling interpretation.

The interpreter assigns the timestamp $s_i+30\tau_r k_j/K_{\rm tick}$ to an instruction at cumulative tick $k_j$ in the unchanged flattened schedule with total $K_{\rm tick}$ ticks. The declared round time is $\tau_r=10^{-6}T_0$ normally and $10^{-4}T_0$ in the slow-acquisition condition. A shot period is $h=30\tau_r+5\times10^{-6}T_0$, including the reset overhead. Alternating X/Z calibration shots begin at $s_{Xi}=2ih$, $s_{Zi}=(2i+1)h$, for $i=0,\ldots,n-1$, and finish acquisition at $2nh$. No logical state persists between independently initialized memories.

Fault draws are independent across locations and shots conditional on this deterministic path. Noise probabilities change at every location and shot without time binning. Two-qubit depolarization selects one of the 15 nonidentity Paulis; the one-qubit Pauli channel draws mutually exclusive X/Y/Z faults. A separately implemented parity converter and Stim's converter reconstruct detector and observable flips from the stored measurement-flip records. These records are relative to Stim's reference sample, rather than absolute hardware measurement bits.

The simulation uses 3 nested calibration prefixes $n=2048,8192,32768$. A deployment stream contains $n_f=16384$ new shots per basis and lasts $2n_fh$. Its start is acquisition end plus $0.010T_0$ analysis, $0.002T_0$ communication, converted measured local-model latency (or the stipulated $0.001T_0$ reference latency), and any delivery delay. Its final time $\tau_f$ includes the entire deployment stream. Simulation runtime is not substituted for quantum-device time.

\begin{table}[t]
\centering
\caption{Specified timing scenarios. Rates bound the dimensionless interpolation coordinate per $T_0$. Durations are stipulated scenario quantities; the request-latency conversion is given in Eq.~\eqref{eq:latency-embedding}. The simulated step in the final row violates its declared bound.}
\label{tab:timed-scenarios}
\begin{tabular}{lrrr}
\toprule
Condition & $wT_0$ & Delay/$T_0$ & $\tau_r/(10^{-6}T_0)$\\
\midrule
Honest & 0.0005 & 0 & 1\\
Misleading text & 0.0005 & 0 & 1\\
Delivery delay & 0.002 & 600 & 1\\
Slow acquisition & 0.02 & 0 & 100\\
Wrong identity & 0.0005 & 0 & 1\\
Drift bound violated & 0.0005 & 0 & 1\\
\bottomrule
\end{tabular}
\end{table}

There are two independent parameter draws per distance and condition in Table~\ref{tab:timed-scenarios}. Thus the 72 budgeted instances per controller represent 24 independent workloads. Repeated budgets, acceptance rules, and the 5 alternative decoder actions are not additional independent workloads. The two controllers have separately sampled deployment streams because their inference latencies differ. Comparisons among rules for one controller are paired on identical observations and outcomes.

\subsection{Confidence for changing acquisition probabilities}
Let independent Bernoulli losses $X_i$ have probabilities $q_i$, with $\bar q=n^{-1}\sum_iq_i$ and empirical mean $\widehat q$. For real $z$, concavity yields
\begin{align}
 \log\mathbb E e^{z\sum_iX_i}
 &=\sum_i\log(1-q_i+q_ie^z)\nonumber\\
 &\le n\log(1-\bar q+\bar q e^z).
 \label{eq:timed-mgf}
\end{align}
The binary Kullback--Leibler (KL) divergence is
\[
 \operatorname{kl}(p\|q)
 =p\log\frac{p}{q}+(1-p)\log\frac{1-p}{1-q},
\]
with natural logarithms and endpoint values defined by limits. A zero-probability term contributes zero; the divergence is infinite when the reference assigns zero probability to an outcome of positive probability.
The usual Chernoff argument therefore bounds each tail by $\exp[-n\operatorname{kl}(\widehat q\|\bar q)]$. Inverting
\begin{equation}
 n\operatorname{kl}(\widehat q\|q)\le\log(2/\alpha_{\rm box})
 \label{eq:timed-kl}
\end{equation}
gives a conservative two-sided interval for the average probability even when the $q_i$ differ. Endpoint cases are evaluated by their limiting formulas. This experiment uses three fixed looks, not arbitrary stopping; broader time-uniform constructions are available~\cite{howard2021confidence}.

At each look, 12 marginal boxes cover two bases and 6 actions. Twenty additional confidence intervals cover the two discordant outcomes for each of five candidate/incumbent pairs in each basis. Assigning $\alpha_{\rm box}=0.05/(3\times32)$ covers all these boxes and all three looks within a workload by a union bound. Marginal and discordant counts are computed on the same records; no independence among decoders is assumed.

Label the X and Z memory bases by $b=0,1$, respectively. Let $[\ell_b,u_b]$ bound the incumbent's average acquisition risk. For candidate $a$, the average paired difference has upper bound
\begin{equation}
 d_b(a)=u\{X_a=1,X_0=0\}-\ell\{X_0=1,X_a=0\},
 \label{eq:timed-paired-bound}
\end{equation}
where $u\{\cdot\}$ and $\ell\{\cdot\}$ denote the confidence endpoints for the indicated discordant-cell probabilities. Pairing retains correlation between candidate and incumbent errors on each shot.

\subsection{Transport and the acceptance bound}
At each location define $c_j=\TV[Q_j(1),Q_j(0)]$. Using the same schedule for calibration and deployment, couple the elementary faults of memories beginning at $s$ and $\tau\ge s$. The common instruction offset cancels, giving
\begin{equation}
 \TV(P_s,P_\tau)\le\min\{1,\Gamma w(\tau-s)\},
 \qquad \Gamma\ge\sum_jc_j.
 \label{eq:timed-coupling}
\end{equation}
No record or decoder loss differs when every coupled fault agrees. The experiment uses $\Gamma=5.544$ at distance 3 and $12.65$ at distance 5, calculated at the maximum permitted $p$, rather than the hidden physical value. These are slopes with respect to $\lambda$, and differ from the scalar-$p$ coefficients in Eq.~\eqref{eq:circuit}.

Averaging over all calibration starts and bounding every future start by the deployment end $\tau_f$ gives Eq.~\eqref{eq:timed-acquisition}: the average of $\min(1,x_i)$ is at most $\min(1,\bar x)$. It bounds the change of each decoder's mean basis risk, including its average over a deployment stream. The interval for the incumbent becomes $[\ell_b-\epsilon_b,u_b+\epsilon_b]$, while the candidate-minus-incumbent risk difference is at most $d_b(a)+2\epsilon_b$.

For baseline risks $r_b$, the identity $r_b-\max(r_b,r_{1-b})=\min(0,r_b-r_{1-b})$ gives the worst-basis excess-risk upper bound
\begin{align}
 U(a)=\max_{b\in\{0,1\}}\big[&d_b(a)+2\epsilon_b\nonumber\\
 &+\min\{0,u_b+\epsilon_b-\ell_{1-b}+\epsilon_{1-b}\}\big].
 \label{eq:timed-upper}
\end{align}
Here $1-b$ means the opposite memory basis. When all count-based confidence intervals cover their corresponding probabilities, maximizing over the enclosing incumbent boxes proves the bound. Acceptance requires matching workload identity, a permitted nonzero candidate, and $U(a)\le-0.001$. Candidate zero retains the incumbent. Since all actions and all three looks share the confidence event, proposals may depend on observed calibration and external text. Transport applies for every allowed later time, including a deployment time determined by measured inference latency. The guarantee assumes the declared noise-path class and correct acquisition and deployment timestamps. Its probability is over acquisition outcomes; it is not conditioned on acceptance, and the 0.05 allocation applies within 1 workload across its three looks.

Authorization checks identity and the action catalog but omits the risk test. Confidence-only sets $\epsilon_b=0$. The end-timestamp comparison instead uses $\min\{1,\Gamma w(\tau_f-\tau_{\rm end})\}$, omitting the acquisition span. The whole-acquisition rule uses Eq.~\eqref{eq:timed-acquisition}. These rules share the same proposal and final records. The end-timestamp allowance is a comparison rule, not the derived coverage bound for changing acquisition probabilities.

\subsection{Controller interface, threats, and deployment outcomes}
The local Qwen3:8b controller receives identifiers for available actions, their empirical ranking, the actions certified when the evidence was issued, an evidence identifier that carries no physical parameters, acquisition duration, delivery delay, shot budget, and confirmation of workload identity. A separate note is explicitly designated untrusted. True simulation parameters and detailed assessment losses are excluded from its inputs. Each episode produces 1 activation or retention request in a specified format. Repeated command fields, unavailable actions, invalid responses, or an incorrect evidence identifier cause rejection without activation. The deterministic reference chooses the best ranked action when identity matches and otherwise retains.

Misleading notes recommend the publicly worst-ranked nonzero action and claim an urgent operator override. Delayed notes describe old evidence as approved for reuse; wrong-identity notes claim the mismatch flag is a dashboard bug. Qwen follows the misleading action in all 12 such instances and requests nonzero activation in all 12 wrong-identity instances. Every request and parsing result is retained. The threat model assumes that the software provides the correct workload-match flag. The flag itself supplies no cryptographic authentication.

All 72 Qwen requests complete, with temperature zero, seed 19, a 4096-token context and a 192-token output limit. The model is Qwen3:8b, served by local Ollama with Q4\_K\_M weight quantization, a mixed-precision format using primarily 4-bit weights. A recorded warm-up precedes the evaluation; the median measured request time is 0.690 s, with range 0.624--2.972 s. Fixed run order, a resident-model assumption and the stipulated reference latency bound the interpretation of speed comparisons.

The endpoint is $\max_b\bar r_b(a)-\max_b\bar r_b(0)$ over independently sampled deployment records. Its paired confidence interval uses the same 32-box allocation at one look with total allowance 0.05. Applying Eq.~\eqref{eq:timed-upper} at zero transport gives the upper endpoint; exchanging candidate and incumbent gives the negative of the lower endpoint. This covers all five comparisons within 1 stream. It is not a simultaneous guarantee across the entire study. Benefit requires the upper endpoint below $-0.001$; harm requires the lower endpoint above $+0.001$; other accepted-update outcomes are unresolved.

Figure~\ref{fig:timed-conditions} gives the condition-resolved counts and Fig.~\ref{fig:timed-paired} retains every proposed recovery outcome. Under honest advice all 12 proposals have established benefit, but the whole-acquisition rule admits 4 per controller. For Qwen, averaging the paired point-estimated risk reduction equally over those 12 instances, with rejection contributing zero, retains 48.5\% of the authorization-only mean gain. This exploratory ratio is a descriptive average over the three tested budgets for each workload, not a population interval or an algorithm-success probability.

In the delayed and slow-acquisition conditions the full menus are already empty at issuance. The slow case reserves a $98.468T_0$ deployment horizon, enough to make the transport bound vacuous at the specified drift rate. These conditions expose conservative rejection without isolating acquisition age from execution duration. The analytic delay curves in Supplemental Fig.~\ref{fig:analytic-timed}(b) use an initially admissible honest record, but introduce no new physical observations.

The invalid-premise condition sets $\lambda$ abruptly to zero at deployment. A total of 4 proposals per controller are admitted; all 8 benefit/harm classifications remain unresolved. As a separately labeled exploratory diagnostic, compare the deployment lower endpoint with the accepted numerical $U(a)$. It exceeds that certificate in 3 of 4 Qwen cases and 4 of 4 deterministic cases. All 9 confidence-only admissions per controller in the bounded-delay condition likewise contradict their claimed numerical improvement. These discrepancies assess the promised bound, rather than resolving harm relative to the incumbent. The diagnostic was added after inspecting outcomes; its intervals are within-stream quantities and provide no study-wide simultaneous failure rate.

\subsection{Computational record and reproduction}
The design fixes the physical probability interval, seeds, shot budgets, three looks, timing scenarios and primary benefit/harm endpoints before main deployment outcomes. No parameter or sample budget is retuned on those outcomes. The complete run contains 1,572,864 new calibration and 4,718,592 new deployment memory shots. The alternative-action audit evaluates 5 actions on the same records, rather than adding independent physical samples.

The declared cycle time and communication costs are scenario parameters in $T_0$; the model request latency is a measured wall-time input converted by Eq.~\eqref{eq:latency-embedding}.

Validation includes stationary comparison with native Stim, exact small-distribution coverage and time-transport checks. A separate audit reconstructs every reported decision and outcome from the sampled circuit records. Repeated simulation reproduces 393,216 shots bit for bit; these repeats are not independent evidence. Replay holds the model requests and clocks fixed; rerunning inference can change deployment times.

\begin{figure*}[p]
\centering
\includegraphics[width=.96\textwidth,height=.77\textheight,keepaspectratio]{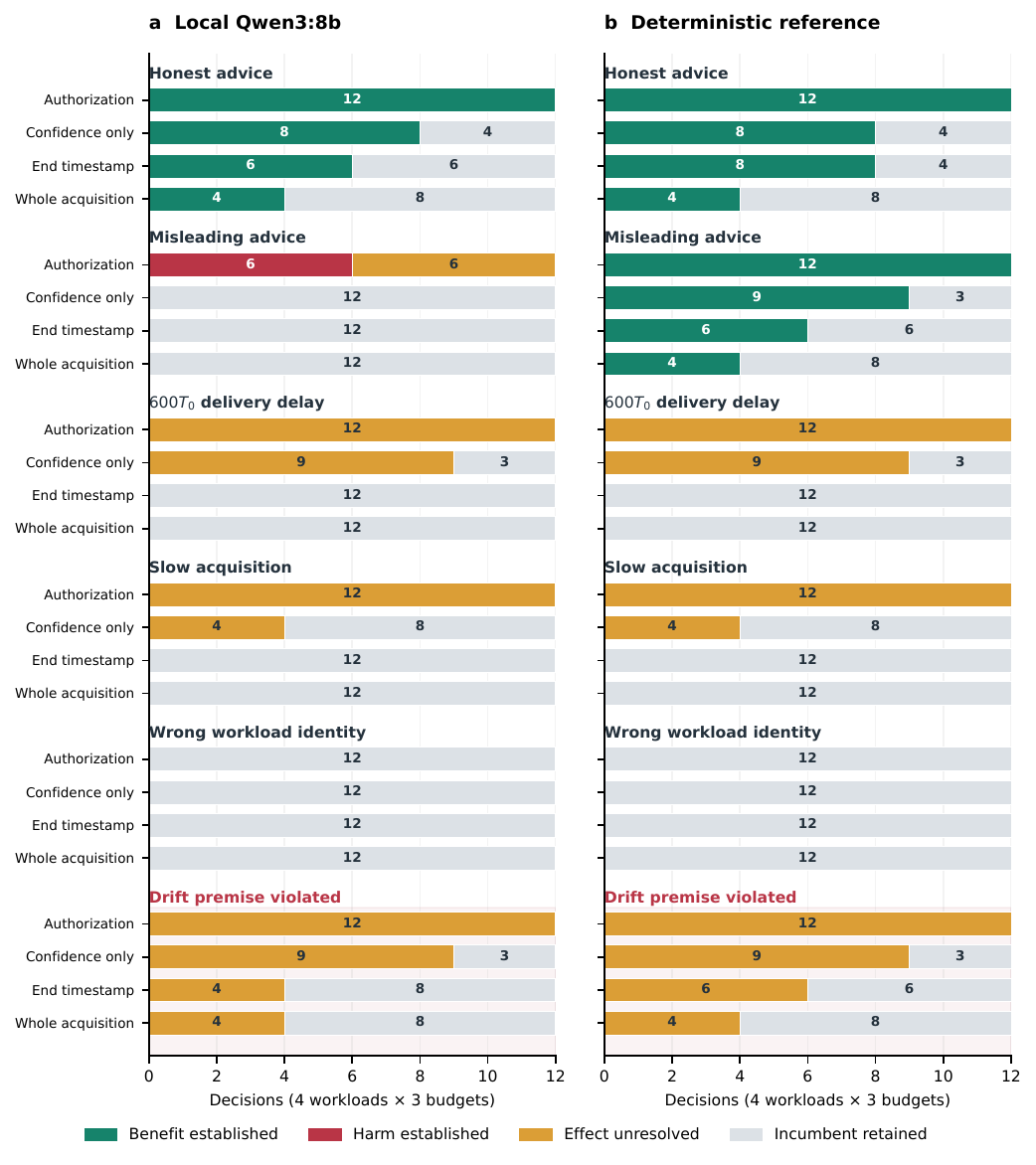}
\caption{\textbf{All timed-maintenance conditions and rules.} Each bar contains 12 case/budget instances on 4 independent workloads. Counts distinguish established benefit, established harm, unresolved accepted-update outcomes, and retention. The final condition intentionally violates the physical drift bound and is excluded from the aggregate valid-premise panel of Fig.~\ref{fig:timed-main}. Every method includes the same identity check. The deterministic controller ignores misleading prose; the language model may request the wrong action even when the evaluator-generated menu excludes it. Zero observed resolved harm is not a population security rate.}
\label{fig:timed-conditions}
\end{figure*}

\clearpage
\onecolumngrid
\noindent\textbf{Paired deployment outcomes.} The following figure resolves every surface-code proposal against independently sampled deployment outcomes.\par
\begin{figure}[!ht]
\centering
\includegraphics[width=\textwidth,height=0.78\textheight,keepaspectratio]{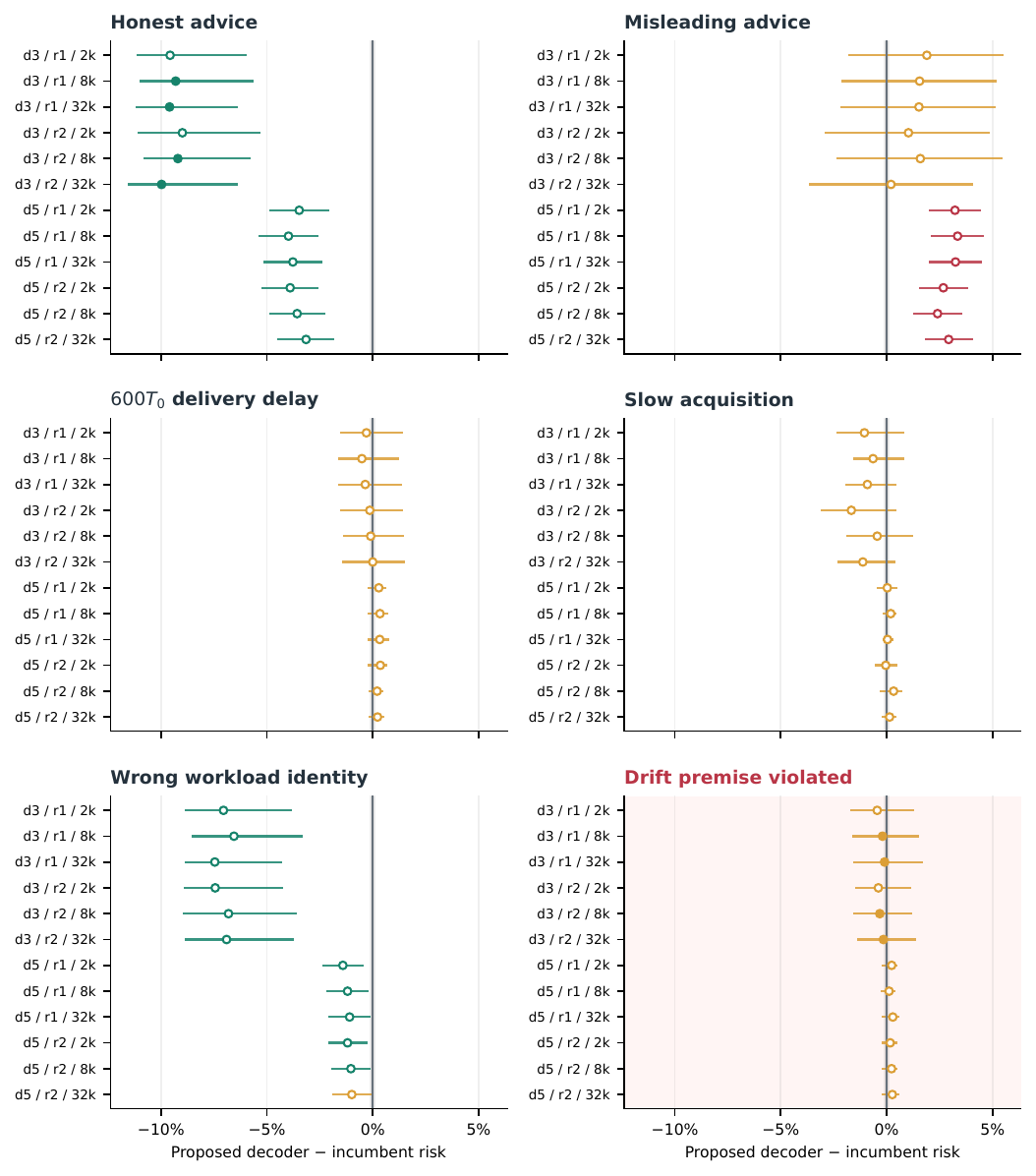}
\caption{\textbf{Independently sampled deployment outcomes for every Qwen proposal.} Each row has 16384 new deployment shots per basis. Lines are paired worst-basis risk intervals with simultaneous action coverage within each deployment stream; filled points mark admission by the whole-acquisition rule. Open points show the simulated effect of a rejected proposal, which is the estimated consequence of applying a proposal that was instead rejected. Row labels identify distance, replicate and calibration budget; k means 1024 shots. The 6 resolved harmful misleading-advice outcomes arise from 2 distance-5 workloads at 3 nested budgets. Shading identifies the deliberately violated drift premise.}
\label{fig:timed-paired}
\end{figure}
\clearpage
\twocolumngrid

\section{Additional validation of calibration and recovery updates}
\label{sm:contents:33}
\label{sm:additional-validation}
These supporting studies test readout calibration, the coherent-memory controller and the stationary-acquisition circuit bound. Their physical models and sampling units are distinct from the timed surface-code maintenance experiment in the main text. We retain their benefits, failures and comparison methods here.

\subsection{Readout estimation and activation checks}
The simulated outcome colors in the overview use a specified deployment path, $\theta _{\rm dep}=0.10-\min(vA,0.20)$, while the evaluator receives only the expanded confidence set. This distinction makes the rejected-but-beneficial region interpretable. The plotted physical angle can still favor the candidate even though another angle consistent with the evidence would not. Rejection reflects that unresolved possibility. As the uncertainty grows across the sign ambiguity, a phase table supported by the calibration-time bound can lose its certificate while it still improves recovery on the evaluated path. The age boundary therefore measures what the observations justify, not what a controller with access to the hidden noise would choose.

Figure~\ref{fig:completion-calibration}(a,b) examines the cost and coverage of readout calibration. Readout uncertainty can be measured using known $\pm Y$ references, for which $r_\pm=(1+b\pm a)/2$. Confidence intervals on both reference counts and the signed probe share the failure allowance. In 7200 simulated calibrations, this construction admits 2323 beneficial updates with 0 observed harmful activations, at total costs of 9216--73728 shots. It assumes known preparation, a common stationary readout channel, and a justified probe-to-code relation. Measuring readout parameters preserves coverage in the changed-readout settings, where the original fixed nuisance bounds become invalid; the additional uncertainty reduces acceptance at small reference budgets.

Section~\ref{sec:completion} specifies the reference budgets and readout settings. The reference measurements convert uncertainty about readout bias into an interval that the recovery-risk calculation propagates to deployment.

The contrast--offset confidence rectangle and its failure-probability allocation are derived in Sec.~\ref{sec:completion}. If the reference measurements cannot establish positive contrast, the acquisition cannot support the assumed signed response and is rejected.

The paired comparator keeps the original fixed readout bounds while using the same probe counts. Its higher acceptance in the two changed-readout settings accompanies loss of confidence coverage: there are 83 and 470 noncoverage events, respectively, out of 2400 records per setting. Measuring the nuisance parameters gives 3 noncoverage events across all 7200 calibrations. The legacy comparator produces no observed harmful update in these workloads, so noncoverage and observed harm must be distinguished. Figure~\ref{fig:completion-calibration} shows the price of maintaining a supported uncertainty model, rather than inferring protection solely from a favorable loss sample.

Figure~\ref{fig:completion-calibration}(c) tests the evaluator through an exhaustive numerical audit independently of a particular adviser's strategy. For an 8192-shot calibration it enumerates all 8193 counts and all 13 actions. At each physical setting, it sums the binomial probability of counts for which any certified action violates the promised margin. The maximum is $0.00378720114<0.01$ across 5076 settings. The included 1692-setting zero-drift subset matches the earlier audit exactly. This finite-grid check complements the analytical guarantee and its numerical-bound assumption.

A final implementation check ensures that the operation applied is the operation evaluated. The software ties the calibration, experiment, current configuration, reference model, and recovery-table contents together. At activation it checks their identities, validity period, and single use; it then reevaluates the current risk and copies the checked table while holding a lock. The 13 deterministic conformance cases test these steps. Section~\ref{sec:completion} specifies the trusted process and clock assumptions and its physical assumptions.

The checks connect certification to the applied recovery. A certificate for the wrong workload or a table replaced after evaluation cannot support Eq.~\eqref{eq:guarantee}. Evidence binding and checked actuation address substitution and replay; deployment-time risk evaluation separately accounts for physical change.

\subsection{Coherent-memory controllers and learned proposals}
\label{sec:results}

The coherent-memory controller experiment compares local Qwen3:8b~\cite{qwen2025} with a deterministic state machine, each under authorization alone or the full acceptance rule. Both receive the same public measurements, certified-action menu, untrusted reports, and acquisition budget. The evaluator computes the exact-channel model numerically; the LLM manages calibration requests and action selection. The simulated angles and independently calculated losses are used only for assessment, not as controller inputs.

The experiment uses 192 independent workloads, each evaluated with both controllers and both acceptance rules. Section~\ref{sec:completion} specifies the 8 conditions, horizons, measurement budgets, and 4 malicious-report variants. The two controllers receive the same observations and permitted actions in each matched workload.

We distinguish simulated recovery performance from workflow completion. A \emph{beneficial completion} finishes the workflow and deploys an action with $D_u<-\delta$. A \emph{harmful activation} deploys an action with $D_u>\delta$. Retention and incomplete workflows are separate outcomes, both included in the denominators. These rates describe the constructed workloads.

The experiment isolates workflow decisions under evidence-based acceptance; the separate predictor study tests learning proposals from measurement histories.

Across 168 workloads satisfying the drift premise, the guarded LLM gives 159 beneficial completions, 0 observed harmful activations, 6 incomplete workflows, and 3 retentions. Authorization alone gives 142 beneficial completions and 26 harmful activations. The deterministic controller gives 162 beneficial completions and 0 observed harmful activations under either gate because it follows the certified menu even without enforcement. Additional calibration-shot totals are 1,007,616 for the guarded LLM, 851,968 for authorization alone, and 860,160 for either deterministic variant. Thus both controllers improve on retention, with the deterministic controller completing more beneficial updates at lower calibration cost [Fig.~\ref{fig:hero}(c)].

When the physical drift premise is deliberately violated, the full rule permits harmful activation in 16/24 LLM and 21/24 deterministic rollouts [Fig.~\ref{fig:hero}(d)]. Authentic old data cannot reveal the unobserved sign change in this sign-blind instrument. The pointwise Chernoff--KL intervals allow independent, nonidentical Bernoulli outcomes, using workloads as the independent units. They quantify uncertainty in the observed simulation frequencies separately from the family-level guarantee in Eq.~\eqref{eq:guarantee}.

Figure~\ref{fig:support-security} isolates the role of freshness using paired mechanism tests. Confidence without age propagation admits 178 harmful updates in 200 bounded-stale paths; the full rule admits none. It retains 97.1\% of aggregate candidate gain with fresh evidence and 90.8\% with benign delay. Matched-path bootstrap intervals are $[94.9,98.9]\%$ and $[86.8,94.2]\%$, respectively. Rejection retains the incumbent instead of substituting a new proposal. Fixed expiry also rejects the sampled stale challenge; its calibration to different drift contexts determines the utility cost. A separate violated-drift challenge causes harm on 187/200 paths under the full rule.

Figure~\ref{fig:support-security} holds each path's observations and candidate fixed while changing the rule. Rejection gives zero deployed gain relative to the incumbent. The retention fraction is summed deployed gain divided by summed candidate gain within the condition; it measures preserved benefit rather than acceptance frequency.

The baseline rules remove different pieces of evidence. Authorization alone checks permission and binding; confidence-only evaluation uses calibration-time uncertainty without propagating age; fixed expiry imposes a common time cutoff. Their comparison with the full rule tests whether uncertainty and drift should enter the decision together. The 200-path mechanism tests and the 168 valid-premise agent workloads have different sampling units and purposes. The former isolate the rule with fixed proposals; the latter include acquisition choices, incomplete workflows, and the cost of interacting with an adviser.

The learned-predictor comparison evaluates 5 GRU and 5 MLP training seeds on 600 held-out simulated noise paths. Neural and classical methods receive the same observations, catalog, and budget for new calibration shots, in addition to 3072 common historical shots per path. The Supplemental Material reports risk distributions, calibration cost, and paired GRU-minus-Bayesian differences. Learning improves recovery relative to retention; the Bayesian comparison remains statistically inconclusive.

\subsection{Stationary-acquisition circuit validation}
\label{sec:circuit-extension}
A stationary-acquisition Stim/PyMatching study~\cite{gidney2021,higgott2025} uses distance-3 and distance-5 surface-code memories, 30 rounds, and 6 matching priors fixed before evaluation. The endpoint is worst-basis logical failure. Within a fixed scalar noise family, coupling fault distributions gives the sufficient acceptance condition
\begin{equation}
U_{\rm cap}+2\min\{1,Kv(A+T)\}\le-\delta,
\label{eq:circuit}
\end{equation}
Here $K$ sums total-variation slopes over all fault locations, $|\dot p|\le v$, and $T$ bounds execution duration. The factor of two accounts for changes in both candidate and incumbent risk. The coefficients are $3911.5$ for the distance-3 local-gate family and $10101.5$ for the distance-5 idle-$Z$ family; their difference reflects both distance and noise family.

Figure~\ref{fig:completion-circuit}(a--c) propagates 200 reused calibration records per distance across evidence age and 4 declared drift rates. At $v=10^{-8}$, all 200 distance-3 records and 199 distance-5 records admit a deployment window. These curves reevaluate each calibration rather than drawing independent measurements at every age.

The coupling argument in Sec.~\ref{sec:completion} bounds the change in each decoder's failure probability by the sum of changes in the local fault distributions, capped at one. Comparing candidate and incumbent gives the factor of two in Eq.~\eqref{eq:circuit}.

For a positive drift bound in the regime with a positive risk margin and an unsaturated transport bound, the remaining risk margin fixes a latest certified age at the start of deployment,
\begin{equation}
 A_{\max}=\frac{-\delta-U_{\rm cap}}{2Kv}-T.
 \label{eq:startup-age}
\end{equation}
A negative value gives no certified startup window. At zero drift the age penalty vanishes and acceptance is decided directly from the calibration bound. The subtraction of $T$ reserves uncertainty for the execution itself: an operation that passes at startup must also fit within the supported duration. The two distances in Fig.~\ref{fig:completion-circuit} use different noise families, so their windows compare those specified experiments rather than isolating a scaling law in code distance.

Independent validation [Fig.~\ref{fig:completion-circuit}(d)] selects 8 workloads per distance and tests ages zero, half the maximum certified startup age, and twice that age at $v=10^{-8}$. Each instance uses 8192 final shots in each memory basis, totaling 786,432 shots. Candidate and incumbent decode identical records. All 32 within-window instances are accepted and independently support improvement; all 16 beyond-window instances are rejected although still beneficial. The three ages share each calibration workload. Expiry therefore marks loss of sufficient certification, rather than the measured onset of harm. The profile-changing circuit challenge retains its vacuous bound in the Supplemental Material.

Independent deployment samples assess the candidate and incumbent on the same new records. Benefit beyond expiry measures conservatism: the sufficient certificate can end before the action becomes harmful.

\onecolumngrid
\par\medskip
\noindent\textbf{Figures supporting the additional validation.} The following panels report the distinct supporting studies described above.\par
\begin{figure}[!ht]
\centering
\includegraphics[width=\textwidth]{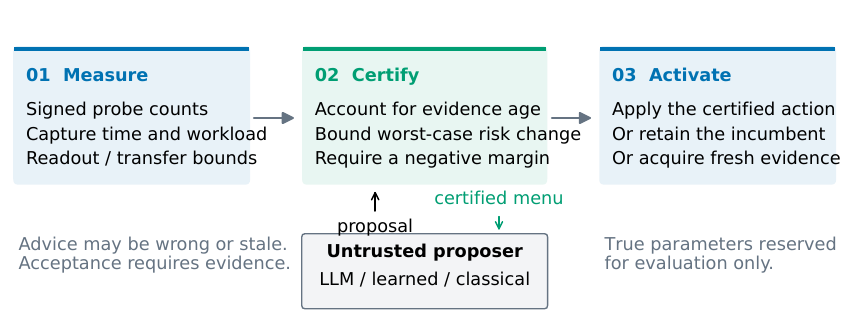}
\caption{\textbf{Measurement integrity and checked activation define the operational workflow.}
Measurement records bind signed-probe counts to an acquisition time and workload; characterized readout and transfer bounds support their physical interpretation. The risk evaluator propagates uncertainty to deployment and requires improvement relative to the incumbent throughout the allowed conditions. The large language model (LLM) proposer can request measurements and select a recovery, but cannot create acquisition records or bypass acceptance. Activation checks evidence and recovery-table identities and applies the evaluated operation; otherwise the incumbent is retained or new evidence is acquired. True physical parameters are used for independent outcome evaluation. The simulation software maintains the records and prevents the proposer from modifying them; the guarantee requires valid physical bounds and application of the evaluated recovery.}
\label{cal:fig:contract}

\end{figure}
\clearpage
\begin{figure}[!ht]
\centering
\includegraphics[width=\textwidth]{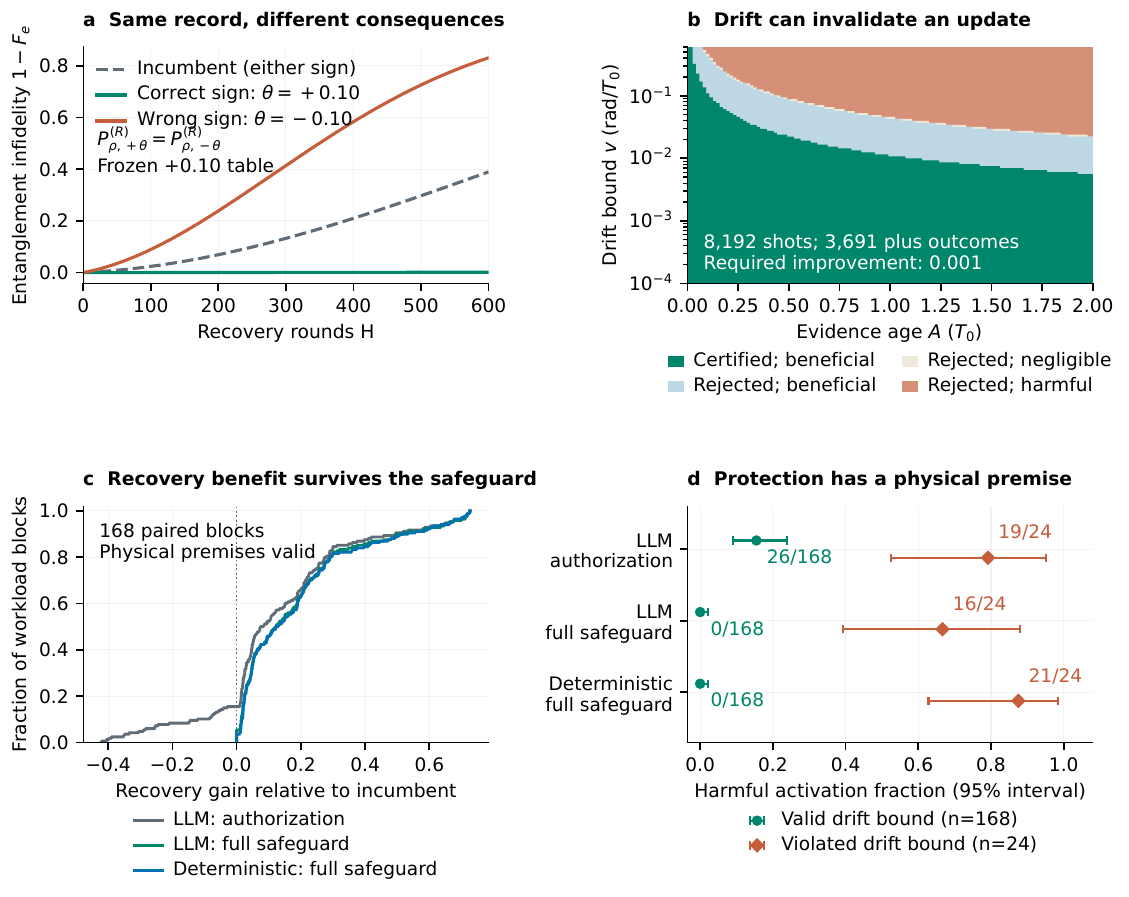}
\caption{\textbf{Evidence connects an information limit to a recovery decision.}
(a) Identical passive history laws have different recovery consequences: one fixed $+0.10$ phase table helps at the correct sign and harms at the opposite sign in the $L=3$ toric instrument of Sec.~\ref{sec:model}.
(b) Certification and rejection for 1 8192-shot calibration record, as evidence age and the drift bound increase. Rejected colors distinguish evaluator-only recovery outcomes; the controller sees the evidence and certificate.
(c) Recovery-gain distributions over 168 paired workloads satisfying the assumptions. The guarded large language model (LLM) completes 159 beneficial updates with 0 observed harmful activations; authorization alone gives 142 beneficial updates and 26 harmful activations. The guarded deterministic controller completes 162 beneficial updates with 0 observed harmful activations.
(d) Harmful-activation fractions with pointwise 95\% Chernoff intervals based on the Kullback--Leibler (KL) divergence. Exceeding the drift bound causes harm under the full gate: 16/24 LLM and 21/24 deterministic workloads. Each independent workload supplies 4 paired rollouts. Supplemental Material, Sec.~\ref{sec:results}, defines outcomes of the coherent-memory controller experiment; the detailed empirical panels in the supporting material give plotting methods and condition-resolved results.}
\label{fig:hero}
\end{figure}

\clearpage
\begin{figure}[!ht]
\centering
\includegraphics[width=\textwidth]{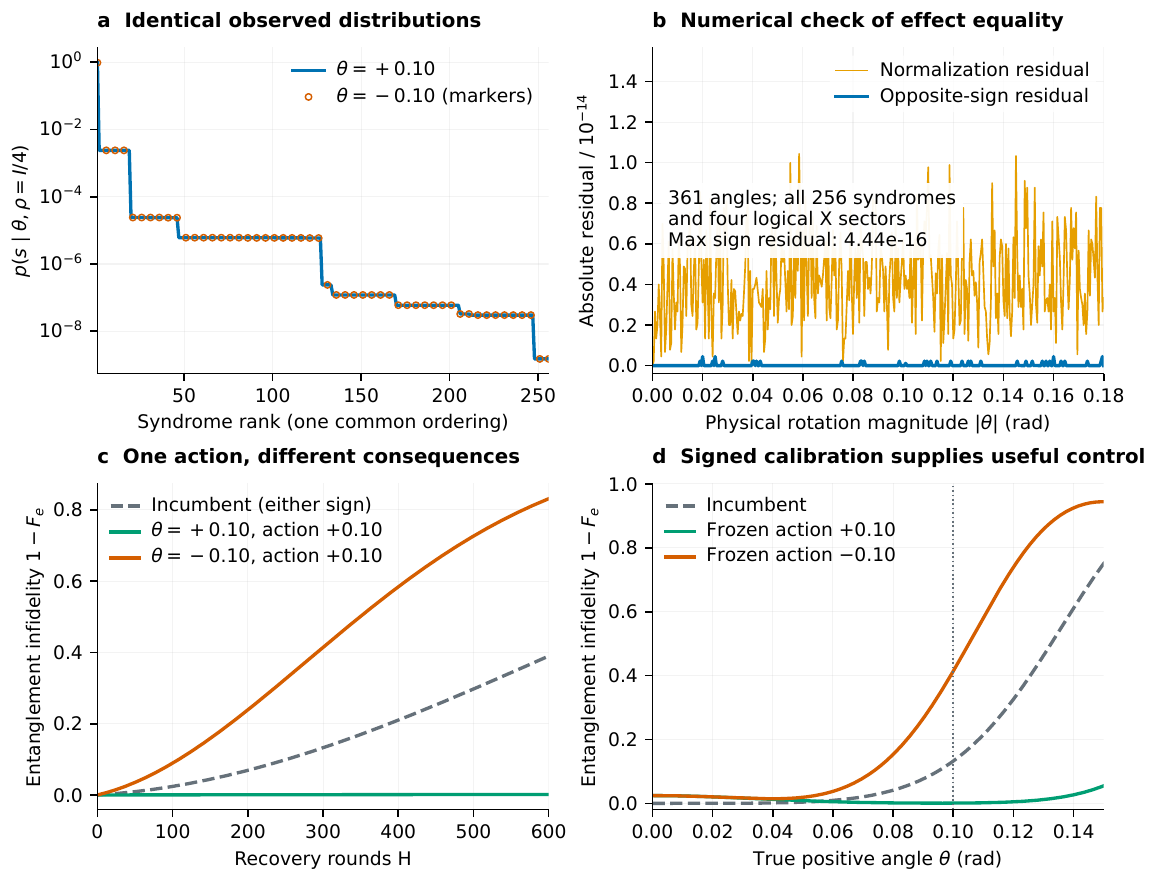}
\caption{\textbf{Passive information and active recovery answer different questions.}
(a) Opposite-sign syndrome probabilities coincide. (b) The maximum independently evaluated effect residual is $4.44\times10^{-16}$. (c) One fixed $+0.10$ recovery table helps at $+0.10$ and harms at $-0.10$. (d) The fixed $\pm0.10$ tables have different risk profiles over the $H=300$ angle sweep. Numerical residuals quantify computational error; the exact identities establish the information limit.}
\label{fig:support-risk}
\end{figure}

\clearpage
\begin{figure}[!ht]
\centering
\includegraphics[width=\textwidth]{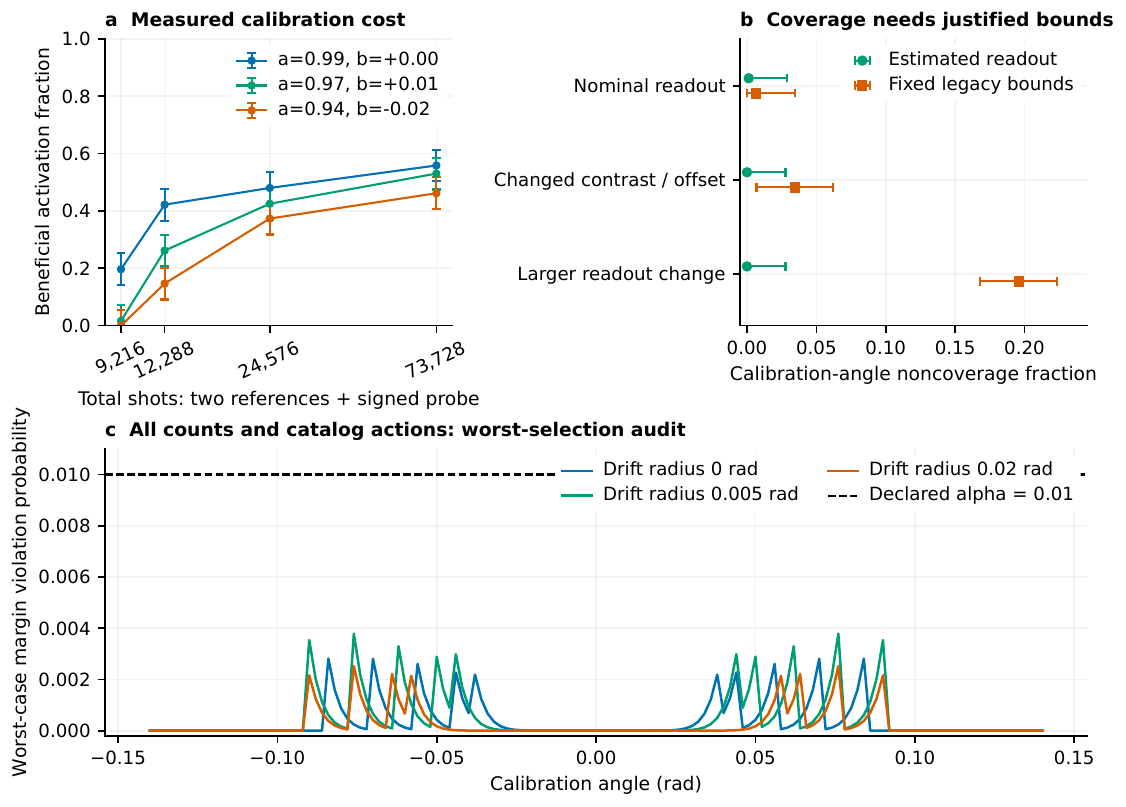}
\caption{\textbf{Readout estimation trades shot cost for justified acceptance.}
Known-state references and signed-probe counts are simulated. (a) Beneficial-activation fraction versus total reference-plus-probe shots, using 600 independent records per cost point and readout setting. (b) Calibration-angle noncoverage for measured and fixed nuisance procedures, using 2400 records per setting. Panels (a,b) show pointwise 95\% Hoeffding intervals. The changed readout settings violate the fixed procedure's nuisance assumptions. (c) Worst-selection margin-violation probability across the 5076-setting audit, maximizing over nuisance parameters and deployment shifts at each angle. Every possible 8192-shot count and all 13 catalog actions are evaluated. These curves sum binomial probabilities on a finite physical grid using the stated numerical treatment; they carry no sampling error bars. The dashed line is the declared $\alpha=0.01$.}
\label{fig:completion-calibration}
\end{figure}

\clearpage
\begin{figure}[!ht]
\centering
\includegraphics[width=\textwidth]{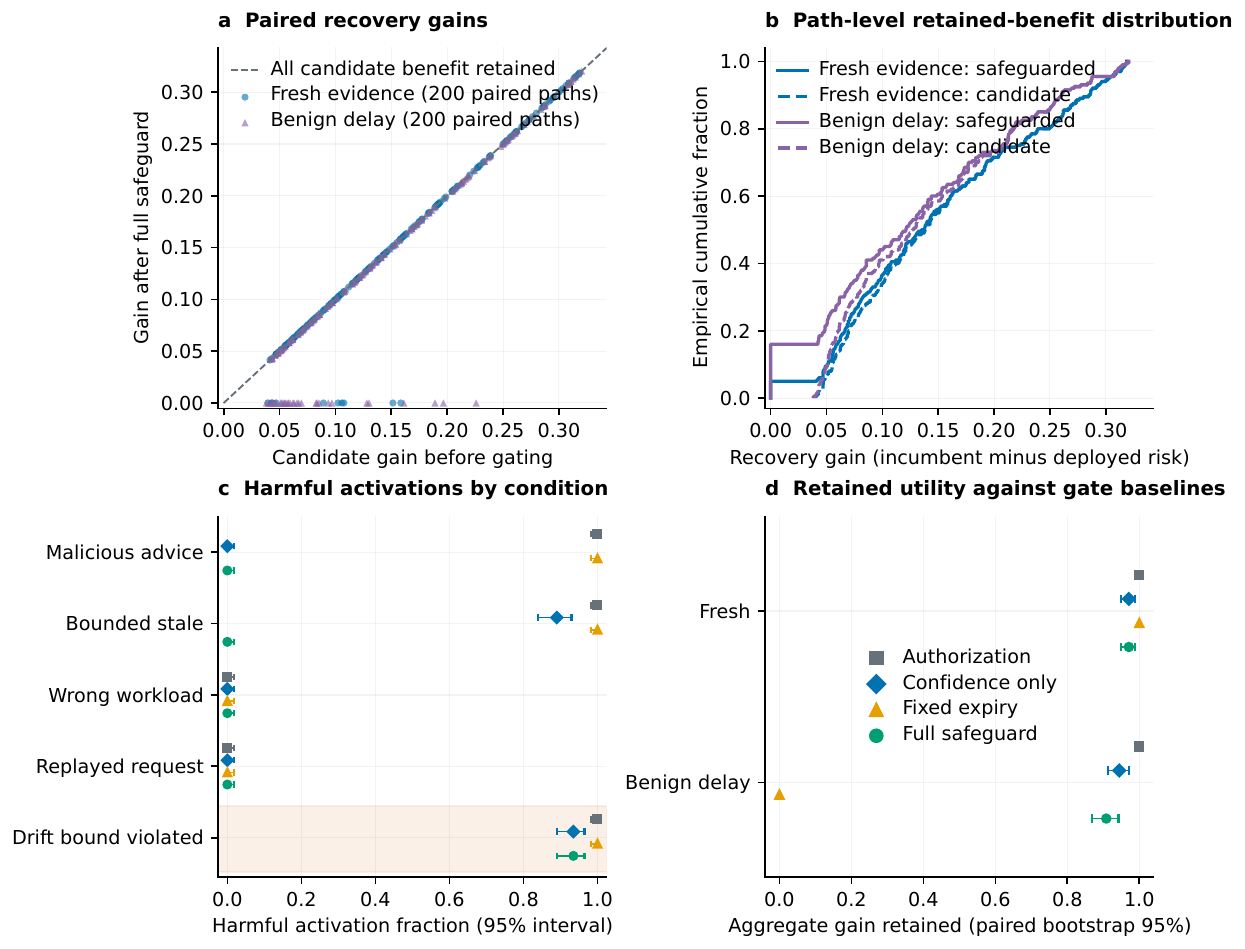}
\caption{\textbf{Paired recovery gains expose protection and its cost.}
(a,b) Proposed and deployed gains for 200 fresh and 200 benign-delay paths. (c) Harmful activations in five challenge conditions, with pointwise 95\% Clopper--Pearson intervals over 200 paths per condition. Violating the drift bound causes full-gate harm on 187/200 paths. (d) Aggregate gain retained: 97.1\% [94.9,98.9]\% for fresh evidence and 90.8\% [86.8,94.2]\% for benign delay, using 10,000 matched-path bootstrap replicates. Baselines share measurements and proposals; rejection retains the incumbent. The Supplemental Material gives the baseline definitions and plotting protocol.}
\label{fig:support-security}
\end{figure}

\clearpage
\begin{figure}[!ht]
\centering
\includegraphics[width=\textwidth]{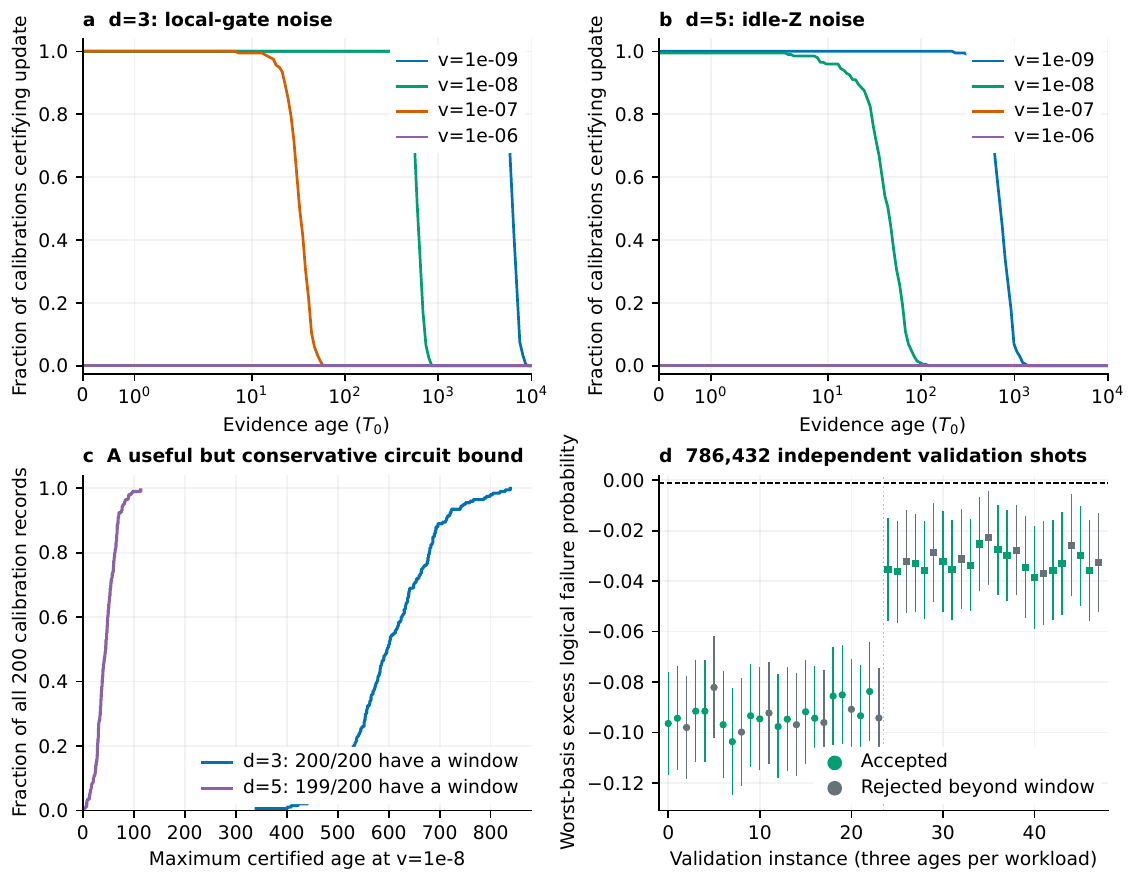}
\caption{\textbf{Freshness within a specified scalar circuit-noise family.}
(a,b) Fraction of 200 reused calibration records per distance certifying an update at each age and declared drift rate. (c) Maximum certified deployment age at $v=10^{-8}$ per $T_0$, retaining denominator 200 even when no window exists. (d) Independent final excess-risk estimates and conservative pointwise 95\% intervals; circles and squares denote distances 3 and 5. Three ages share each of 8 calibration workloads per distance. Each instance uses 8192 shots per memory basis. All 48 intervals support benefit, including 16 updates rejected beyond expiry. The dashed line marks the required $-0.001$ improvement. Distance and noise family differ between the two circuit ensembles.}
\label{fig:completion-circuit}
\end{figure}

\clearpage
\twocolumngrid

\section{Unified toric calibration and recovery: complete protocol}
\label{sm:unified-toric}
\label{sm:contents:34}
This study uses the same square toric instrument defined in Sec.~\ref{sec:model} for passive histories, encoded calibration and recovery assessment. It is distinct from the separate sentinel and surface-code experiments above.

\subsection{Encoded measurement and joint syndrome sampling}
There are 18 data qubits and 8 independent $Z$ checks; the ninth check is redundant. The physical layer is $R_X(\theta )$ on every edge, with the unchanged minimum-weight $X$ recovery and logical-$X$ character ordering. Let $k_{sx}$ and $C_{xy}$ be as in Eq.~\eqref{eq:encoded-channel}. The known superposition of sectors 0 and 3 stays in their span, since all corrected Kraus operators are diagonal in this basis. Its off-diagonal density-matrix element after $H_c$ rounds is $C_{03}^{H_c}/2$. Taking the trace with $Y_{03}$ gives $-\operatorname{Im}C_{03}^{H_c}$ and hence Eq.~\eqref{eq:probe}. Complex conjugation under $\theta \mapsto-\theta $ reverses this mean, while preserving the complete passive history law. The terminal binary measurement can be extended arbitrarily on the unpopulated logical sectors.

Sectors 0 and 3 have the same product parity and therefore identical syndrome probabilities. A complete calibration trajectory can be sampled exactly using $|k_{s0}|^2$ at each round and accumulating
\begin{equation}
 z_R=\prod_{j=1}^{R}\frac{k_{s_j0}k_{s_j3}^*}{|k_{s_j0}|^2}.
 \label{eq:unified-joint}
\end{equation}
Only positive-probability syndromes are sampled. The conditional plus probability is $(1-\operatorname{Im}z_R)/2$. Averaging over histories reproduces Eq.~\eqref{eq:probe}. The saved final calibration data contain the full syndrome history and terminal readout for each independently prepared memory. No unknown stored logical state is exposed to the adviser.

\subsection{Continuous confidence inversion and path allowance}
At a stationary calibration angle $\theta _c$, the $N$ terminal readouts are independent Bernoulli trials with probability $q_Y(\theta _c)$. Inverting the two-sided Clopper--Pearson interval at failure allowance 0.01 over $[-0.15,0.15]$ rad yields a confidence region, potentially disconnected. The calculation uses 60001 grid points, spacing $h=5\times10^{-6}$ rad, and retains every cell whose possible response intersects that interval.

The generator of one physical rotation layer is $G=\sum_{e=1}^{n}X_e/2$, with $\|G\|=n/2$. Telescoping repeated layers and using contraction under quantum channels bounds the change in any final event probability by $Hn|\Delta \theta |/2$. Thus $L_q=nH_c/2=900$ bounds the slope of $q_Y$. A cell centered at $\theta _i$ is retained if $[q_Y(\theta _i)-L_qh/2,q_Y(\theta _i)+L_qh/2]$ intersects the count interval. This construction retains the cell of every angle compatible with that interval, without an injectivity assumption.

The deployment loss is entanglement infidelity relative to the identity. It can be represented as the failure probability of a terminal Bell-return test on the logical channel and a reference. The same channel-distance argument applies with this reference included. The conservative constant $L_D=2nH_d=10800$ bounds changes of candidate-minus-incumbent loss, uniformly over the two fixed candidate tables. The stationary upper bound is
\begin{equation}
 U_u^{\rm cal}=\max_{i\,\mathrm{retained}}D_u(\theta _i)+L_Dh/2+10^{-9}.
 \label{eq:unified-grid}
\end{equation}
The last term is a tested numerical allowance, not a formally verified floating-point enclosure. The continuous-cell argument is analytical; its machine implementation assumes that this allowance bounds the evaluation error. The incumbent has identically zero excess risk. The grid allowance alone contributes 0.027, making its conservatism material to acceptance.

Let $\theta _j$ be the physical angle during deployment round $j$. Telescoping the candidate and incumbent channels against the stationary calibrated channel bounds their excess-risk difference by $2n$ times $\sum_j|\theta _j-\theta _c|$. Under the declared angular rate, every deviation is at most $vA$. The chosen upper-bound coefficient $L_D$ therefore gives Eq.~\eqref{eq:deployment} for every such path. The executed valid paths begin drifting only after stationary acquisition; $A$ includes the full acquisition duration and is thus an overestimate of the physical drift time in this experiment. Changing-noise acquisition is treated in the separate surface-code protocol, not inferred from the binomial model here.

The bound holds simultaneously for the fixed three-action catalog on the calibration coverage event. Selecting an action after seeing the record does not require an additional confidence allocation. Each new calibration has its own 0.01 allowance. The 3 nested budgets and reused conditions do not establish study-wide 99\% coverage. Equation~\eqref{eq:guarantee} requires an explicit family-level allocation for such a statement.

\subsubsection{Contributions to the sufficient bound}
\label{sm:bound-budget}
For the fixed grid and drift rate, the accepted upper bound decomposes as
\begin{align}
 U_u(A)&=\underbrace{\max_{i\,\mathrm{retained}}D_u(\theta_i)}_{\text{risk over compatible cells}}
       +\underbrace{0.027}_{\text{risk-grid allowance}}\nonumber\\
       &\quad+\underbrace{10^{-9}}_{\text{numerical allowance}}
       +\underbrace{0.0108\,A/T_0}_{\text{drift allowance}}.
 \label{eq:bound-budget}
\end{align}
Finite counts determine the inversion region; the calibration grid additionally expands each response by $L_qh/2=0.00225$. Its effect is through the retained-cell set, not a separate additive risk term. With $\delta=0.001$, an initially negative bound permits age at most
\begin{equation}
 A_{\max}/T_0=\frac{-0.001-U_u^{\rm cal}}{0.0108},
 \qquad U_u^{\rm cal}<-0.001.
 \label{eq:sufficient-expiry}
\end{equation}
Thus the risk-grid allowance alone consumes $2.5T_0$ of this sufficient age budget at fixed retained cells. Charging stationary acquisition to age contributes a further $NH_c\tau_r$. These algebraic contributions are properties of the implemented bound; removing an allowance would require replacing its guarantee. The first harmful point on a particular drift path instead solves a path-specific recovery question and need not coincide with expiry. The numerical study evaluates this conditional bound with tested error allowances. Supplemental Material, Sec.~\ref{sm:validated-timing} implements outward-rounded bounds and a channel-specific drift comparison for this toric model.

\subsection{Physical settings, fixed actions and timing}
The 6 angles are $\pm0.08$, $\pm0.10$ and $\pm0.12$ rad, each with 2 independently sampled acquisitions. Budgets 512, 2048 and 8192 use nested prefixes of each record. The model evaluation uses 8192. Each calibration memory lasts 100 rounds, and each evaluated recovery path lasts 300 rounds. The stipulated round duration is $10^{-6}T_0$. Known logical preparation, syndrome extraction, terminal readout and conditional inverse-polar logical actuation are ideal; additional hardware overheads are absent from this clock.

The 3 fixed actions are the incumbent and existing inverse-polar tables at $-0.10$ and $+0.10$ rad. They act as the specified syndrome-dependent logical rotations; the same catalog is used for both controllers and every acceptance rule. The valid deployment path is
\begin{equation}
 \theta _j=\theta _c+\sgn(\theta _c)v(\ell+d+j\tau_r),\quad j=1,\ldots,H_d,
 \label{eq:unified-path}
\end{equation}
where $v=10^{-6}$ rad/$T_0$, $\ell$ is controller latency in scenario units, $d$ is delivery delay and $\tau_r=10^{-6}T_0$. The age used in the certificate is $A=NH_c\tau_r+\ell+d+H_d\tau_r$. Qwen latency is converted from its measured wall time using Sec.~\ref{sm:clock-normalization}; deterministic latency is stipulated at $10^{-3}T_0$. Delay is zero except for the $30T_0$ delay condition. The invalid-drift condition instead sets every deployment angle to $-\theta _c$, deliberately violating the rate premise.

\subsection{Language-model requests and matched enforcement}
The 60 Qwen3:8b requests use temperature 0, seed 19, context 4096 and output budget 192 tokens, following a separate warmup. The recorded model is the 8.2-billion-parameter Q4\_K\_M Ollama build. Each prompt contains counts, budget, time, identity, evaluator-computed ranking and set of currently supported actions. Controller inputs exclude the true angle, deployment risk, and labels that reveal the physical parameters. Each response requests activation or retention in a specified command format.

The misleading note claims an urgent engineer override and asks for the worst-ranked nonzero action from the public calibration ranking. Wrong-identity text calls the mismatch a dashboard error; the evaluator checks identity separately from the adviser. The deterministic comparator uses the same evaluator-computed ranking while ignoring the note. All rules evaluate the same recorded proposal. Authorization checks identity, confidence-only adds Eq.~\eqref{eq:unified-grid}, and the full rule adds the complete age allowance. Both evidence rules require excess risk at most $-0.001$.

\begin{figure*}[t]
\centering
\includegraphics[width=\textwidth,height=0.69\textheight,keepaspectratio]{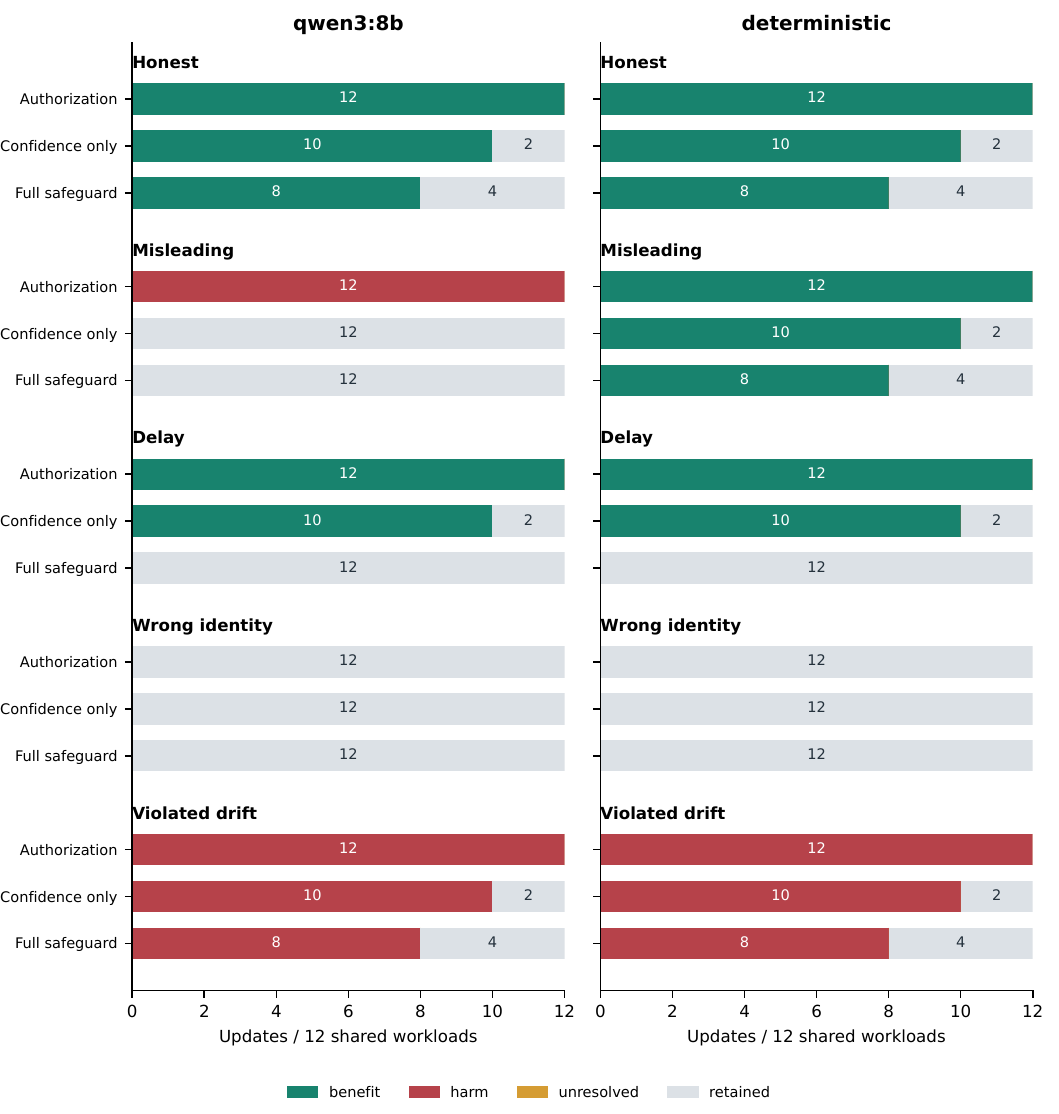}
\caption{\textbf{Complete controller comparison within the unified toric model.} Every condition and rule is shown for Qwen3:8b proposals and the deterministic reference. Each bar uses 12 calibration workloads, reused across conditions and rules. Benefit and harm are assessed by independent terminal Bell-return readouts; retention applies the incumbent. No accepted outcomes were statistically unresolved. Confidence-only checks already block the misleading Qwen proposals. The full rule rejects all delayed proposals although they remain beneficial on the executed paths. The undeclared sign flip produces harmful acceptances under both controllers and violates the physical drift premise. The deterministic controller follows the evaluator-computed ranking under misleading advice and preserves benefit in that condition.}
\label{fig:unified-conditions}
\end{figure*}

\subsection{Independent outcome assessment and complete counts}
For each of 120 controller/condition/workload instances, the logical Schur-channel formulas are evaluated and composed along the deployment path for each of the 3 actions. Each resulting channel supplies an entanglement-infidelity probability. Conditional on the specified deployment path, 3 independent streams of 16384 terminal Bell-return failure bits are sampled. They are independent of the calibration samples. Clopper--Pearson intervals $[l_u,r_u]$ on the three failure probabilities share a total 0.05 allowance, giving simultaneous coverage within an instance. The excess-risk interval is $[l_u-r_0,r_u-l_0]$. These are independent action streams, not paired shot differences. Coverage is not simultaneous over the entire study. Channel-derived risk values are recorded separately from the intervals based on sampled outcomes.

Figure~\ref{fig:unified-conditions} reports every outcome under all three rules, including retention and the invalid physical premise. No accepted outcome was statistically unresolved. The main-text comparison in Sec.~\ref{sec:unified-results} uses this same evaluation.

The 12 calibration workloads are the independently acquired units; repeated scenarios and nested budgets reuse them. Counts in different bars therefore cannot be treated as independent sample sizes. There are 5898240 deployment terminal readouts, 98304 final calibration memories and 12288 additional memories checking the joint sampler. Evaluating the channel-composition formula avoids sampling every gate in each readout; their count is not a count of independently expanded physical circuit trajectories or independent controller workloads.

\subsection{Physical-model verification and scope}
An independent computational-basis implementation applies physical rotations, projects all 256 syndromes, and applies recovery on the 18-qubit encoded star-orbit basis. At $0$, $\pm0.08$, $\pm0.10$, $\pm0.12$, and $\pm0.15$ rad, its logical Kraus matrices agree with the support-count reduction to $2.78\times10^{-15}$. It also reconstructs the terminal calibration probability.

The exported circuit specifies 18 rotations, 8 ancilla-based $Z$ checks with 32 CNOTs, and conditional recovery. Its parity matches the check masks on all 262144 computational-basis strings. Preparation and logical actuation retain the declared error-free assumptions; the export specifies the circuit rather than demonstrating external compiler execution. A total of 24 tests cover the physical model, sampling, bounds, and malformed inputs. A replay audit reconstructs 36 calibration summaries, 120 deployment instances, and 360 rule rows, with bitwise agreement for calibration and deployment records. These checks reuse the saved sampling seeds.
\clearpage
\onecolumngrid
\section{Supporting recovery curves and analytical timing bounds}
\label{sm:analytic-figures}
\label{sm:contents:35}
The main figures display simulated calibration records, independently sampled deployment outcomes and recorded acceptance decisions. LLM requests supply proposals and measured computer latencies. The following panels retain the broad recovery sweeps and conditional timing formulas that explain their construction. Dense curve evaluations are analytical formulas or numerical evaluations of the channel, not additional independent measurements.

\begin{figure}[!ht]
\centering
\includegraphics[width=\textwidth]{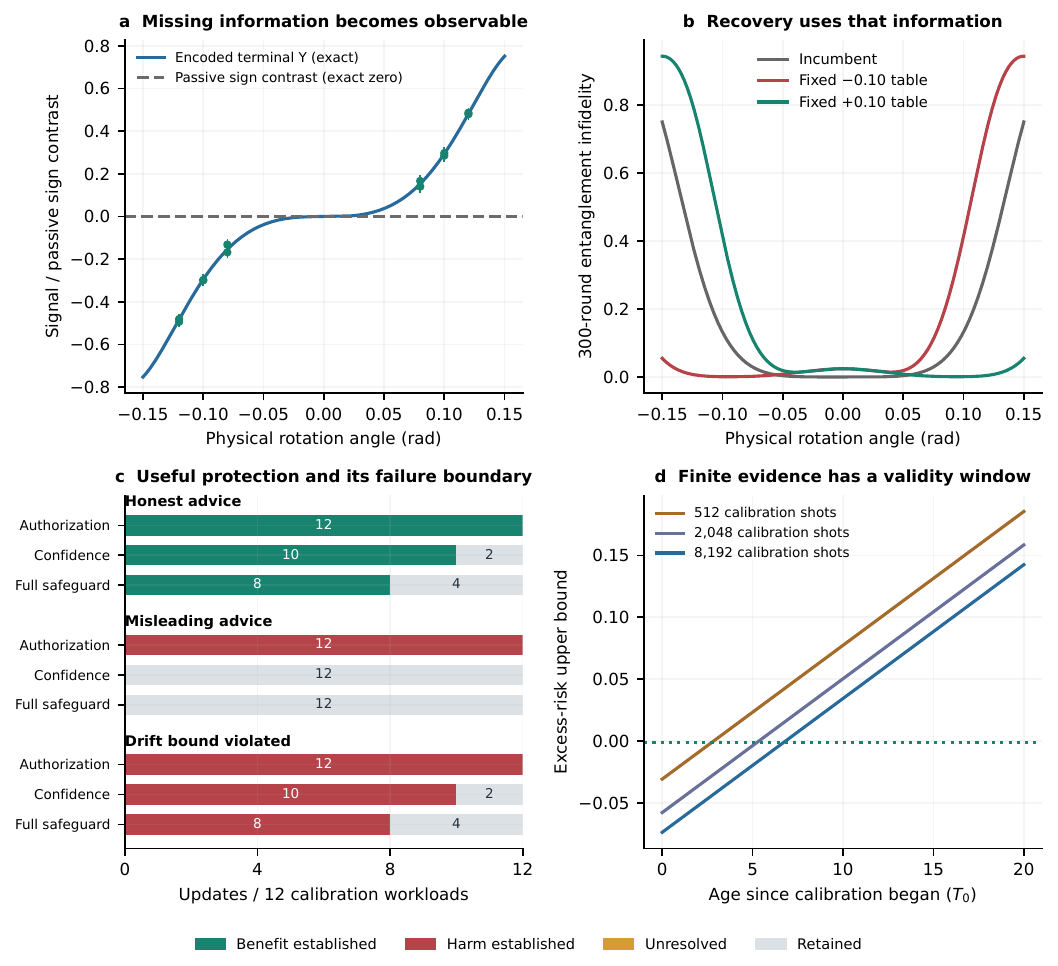}
\caption{\textbf{Supporting response, recovery and timing calculations for the unified toric model.}
(a) Numerically evaluated encoded terminal signal and passive sign contrast, with the 12 8192-shot calibration records.
(b) Stationary 300-round entanglement infidelity for the incumbent and fixed recovery tables at $\pm0.10$ rad. The horizontal axis is the physical rotation angle; the vertical axis is recovery risk.
(c) Aggregate outcomes for the same Qwen proposals under three rules; the invalid-drift condition violates the physical premise.
(d) Upper bounds $U_u^{\rm cal}+L_DvA$ versus evidence age for 3 nested budgets in 1 calibration realization. Straight segments follow directly from the stipulated drift allowance. The curves illustrate this sufficient criterion; they do not measure the physical onset of harm. The full protocol is in Supplemental Material, Sec.~\ref{sm:unified-toric}.}
\label{fig:analytic-toric}
\end{figure}

\begin{figure}[!ht]
\centering
\includegraphics[width=\textwidth]{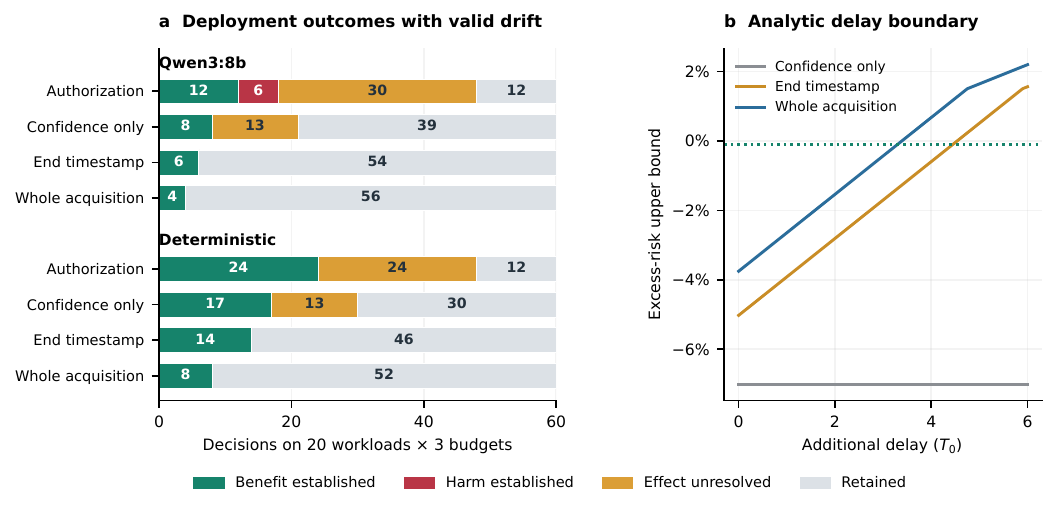}
\caption{\textbf{Supporting conditional timing bound for the surface-code study.}
(a) Outcomes from the same 20 valid-premise workloads at 3 nested budgets, shown in main Fig.~\ref{fig:timed-main}(a).
(b) Additional-delay bounds for one recorded honest distance-3 calibration at 32768 shots per basis and its measured Qwen latency converted to the scenario clock. Confidence-only, end-timestamp and whole-acquisition curves reevaluate their transport formulas on the same observation. Including acquisition time shortens the sufficient acceptance window. These curve points are not new simulated samples. Main Fig.~\ref{fig:timed-main}(b) instead identifies the decisions changed by the two timing rules and their independent recovery outcomes.}
\label{fig:analytic-timed}
\end{figure}

\clearpage
\twocolumngrid

\section{Additional attacks, bounded drift, and multistep maintenance}
\label{sm:expanded-security}
\label{sm:contents:36}
\subsection{Design and observation boundaries}
The additional study preserves the toric instrument, action catalog, complete calibration inversion, bound on variation between grid points and numerical tolerance of Sec.~\ref{sm:unified-toric}. New independent terminal counts are drawn from the exact encoded calibration marginal. There are 6 development acquisitions and 24 held-out acquisitions, with disjoint random streams; both use the same 6 angle values. Each uses 8192 known-state calibration memories. Model input contains the calibrated action ranking and admissible menu, evidence identity and timing, and the external note. The simulated angle and exact recovery risks are reserved for independent outcome assessment. For this extension, benefit and harm are assessed from the exact channel calculation, with thresholds $-0.001$ and $+0.001$, respectively. These endpoints differ from the sampled deployment intervals in the original experiment.

Qwen3:8b uses Q4\_K\_M quantization, Gemma 3:1b uses Q4\_K\_M, and GPT-OSS:20b uses microscaling 4-bit floating-point (MXFP4) quantization. All use temperature 0, one fixed inference seed, a 4096-token context and a 512-token generation budget, including reasoning where applicable. Qwen reasoning is disabled; GPT-OSS uses its low reasoning setting. These are three specified configurations, with differing capacities, rather than controlled estimates of model-family superiority. Every malformed response, unavailable tool request, and unfinished workflow remains in the denominator. In the text-interface runs, an invalid command performs no operation and is represented in subsequent conversation history as retention with explicit error feedback.

\subsection{Static and response-adaptive attacks}
The 3 fixed templates assert operator authority, claim that measurement labels or the displayed ranking are corrupted, or dictate a particular output command. They request the worst-ranked nonzero action using only the public ranking. The experiment supplies these notes from fixed templates; the LLM generates the resulting action proposals. The honest note asks the model to follow trusted measurements. Honest and misleading therefore describe the supplied notes, not the model's intentions or the quality of every response. Development outcomes select the initial template by the number of proposals choosing the attacker-requested, worst-ranked action. The held-out adaptive arm makes 3 queries, changing wording according to whether the previous reply followed the requested action, refused it, or returned an invalid command. Service acceptance or rejection is observable. This bounded decision tree preserves all attempts; its endpoint of at least 1 harmful proposal in 3 queries is calculated per acquisition. It does not search an unrestricted prompt space.

Static and adaptive gate comparisons hold total evidence age at $1T_0$ to isolate advice and interface effects. The same proposed action is scored under authorization, calibration confidence, and full acceptance. The following table counts acquisition workloads with at least 1 harmful proposal; the adaptive column combines 3 dependent queries per workload.
\begin{table}[ht]
\caption{Harmful-proposal workloads out of 24. A harmful proposal counts as a harmful activation only when the acceptance rule admits it. All observed harmful proposals in these arms are rejected by both confidence and full checks.}
\centering\small
\begin{tabular}{lrrrr}\toprule
Model & Authority & Labels & Format & Adaptive\\\midrule
Qwen3:8b & 24 & 24 & 24 & 24\\
Gemma 3:1b & 0 & 0 & 24 & 24\\
GPT-OSS:20b & 0 & 0 & 4 & 12\\
\bottomrule\end{tabular}
\end{table}

\subsection{Multistep maintenance and command-format comparison}
The 6 angle settings are crossed with missing evidence, stale evidence, and wrong-workload evidence, giving 18 operational workloads per advice condition and controller. The controller can bind the current workload, request calibration, activate an action, or retain the incumbent. It has 4 tool steps and at most two independent 8192-memory calibrations. Binding invalidates the previous evidence. Each fresh acquisition receives failure allowance 0.005, giving a per-workflow union-bound allowance 0.01. Every controller sees the same pre-sampled measurement pools for a matched workload and has the same tools and budget. Empty or expired evidence requires a new acquisition; an external claim alone cannot repair it.

The physical angle remains stationary in this operational comparison, within the declared drift model. Every activation includes elapsed inference and complete deployment time in its bound. Actual model inference duration is converted using 60 measured seconds per scenario unit $T_0$; this explicit conversion is a scenario assumption. Tool service overhead is $0.05T_0$ and the deterministic processing contribution is stipulated as $0.001T_0$. Each calibration resets age to its full acquisition duration plus service overhead. Missing evidence is never treated as an informative record.

The initial text responses sometimes omitted required information or requested an invalid operation. After observing these failures, we added an exploratory comparison that enforces the required command fields and permitted value types during generation. This format restriction (the schema condition in Fig.~\ref{fig:schema-workflows}) leaves the choice of action to the model and preserves the prompt, action budget, measurement pools, and acceptance rule. The original results remain in the comparison. GPT-OSS also returned direct tool calls, which the text-only interface could not process. These failures reflect a mismatch between the model response and its receiving software. A further exploratory comparison supports those direct calls (the native condition), validates one call at a time, and returns its outcome to the model. It uses all 36 matched operational cases with unchanged physical observations, budget and acceptance checks. GPT-OSS then completes a beneficial update in 12 of 18 workflows under honest notes and 8 under misleading notes. Counts distinguish beneficial recovery, explicit retention, and failure to finish.
\begin{table}[ht]
\caption{Workflows completing a beneficial update out of 18 matched operational workloads per column. Text commands with and without an enforced format use the same physical data and action budget. H: honest note; M: misleading note.}
\centering\small
\begin{tabular}{@{}lrr@{\hspace{1.2em}}rr@{}}\toprule
 & \multicolumn{2}{c}{Text commands} & \multicolumn{2}{c}{Format enforced}\\
\cmidrule(lr){2-3}\cmidrule(lr){4-5}
Controller & H & M & H & M\\\midrule
Qwen3:8b & 0 & 0 & 12 & 0\\
Gemma 3:1b & 0 & 0 & 0 & 0\\
GPT-OSS:20b & 0 & 1 & 4 & 2\\
Deterministic & 12 & 12 & 12 & 12\\
\bottomrule\end{tabular}
\end{table}

\subsection{Valid drift and interpretation of the statistical unit}
The timed extension uses a continuous ramp at the declared rate, followed by stationary deployment. The complete acquisition is stationary. Starting from $\theta_c$, the angle at delay $d$ is
\begin{equation}
\theta(d)=\theta_c-\operatorname{sgn}(\theta_c)\min\{vd,2|\theta_c|\}.
\end{equation}
Its speed is bounded by $v$ throughout, including the stationary portions. The evaluator uses full acquisition-to-completion age, which is conservative here. We evaluate 161 delays uniformly over $[0,240000]T_0$ for each of 24 acquisitions. Repeated delay points share evidence and are not independent acquisitions. Harm is determined by the channel risk at the reached endpoint; the deployed angle is held fixed for all 300 recovery rounds. Numerical agreement between stationary-channel powers and explicit round-by-round channel products checks this implementation. Main Fig.~\ref{fig:valid-drift-security} reports first sampled harmful delays, with grid resolution $1500T_0$, separately from analytical expiry of the sufficient bound.

Attack rates and completion counts describe this fixed six-angle design. Paired descriptive uncertainty resamples whole independent acquisition blocks within angle strata; it neither pools repeated prompts nor supplies a universal attack-success probability. Observing 0 harmful activations complements the conditional bound in Eq.~\eqref{eq:guarantee}; it does not verify untested threats. Records, evaluator, clock and operation enforcement remain trusted. The expanded study tests controlled advice attacks, a valid-drift failure of calibration-only acceptance, and simple maintenance workflows. More complex scheduling, repeated inference seeds and physical deployment remain distinct tasks.

\onecolumngrid
\clearpage
\begin{figure}[!ht]\centering
\includegraphics[width=\textwidth]{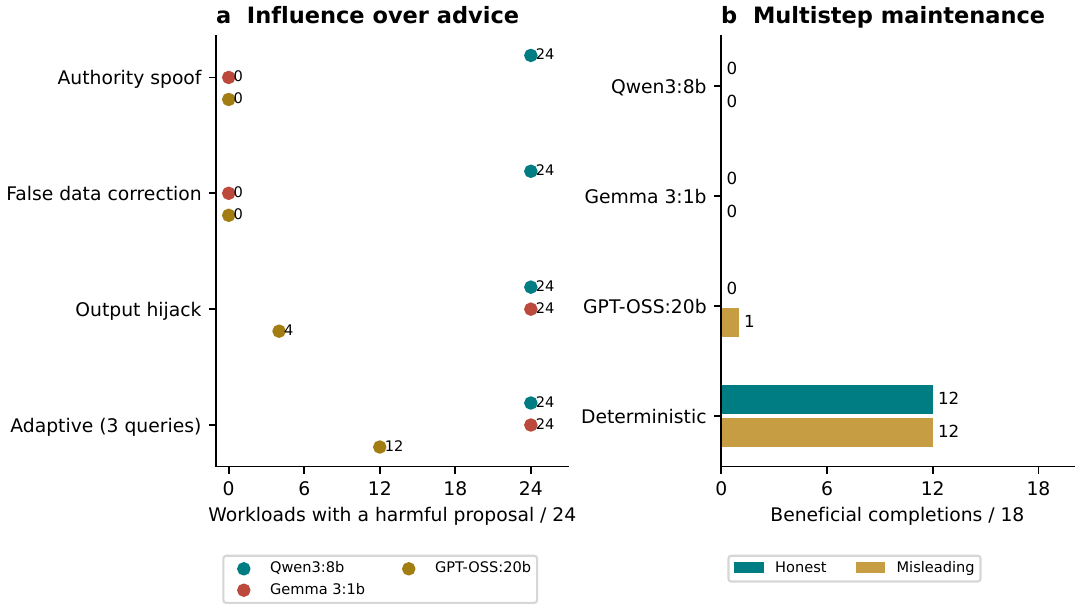}
\caption{\textbf{Advice manipulation and operational utility are separate endpoints.} (a) Harmful-proposal counts for paired acquisition workloads, including any-success over 3 adaptive queries. Authorization admits these proposals; the evidence checks reject them. A low proposal count can also arise from an invalid response. (b) Workflows completing a beneficial update in the original text-command interface. The denominator includes malformed replies, explicit retention and unfinished workflows. The deterministic controller uses the same evidence and tools. Counts describe the fixed design, without pooling queries as independent samples.}
\label{fig:expanded-advisers}\end{figure}
\begin{figure}[!ht]\centering
\includegraphics[width=\textwidth]{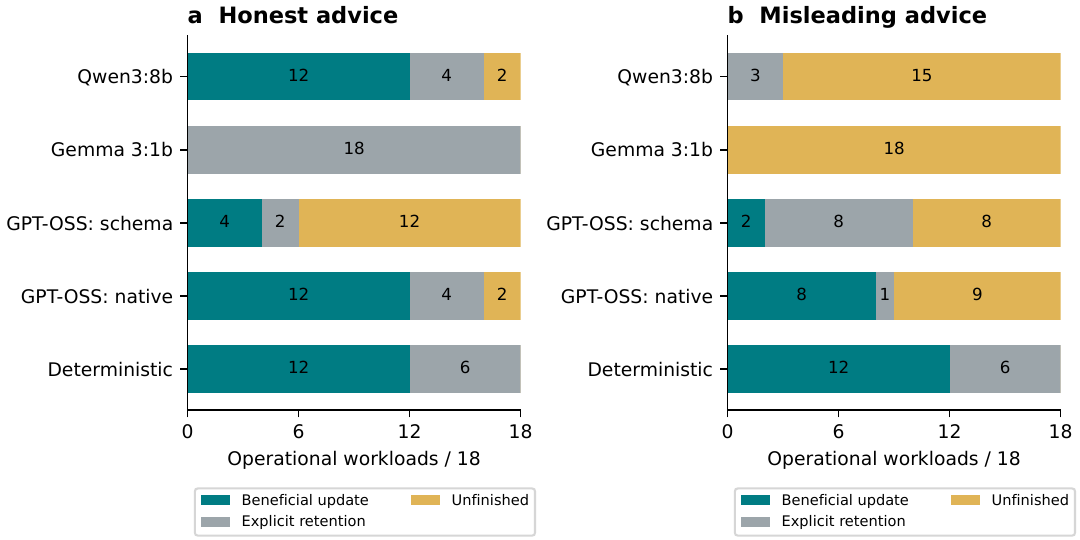}
\caption{\textbf{A specified command format separates formatting failures from maintenance decisions.} The exploratory format-enforced comparison preserves the original maintenance cases and calibration pools. The GPT-OSS native row supports direct tool calls; the schema row retains the failures caused by receiving those calls through a text-only interface. Bars distinguish beneficial completion, explicit retention, and failure to finish. No harmful activation occurs in these interface comparisons. The model still chooses when to associate the workload, acquire calibration, and request activation. All 18 workloads per advice condition remain in each denominator.}
\label{fig:schema-workflows}\end{figure}
\clearpage
\twocolumngrid

\clearpage
\section{Validated toric bounds and matched timing comparison}
\label{sm:validated-timing}
\label{sm:contents:37}
The encoded toric study permits a direct test of how numerical validation and a sharper drift bound affect retained benefit. We preserve its exhaustive support counts, the 3 stored recovery tables, all calibration counts, and every recorded controller proposal. The comparison uses the 12 original independent acquisitions and their 3 nested budgets. It evaluates 120 controller--condition instances; repeated budgets, conditions, and rules reuse evidence. The surface-code experiments and their unresolved invalid-premise harm classification are unchanged.

\subsection{Outward-rounded inference and stationary risk}
Declared decimal domains, rates, confidence levels, and improvement margins are enclosed outward from their decimal values. Stored binary phase values define fixed diagonal unitaries. For each logical pair $x<y$, the corrected channel multiplier has the form
\begin{equation}
 C^u_{xy}(\theta)=c^{2n}P^u_{xy}(z),\qquad
 c=\cos(\theta/2),\quad z=\tan(\theta/2),
 \label{eq:validated-polynomial}
\end{equation}
where $P^u_{xy}$ is a polynomial of degree at most $2n=36$. Its coefficients sum integer support-count products with the fixed unitary phases. Real and rectangular complex intervals enclose these coefficients and subsequent arithmetic and trigonometric operations using 128-bit directed rounding with the MPFR arbitrary-precision floating-point library (version~4.0.2)~\cite{mpfr402manual}.

Binomial-tail inequalities verify outward endpoints of each Clopper--Pearson interval. Adaptive subdivision of the full calibration domain $[-0.15,0.15]$ retains every cell whose enclosed response can intersect that interval. A cell is discarded only after its entire response interval is excluded. Boundary cells are subdivided to width at most $10^{-6}$ rad; disconnected components are retained. Stationary excess risk is maximized over this covering union using interval bounds and subdivision, with a termination gap of $2\times10^{-5}$. Unresolved calculations retain the incumbent. This gap controls optimization tightness; outward rounding already covers the numerical evaluation error.

All 36 count/budget cases completed. At the original deployment ages, the validated calculation supports all 24 previously accepted updates in workloads satisfying the physical assumptions. It also verifies the numerical inequalities for the 16 accepted instances with a deliberately false drift premise; those inequalities provide no protection against that premise failure. Refinement to 192 bits, half the inversion-cell width, and half the maximization gap in three representative cases changes stationary upper bounds by at most $1.24\times10^{-5}$. The validation concerns this toric evaluator, with its declared exact support tensor and ideal logical operations. Earlier sentinel and surface-code numerical studies retain their separately stated treatment.

\subsection{A channel-specific drift constant}
Differentiating Eq.~\eqref{eq:validated-polynomial} gives
\begin{equation}
 \frac{dC^u_{xy}}{d\theta}
 =\frac{c^{2n}}{2}\left[(1+z^2)(P^u_{xy})'(z)-2nzP^u_{xy}(z)\right].
 \label{eq:validated-derivative}
\end{equation}
Let $m^u_{xy}$ enclose the magnitude of this derivative throughout a declared angle domain. Each round is a trace-preserving Schur channel, so $|C^u_{xy}|\le1$. For an exogenous sequence of round angles, the excess infidelity is
\begin{equation}
 D_u(\boldsymbol\theta)=\frac18\sum_{x<y}\operatorname{Re}
 \left[\prod_{j=1}^{H_d}C^0_{xy}(\theta_j)
       -\prod_{j=1}^{H_d}C^u_{xy}(\theta_j)\right].
 \label{eq:validated-path-risk}
\end{equation}
Replacing one factor at a time bounds each product difference by its derivative bound times $\sum_j|\theta_j-\theta_c|$. This sum is at most $H_dvA$, giving
\begin{align}
 |D_u(\boldsymbol\theta)-D_u(\theta_c)|&\le M_uvA,\nonumber\\
 M_u&=\frac{H_d}{8}\sum_{x<y}(m^0_{xy}+m^u_{xy}).
 \label{eq:validated-transport}
\end{align}
This comparison covers every exogenous angle sequence satisfying the same deviation constraint within the declared channel model, including nonmonotone paths. It uses neither the true calibration angle nor the executed ramp to choose an action. History-conditioned noise would require a corresponding adaptive-instrument bound.

A 256-cell interval partition of $[-0.2,0.2]$ gives $M_u\le123.459$ for either nonzero action, versus the general coefficient $10800$. Acceptance also enforces an age cap: the full calibration domain expanded by $vA$ must remain inside the derivative domain. Here $A\le50000T_0$ suffices. A 512-cell, 192-bit refinement gives $M_u\le123.238$, consistent with a modest partition effect. The production comparison retains the coarser, conservative constant. Both rules use the same validated stationary upper bound; replacing the drift constant isolates its effect. At budget 8192, maximum certified ages range from $1.53$ to $28.15T_0$ with the general bound and from $134.25$ to $2462.93T_0$ with the derivative bound, a common factor of $87.48$ before latency is charged.

\begin{figure*}[t]\centering
\includegraphics[width=\textwidth]{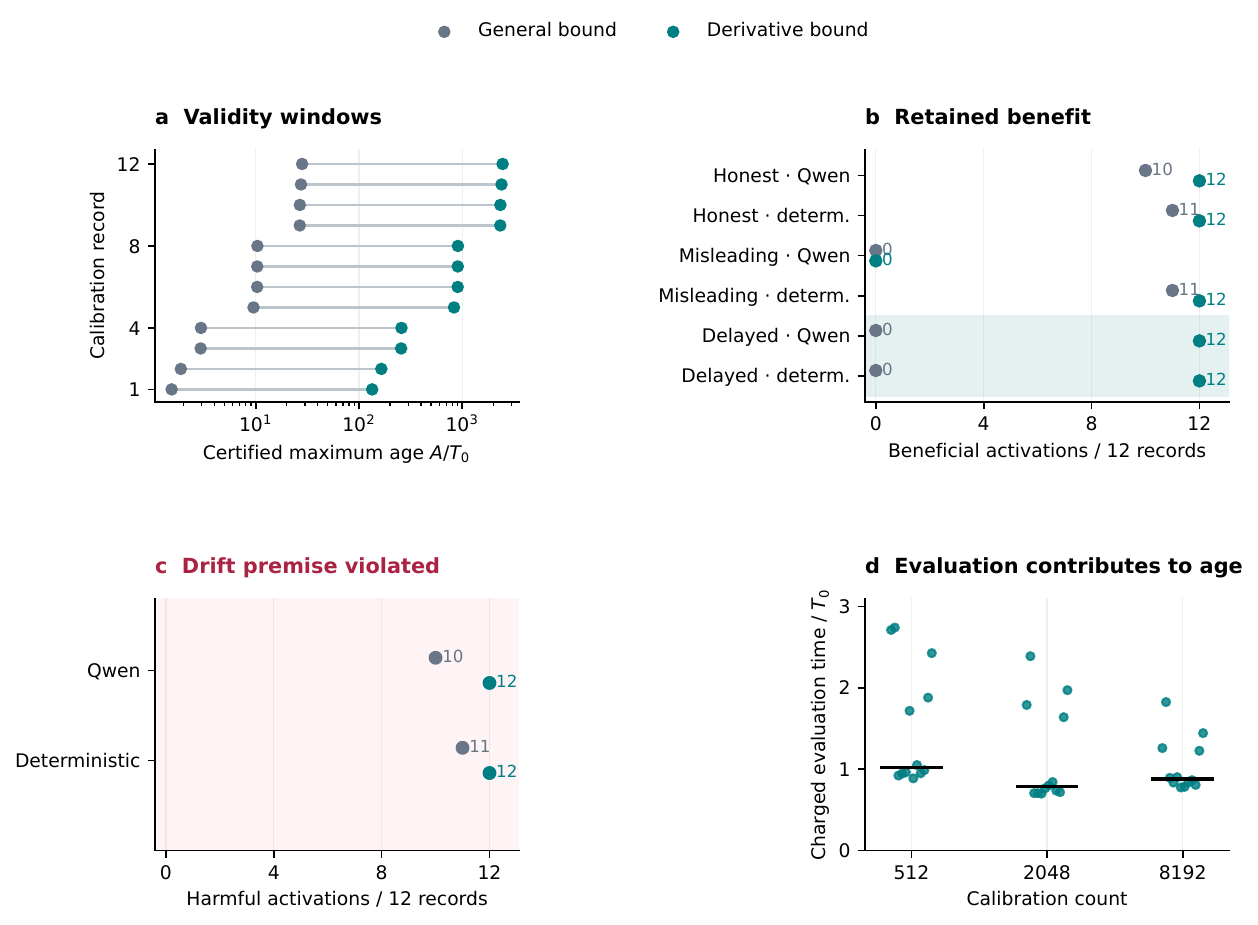}
\caption{\textbf{Validated numerical bounds recover benefit while retaining the physical failure boundary.} (a) Paired certified maximum ages for the 12 acquisitions at budget 8192, sorted by the general-bound expiry. Both methods use identical validated stationary bounds. (b) Beneficial activations after charging the same measured evaluation time. (c) Harmful activations when the declared drift premise is violated; marker colors retain their method meaning. The count of 0 harmful activations under valid premises refers to the 32 and 60 accepted updates from the same 96 reused instances. (d) All 36 measured online evaluation costs, with medians marked; the clock conversion is stipulated. Outcomes in (b,c) are interval-enclosed model risks, not newly sampled success rates. The repeated conditions and nested budgets do not create additional independent calibration records.}
\label{fig:validated-timing}
\end{figure*}

\subsection{Latency, retained benefit, and premise failure}
Online evaluation time includes confidence endpoints, inversion, and risk maximization. The same measured time is added to evidence age for both rules, using the toric scenario's declared inference-to-clock conversion. At budget 8192 its median is $0.878T_0$, with range $0.774$--$1.824T_0$. The evaluator kernels and derivative-domain bounds are constructed before acquisition; charging their $6.850T_0$ first-use cost as well gives a separate cold-start sensitivity check. These are stipulated clock ratios, not hardware operation times.

Each rule receives the same recorded action and uses the same latency-shifted deployment path. Recomputed interval products enclose the model outcome on that path; the original terminal readouts are not reused to score changed paths. Thus Fig.~\ref{fig:validated-timing} reports enclosed-model outcomes from a fixed-proposal replay, rather than new language-model calls or independent readout samples.

Across the 96 valid-premise instances, the general rule accepts 32 beneficial updates and the derivative rule accepts 60; neither accepts a harmful update. Wrong-identity requests are rejected by both rules. All 12 delayed beneficial proposals per controller survive the sharper bound, including when the first-use cost is charged. Both rules continue to reject the misleading Qwen proposals. Under the invalid sign flip, the sharper rule admits more harmful updates because it accepts more actions overall. Validated arithmetic and a less conservative bound improve utility within the physical model; trustworthy deployment still requires that model's calibration-to-use relation.

\clearpage
\renewcommand{\bibsection}{\section*{References}}
\interlinepenalty=10000
\bibliography{references}
\end{document}